\documentclass[twocolumn,showpacs,amsmath,amssymb,prx,nofootinbib]{revtex4}

\makeatletter
\DeclareFontFamily{OMX}{MnSymbolE}{}
\DeclareSymbolFont{MnLargeSymbols}{OMX}{MnSymbolE}{m}{n}
\SetSymbolFont{MnLargeSymbols}{bold}{OMX}{MnSymbolE}{b}{n}
\DeclareFontShape{OMX}{MnSymbolE}{m}{n}{
    <-6>  MnSymbolE5
   <6-7>  MnSymbolE6
   <7-8>  MnSymbolE7
   <8-9>  MnSymbolE8
   <9-10> MnSymbolE9
  <10-12> MnSymbolE10
  <12->   MnSymbolE12
}{}
\DeclareFontShape{OMX}{MnSymbolE}{b}{n}{
    <-6>  MnSymbolE-Bold5
   <6-7>  MnSymbolE-Bold6
   <7-8>  MnSymbolE-Bold7
   <8-9>  MnSymbolE-Bold8
   <9-10> MnSymbolE-Bold9
  <10-12> MnSymbolE-Bold10
  <12->   MnSymbolE-Bold12
}{}

\let\llangle\@undefined
\let\rrangle\@undefined
\DeclareMathDelimiter{\llangle}{\mathopen}%
                     {MnLargeSymbols}{'164}{MnLargeSymbols}{'164}
\DeclareMathDelimiter{\rrangle}{\mathclose}%
                     {MnLargeSymbols}{'171}{MnLargeSymbols}{'171}
\makeatother

\makeatletter
\def\l@subsubsection#1#2{}
\makeatother 

\usepackage{scrextend}
\usepackage{amsmath} 
\usepackage{amsfonts} 
\usepackage{amsthm}
\usepackage{dcolumn}
\usepackage{bm}
\usepackage{chemarrow}
\usepackage{graphicx}
\usepackage{epsfig}
\usepackage{mathrsfs}
\usepackage{hyperref}
\usepackage{yhmath}
\usepackage{color}
\usepackage{mathtools}
\usepackage{txfonts}
\usepackage{enumitem}
\usepackage{comment}
\usepackage[dvipsnames]{xcolor}
\usepackage[color,matrix,arrow,all,2cell]{xy}
\usepackage[T1]{fontenc}

\newcommand{\be}{\begin{equation}}
\newcommand{\ee}{\end{equation}}
\newcommand{\bea}{\begin{eqnarray}}
\newcommand{\eea}{\end{eqnarray}}
\newcommand{\bes}{\begin{subequations}\bea}
\newcommand{\ees}{\eea\end{subequations}}
\newcommand{\ba}{\begin{array}}
\newcommand{\ea}{\end{array}}

\newcommand{\tw}[1]{{\color{White} #1}}

\newcommand{\bs}[1] {\boldsymbol{#1}}

\newcommand{\C}{\mathscr{C}}
\newcommand{\E}{\mathscr{E}}
\newcommand{\G}{\mathscr{G}}
\newcommand{\I}{\mathscr{I}}
\newcommand{\T}{\mathscr{T}}
\newcommand{\ST}{\tau}

\renewcommand{\ij}{{i\!j}}
\newcommand{\ji}{{\!ji}}
\newcommand{\cyc}{c}
\newcommand{\cocyc}{c}

\renewcommand{\matrix}[1]{#1}

\newcommand{\Pp}{\matrix{P}}
\newcommand{\Past}{{\matrix{P}_{\!\st}}}
\newcommand{\U}{\matrix{U}}

\newcommand{\W}{\matrix{W}}
\newcommand{\M}{\matrix{M}}

\newcommand{\R}{\matrix{R}}

\newcommand{\Q}{\mathcal{Q}}
\newcommand{\A}{\mathcal{A}}
\newcommand{\F}{\mathcal{F}}

\newcommand{\tflu}{\Psi}
\newcommand{\flu}{\psi}
\newcommand{\tcur}{\Phi}
\newcommand{\cur}{\phi}
\newcommand{\ccur}{\phi}
\newcommand{\lcur}{\phi}

\newcommand{\phys}{\varphi}

\newcommand{\traj}{{\omega^t}}
\newcommand{\invtraj}{\overline{\omega}^t}
\newcommand{\minvtraj}{\widetilde{\omega}^t}
\newcommand{\prob}{P}

\newtheorem{theorem}{Proposition}
\newcommand{\Th}[1]{{Proposition \ref{#1}}}
\renewcommand{\th}{proposition}

\newcommand{\ot}{1\mbox{--}2}
\newcommand{\niton}{\not\owns}
\newcommand*{\Scale}[2][4]{\scalebox{#1}{$#2$}}%
\newcommand{\tr}{\mathrm{tr}\,}
\newcommand{\MJPG}{$\rm{MJPG}{}$}
\newcommand{\eq}{\mathrm{eq}}
\newcommand{\st}{\mathrm{st}}
\newcommand{\ssum}{ \Scale[1.2]{\sum}\,}
\newcommand{\x}{\bs{x}}
\newcommand{\rest}{\bs{z}} 
\newcommand{\RN}[1]{\textup{\footnotesize\uppercase\expandafter{\romannumeral#1}}}

\begin{document}

\title{Effective fluctuation and response theory}

	\author{Matteo Polettini} 
	\email{matteo.polettini@uni.lu}
	\affiliation{Physics and Materials Science Research Unit, University of Luxembourg,
	162a avenue de la Fa\"iencerie, L-1511 Luxembourg (Luxembourg)} 
	
	\author{Massimiliano Esposito}
	\affiliation{Physics and Materials Science Research Unit, University of Luxembourg,
	162a avenue de la Fa\"iencerie, L-1511 Luxembourg (Luxembourg)} 

	\date{\today}

\begin{abstract}

The response of thermodynamic systems slightly perturbed out of an equilibrium steady-state is described by two milestones of early nonequilibrium statistical mechanics: the reciprocal and the fluctuation-dissipation relations. At the turn of this century, the so-called fluctuation theorems extended the study of fluctuations far beyond equilibrium. All these results rely on the crucial assumption that the observer has complete information about the system: there is no hidden leakage to the environment, and every process is assigned its due thermodynamic cost. Such a precise control is difficult to attain, hence the following questions are compelling: Will an observer who has {\it marginal} information be able to perform an {\it effective}  thermodynamic analysis? Given that such observer will only be able to establish local equilibrium amidst the whirling of hidden degrees of freedom, by perturbing the {\it stalling} currents will he/she observe equilibrium-like fluctuations nevertheless? We address these two fundamental problems, providing a broad theory of the statistical behavior of some out of many currents that flow across a thermodynamic system. 

We model the dynamics of open systems as Markov jump processes on finite networks. Configuration-space currents count the net number of transitions between pairs of configurations; conjugate forces quantify their thermodynamic cost. Phenomenological currents are linear combinations of configuration currents, and only ensue when affinities enjoy appropriate symmetries, granting {\it thermodynamic consistency}. A complete thermodynamic description is achieved when the set of currents under consideration covers all cycles in the network, otherwise the set is marginal.

Within  this formalism, we establish that: 1) While marginal currents do not obey a full-fledged fluctuation relation, there exist effective affinities for which an integral fluctuation relation holds; 2) Under reasonable assumptions on the parametrization of the rates, effective and ``real'' affinities only differ by a constant; 3) At stalling, i.e. where the marginal currents vanish, a symmetrized fluctuation-dissipation relation holds while reciprocity does not; 4) There exists a notion of {\it marginal time-reversal} that plays a role akin to that played by time-reversal for complete systems, which restores the fluctuation relation and reciprocity; 5) The effective affinity is the putative affinity of an observer who only has marginal information about a system and formulates a minimal model accounting for his/her steady-state observations; 6) There exist fluctuation relations across different levels in the hierarchy of more and more ``complete'' theories. The above results hold for configuration-space currents, and for phenomenological currents provided that certain symmetries of the effective affinities are respected --- a condition that we call {\it marginal thermodynamic consistency}, which is stricter thermodynamic consistency and whose range of validity we deem the most interesting question left open to future inquiry.  Our results are constructive and operational: we provide an explicit expression for the effective affinities in terms of the steady-state reached by the system when all transitions supporting the marginal currents are turned off and propose a procedure to measure them in laboratory. 




\end{abstract} 


\maketitle

~

\newpage

~

\newpage

\tableofcontents

~

\newpage

\newpage

\section*{Notation and abbreviations}
\label{notation}

\noindent {\it \; \hspace{1em}  Acronyms}
\vspace{.2cm}
\begin{addmargin}[3em]{0em}
\begin{itemize}[noitemsep,nolistsep]
\item[FR] Fluctuation Relation
\item[IFR] Integral Fluctuation Relation
\item[MJPG] Markov Jump-Process Generator
\item[NESM] Non-Equilibrium Statistical Mechanics
\item[RR] Reciprocal Relations
\item[SCGF]  Scaled-Cumulant Generating Function
\item[SFDR] Symmetrized Fluctuation-Dissipation Relation
\item[TR] Time Reversed, Time Reversal
\item[p.d.f.] probability density function 
\end{itemize}
\vspace{.2cm}
\end{addmargin}
{\it \; \hspace{1em} Graphs}
\begin{addmargin}[3em]{0em}
\vspace{.2cm}
\begin{itemize}[noitemsep,nolistsep]
\item[$|\,\cdot\,|$]  Cardinality of a set or range of an index
\item[$\G$] Graph
\item[$\E$] Edge set of a graph
\item[$\I$] Site (vertex) set of a graph
\item[$\C$] Simple oriented cycle
\item[$\T$] Spanning tree
\item[$i,j,\ldots$] Sites
\item[$\ij,\ji,\ldots$] Oriented edges
\item[$\partial$] Incidence matrix 
\end{itemize}
\vspace{.2cm}
\end{addmargin}
{\it \; \hspace{1em} Linear algebra}
\begin{addmargin}[3em]{0em}
\begin{itemize}[noitemsep,nolistsep]
\item[]
\item[$\vec{v}$] Vector in $\R^{|\I|}$
\item[$\vec{1}$] Vector with all unit entries
\item[$\matrix{A}$] Matrix $\mathbb{R}^{|\I|} \to \mathbb{R}^{|\I|}$
\item[$\bs{v}$] All other vectors
\end{itemize}
\vspace{.2cm}
\end{addmargin}
{\it \; \hspace{1em} Observables}
\vspace{.2cm}
\begin{addmargin}[3em]{0em}
\begin{itemize}[noitemsep,nolistsep]
\item[$\;\cdot\;_{\ij}$] Edge observable
\item[$\;\cdot\;_{\alpha}$] Phenomenological observable 
\item[$\flu$,$\cur$] Mean flux, mean current
\item[$\tflu^t$, $\tcur^t$] Time-integrated stochastic flux, and current
\item[$\F_{\ij}$] Thermodynamic force of a transition
\item[$\A$] ``Real'' affinity
\item[$\Q$] Effective affinity
\item[$\sigma$] Mean entropy production rate
\item[$\Sigma^t$] Stochastic entropy production
\item[$\vec{p}$] Steady-state of master equation
\item[$\;\cdot\;^{\,\st}$,$\;\cdot\;^{\,\mathrm{eq}}$] At stalling, at equilibrium
\item[$w_{\ij}(\x)$] Parametrized rates of the master equation
\end{itemize}
\vspace{.2cm}
\end{addmargin}
{\it \; \hspace{1em} Stochastic tools}
\begin{addmargin}[3em]{0em}
\vspace{.2cm}
\begin{itemize}[noitemsep,nolistsep]
\item[$\traj$] Stochastic trajectory
\item[$\mathcal{D} \traj$] Path measure
\item[$\prob$] Path p.d.f. and its marginals
\item[$\langle\,\cdot\,\rangle$] Expected value w.r.t. $P(\traj) \mathcal{D}\traj$
\item[$\widetilde{\;\cdot\;}$] Hidden time-reversal
\item[$\lambda$/$\Lambda$] SCGF of edge/phenomenological currents
\item[$\zeta^t$] Cumulant generating function at time $t$
\item[$\M(\{q_\alpha\})$] Tilted operator
\end{itemize}
\end{addmargin}

~

\newpage

\section{Prologue}

Consider an experiment where the flows of certain quantities are measured, due to certain applied forces. All it takes to properly address the thermodynamics of such setup is to be able to draw a net demarcation between the open changing system\footnote{In textbooks of thermodynamics and systems theory a distinction is made between open, closed and isolated systems. In our perspective there is no substantial difference between flows of matter, energy, or information, for that matters. Hence closed systems are open; isolated systems are idealizations of open systems where external influences are extremely feeble.}, wherethrough physical quantities flow, and the decorrelated frozen environment wherefrom they come and go. Ideally for a proper thermodynamic analysis it is crucial to account for all of the currents flowing through the system, and to assign them their due thermodynamic cost.

However, things do not  quite work that way, neither in practice nor in theory. The system's boundary, drawn for example on criteria of time-scale separation, of spatial localization, and of coarse-graining of irrelevant degrees of freedom, might not be crystal-clear. As a consequence, the measurement apparatus might not resolve important sources of dissipation due to parasitic currents. Furthermore, unless a microscopic theory is available that explains what exact causes produce which precise consequences, currents might not be assigned their proper thermodynamic cost. In other words, the observer might only have {\it marginal} information about the setup, due to technological and theoretical limitations. Nevertheless, the vast majority of physicists develop a sort of thermodynamic craftsmanship\footnote{The role of craftsmanship in the preparation and interpretation of a scientific experiment has been discussed by sociologist of science H. Collins \cite{collins}. In particular, his analysis of the trimmings with Bayesian inference ({\it ivi}, Chapter 5) struck us as particularly relevant for the foundations of statistical mechanics.}, a learned sense for what is relevant in the controlled laboratory of their experiment (be it factual or thought). This leads them to identify {\it effective} thermodynamic forces that do not dispense with the laws of thermodynamics. In this process the logic of thermodynamic reasoning is preserved, but the nature of most of its measurable quantities needs to be renegotiated.

While according to Einstein thermodynamics is ``the only physical theory of universal content [\ldots] that will never be overthrown'' \cite{einsteinonth}, unlike other theories like Quantum Mechanics or General Relativity thermodynamics has for long been an intrinsically phenomenological science, a patchwork of profound laws and contingent principles that do not fit into a coherent mathematical framework. On these premises, it is common practice to invoke textbook thermodynamics in a literal way, deploying its jargon and formulas with little reference to the mental and physical processes through which such concepts were formulated. This is the fertile soil that foments the never-ending stream of pseudoscientific claims of violations the second law of thermodynamics, whose common root is a fundamental misunderstanding of which marginal currents and effective forces are actually at play.

\begin{figure}
\centering
\includegraphics[width=\columnwidth]{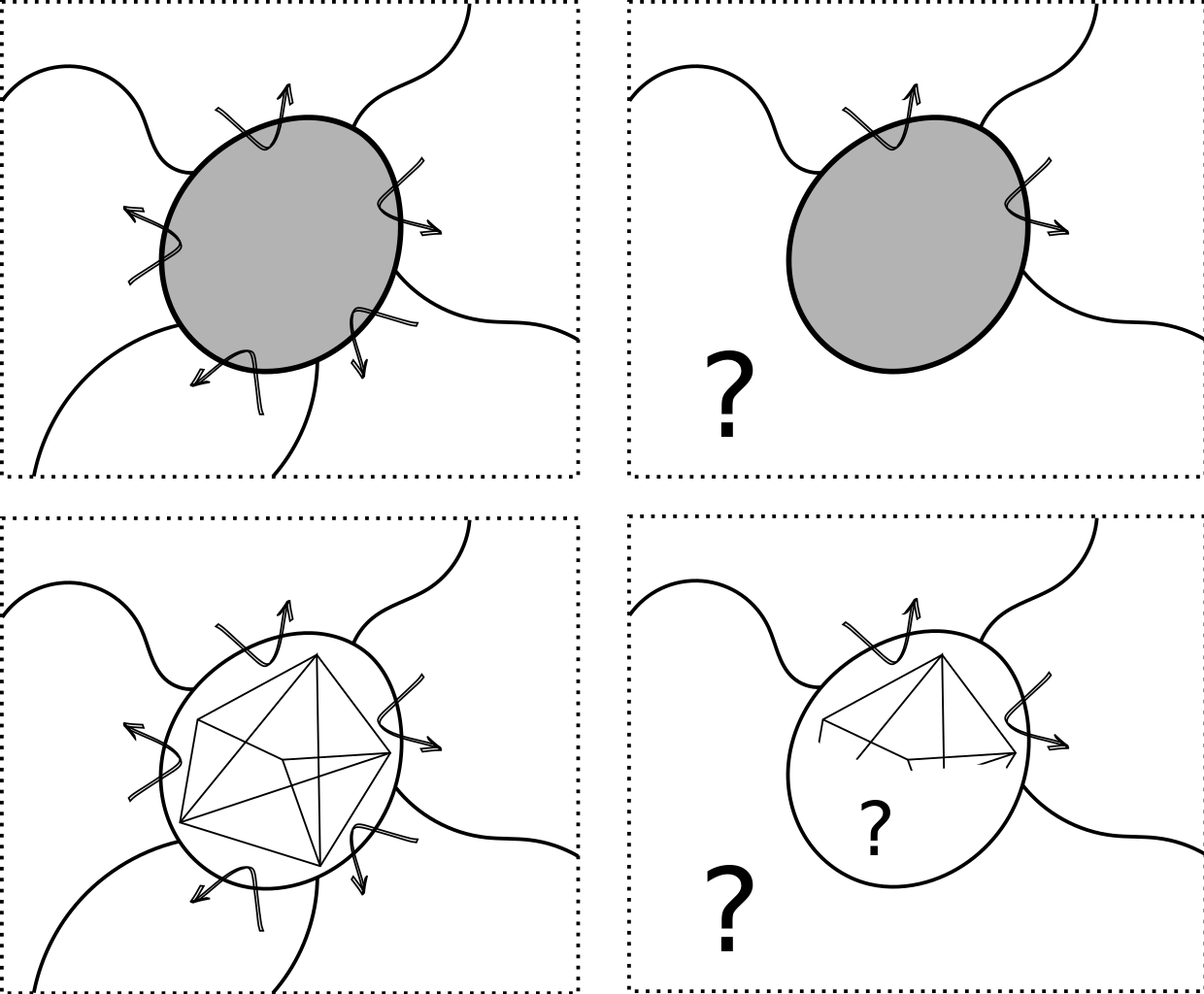}
\caption{First row: on the left, pictorial representation of a system through which currents flow; on the right, the standpoint of a marginal observer that only measures certain currents. In the second, the system's state space that mediates the passage of currents is resolved: It is a discrete network.}
\label{fig:illustration}
\end{figure}

Today a rigorous framework for the logical deduction of thermodynamic instances far from equilibrium is available \cite{ldbseifert} and is in the course of experimental validation \cite{experiments}, and it claims to become {\it the} way thermodynamics is thought of and taught. Modern nonequilibrium statistical mechanics is based on the assumption that the system's configurations are explored by a Markovian dynamics, with transition probabilities biased according to thermodynamic incentives coming from the environmental reservoirs. This framework allows to characterize fluctuations of observables, and therefore it applies to small systems, not necessarily in the so-called ``thermodynamic limit,'' and in principle it applies arbitrarily far from equilibrium as long as the Markov assumption remains valid.

At the heart of equilibrium statistical mechanics lies the identification of {\it static} physical properties of a system (e.g. temperature, pressure etc.) with the average behavior of microscopic degrees of freedom that fluctuate (e.g. average kinetic energy, velocity etc.). The first step out of equilibrium consists of slight perturbations of such observables, whose response can be characterized in terms of their spontaneous fluctuations at equilibrium, according to the so-called fluctuation-dissipation relation (FDR) and of the reciprocal relations (RR) that take the names names of some of the heroes of 20th century physics \cite{nyquist,green,onsager,kubo,miller}. For nonequilibrium systems, the picture is varied in a {\it dynamical} way: here the observables of interest quantify motility and directionality within a system. The connection between physical and statistical laws is encoded in the fluctuation relation (FR) -- whose precise formulation is embodied in a plethora of so-called Fluctuation Theorems \cite{bochkov,kurchan,maes,lebowitz}. The FR states that the rate at which a system delivers entropy to the environment is a measure of the arrow of time, viz. of the asymmetry between the probability of microscopic paths and their time-reversed. To use a metaphor, the probability of ``getting the toothpaste back into the tube'' is exponentially suppressed with respect to that of ``getting the toothpaste out of the tube,'' using Woody Allen's characterization of irreversibly in {\it Whatever works} \cite{allen}.

Most often the paste spreads out according to the second law of thermodynamics (2nd), an inequality that in this setup easily follows from a more general identity, the integral fluctuation relation (IFR). All such relations can be resumed in the following implication diagram:
\begin{align}
\begin{array}{c}\xymatrix{
\text{{\it nonequilibrium}} & \mathrm{FR}  \ar@{=>}[r]   \ar@{=>}[d] & \mathrm{IFR} \ar@{=>}[d]  \ar@{=>}[r]  & \mathrm{2nd} \\
\text{{\it near equilibrium}} & \mathrm{RR} & \mathrm{S-FDR}
} 
\end{array},
\nonumber \end{align}
where by S-FDR we intend a symmetrized version of the usual Green-Kubo relation. The role of the First Law and other conservation laws is more subtle: it can be seen as a requirement on the form of the rates, which allows to identify the abstract Markovian jumps with physical currents. We will not consider the other laws of thermodynamics.

Crucially, establishing the above scheme requires that the observer has complete information about the system's currents and forces. The question is then open as about how many of these results still apply to marginal observables of experimental interest, and what effective adjustments need  eventually to be made. In particular, if the Markov process ventures into some sector of the configuration space that is hidden to the observer, how should we quantify the thermodynamic incentives over there?

The purpose of this paper is to present a general theory of the thermodynamics of a marginal set of currents and of the effective forces that drive them, under the assumption that somewhere in the belly of these coarser observables there lurk fundamental currents and forces that abide by the principles of Markovian stochastic thermodynamics. We first show that only the right-hand side of the above implication diagram stands:
\begin{align}
\begin{array}{c}\xymatrix{
\text{{\it nonequilibrium}} & \tw{\mathrm{FR}} \ar@[White]@{=>}[r]   \ar@[White]@{=>}[d] & \mathrm{IFR} \ar@{=>}[d]  \ar@{=>}[r]  & \mathrm{2nd} \\
\text{{\it near stalling}} & \tw{\mathrm{RR}} & \mathrm{S-FDR}
}.
\end{array}
\nonumber \end{align}
Notice that in the marginal theory the analog of an equilibrium state is a {\it stalling} state in which the marginal currents vanish, while the hidden currents might still be flowing, as if the observer was in the eye of a hurricane. Importantly, the effective forces can be determined operationally by a simple tuning procedure,  thus opening the way to experimental implementations of our theory.

Furthermore, from a more mathematical perspective, the left part of this diagram can be reinstated upon an appropriate redefinition of the underlying dynamics of the Markovian walker. Under this new ``hidden time-reversed'' dynamics, denoted by a wiggle, we obtain the inference scheme
\begin{align}
\begin{array}{c}\xymatrix{
\text{{\it nonequilibrium}} & \widetilde{\mathrm{FR}} \ar@{=>}[r] \ar@{=>}[d] & \mathrm{IFR} \ar@{=>}[d]  \ar@{=>}[r]  & \mathrm{2nd} \\
\text{{\it near stalling}} & \widetilde{\mathrm{RR}} & \mathrm{S-FDR}}
\end{array}.
\nonumber \end{align}
Finally, as the observer adds more and more currents and forces to his accounting, a ``hierarchy'' of marginal theories is explored: a cross-hierarchical FR holds, that is the ultimate core result of the whole construction.

\begin{center}
\includegraphics[width=.2\columnwidth]{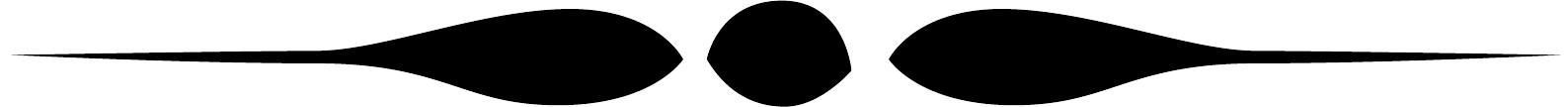}
\end{center}

The validity of FRs under coarse-graining and of FDRs far from equilibrium, in particular at stalling, are two questions that have been frequently addressed. Let us attempt a overview --- itself marginal.

Gallavotti \cite{gallavotti} produced a convincing argument for why it is necessary to address FRs of local observables: As most systems of thermodynamic interest are large in volume, global observables are subjected to two extensive limits --- one with respect to size and one with respect to time --- hence rare events are even rarer. A special license is then needed to focus on localized non-extensive observables. In the formalism of chaotic dynamical systems, Gallavotti heuristically defined a local entropy production rate associated to a microscopic region of space that satisfies a FR.

The validity of FRs for coarse-grained observables has been considered in Refs.\,\cite{garcia} and \cite{wachtler}, where the IFR is studied when the observer has incomplete information. In particular, in the latter work measurement errors are introduced via a kernel that smoothens the sharp values of the ``real''  degrees of freedom into a distribution of coarser observed values. In one specific model it is found that the IFR can be preserved given a notion of effective work. However, differing from our setup, this quantity is not stochastic. The coarse-graining of the statistics of the currents for biochemical systems has been considered in Ref.\,\cite{altanerPRL}. The partial fluctuation theorem in systems weakly coupled to the environment has been studied in Ref.\,\cite{gupta}, where it is argued that a violation of the FR can persist even in the limit of vanishing interaction. Uhl et al. \cite{uhl} have considered the fluctuations of an apparent entropy production in bipartite systems, finding many cases where an effective affinity restores the FR. For chemical networks where only some molecular species can be monitored experimentally, Bravi and Sollich \cite{bravi} derived systematic models for subsystem dynamics that can help with the inference problem of estimating properties of the environment from observed sub-network dynamics. Another situation where the observer does not have access to all of the thermodynamic currents are stochastic models of so-called ``Maxwell demons'', systems composed of an engine and a memory that operates a feedback control on the engine. To the total dissipation rate contribute fluxes of energy and of information, and an observer that does not duly keep into account the demon observes controversial behavior \cite{strasberg,mandal,parrondo}. The FR in such models was investigated in Refs.\,\cite{frenzel}, where the problem was solved by defining suitable observables that reinstate the FR, but which differ in nature from currents. Similar in spirit are the FRs for conditional and marginal probabilities discussed in Ref.\,\cite{crooks}, where appropriate terms are added to the entropy production rate in bipartite systems where one degree of freedom is neglected, and the hidden Markov models considered in Ref.\,\cite{bechhoefer}. All these approaches differ from ours as we assess properties of marginal observables without resorting to {\it ad hoc} redefinitions of the stochastic observable under consideration (the currents).

Because of hidden heat flows, system-bath correlations in either classical or quantum systems, if not taken into proper account, might lead to violations of the laws of thermodynamics \cite{partovi}. The authors of Ref.\,\cite{bera} comment that ``in order to re-establish the laws of thermodynamics, one not only has to look at the local marginal systems, but also [at] the correlations between them'', and this can be achieved by some effective description. Effective thermodynamic potentials also play a role in systems strongly coupled to their surroundings \cite{jarzynski}.

Another procedure that naturally leads to questions about marginal currents is the separation of fast vs. slow degrees of freedom. An effective affinity has been proposed to analyze experiments where a slow degree of freedom has been observed while the fast ones were integrated away \cite{mehl}. The effect of time-scale separation in thermodynamics has been studied in Ref.\,\cite{espositocg} and recently in Refs.\,\cite{bo,wang}. The former highlights that effective dynamics only preserve certain thermodynamic properties if internal detailed balance is obeyed, a feat that will play some role in our analysis of phenomenological currents. The latter show that in general the blanket is too narrow, and either dynamics or thermodynamics need to be sacrificed: while in their case it is thermodynamics, in ours it will be dynamics --- viz. we are {\it not} presenting a theory of an effective dynamics in the observable configuration space.  

Stalling steady states have been considered before by Qian \cite{qian2}, who dubbed the effective affinity that we will later introduce ``isometric force'', and they play a role in the analysis of molecular motors \cite{kolomeisky}. Stalling currents also appear to play an important role in efficiency optimization, as e.g. in so-called B\"uttiker probes \cite{buttiker,dubi,brandner}. An effective two-terminal thermoelectric nanomachine, obtained starting from a more complete three-terminal machine with one stalling current, has been considered in Ref.\, \cite{yamamoto} to study the effect of asymmetric Onsager coefficients on efficiency.

Response far from equilibrium is a broadly studied subject. In general, it is well-understood that the FDR has to be modified by including the correlation of the current with a quantity that is symmetric under time-reversal, alongside with the current's self-correlation. This can give rise to interesting behavior such as negative differential mobility, i.e. the fact that one can ``get less by pushing more,'' as is well illustrated in the driven lattice Lorentz gas described in Ref.\,\cite{leitmann}. Experimental verifications of modified FDRs are also available \cite{solano}. We will only briefly address the response of systems arbitrarily far from equilibrium, and mostly focus on response at stalling. Far from equilibrium, the notion of an effective temperature has been investigated in weakly ergodic ageing systems \cite{cugliandolo,crisanti}.

Part of the material covered in this manuscript has been anticipated by the Authors in Ref.\,\cite{polettiniobs} for the case where the currents count a single transition in configuration space. Response out of stalling was presented to some extent in Ref.\,\cite{altaner16}, which was stimulated by the specific analysis found in Refs.\,\cite{lau,lacoste}. One of the Authors considered FRs for a marginal current in the case of electron transport in a double quantum dot in Ref.\,\cite{bulnes}. A construction analogous to ours that allows to prove IFRs for appropriate functionals was advanced by Shiraishi and other authors \cite{shiraishi,rosinberg,hartich}. A comparison of our proposal and Shiraishi's was provided in Ref.\,\cite{gili}.

\begin{center}
\includegraphics[width=.2\columnwidth]{ornament2.pdf}
\end{center}

Outlying the theory in full requires to deploy a broad spectrum of techniques, ranging from Markov processes, to algebraic graph theory \cite{biggs}, to the theory of large deviations of stochastic processes \cite{touchetterep}, etc. We can only introduce them in a very compact form in Sec.\,\ref{setup} and give numerous references. We take the chance to shortly review in some detail some of the mathematical techniques that support the logical development of the theory, casting them in our own language. For some of these results, we provide novel derivations. A disclaimer about the mathematics:  All our new results will be framed as ``Propositions'' in order pinpoint the logical structure of the discourse. Propositions are statements that, to the best of our understanding, are outlined and argued in a sufficiently self-consistent way, but which might fall short in meeting the standards that mathematicians intend. In particular, we make no distinction between Theorems, Lemmas, Corollaries, Remarks etc. We give no complete statement of the assumptions for each proposition because of an objective lack of expertize in the more advanced issues of probability theory and Markov processes. Nevertheless we trust the overall coherency of our argumentation, and we encourage improvement on rigor. The ornament on p.\,\pageref{victorian} marks the point where most of the results are either new, or they are reinterpreted in a novel way.  

Inspired by the pedagogical principle by Albert V. Baez in his unconventional physics textbook Ref.\,\cite{spiralling}, we use a spiralling approach to the presentation of the material. The title, the abstract and this prologue represent the first three spirals of five more and more in-depth variations on the theme. The paper is structured as a Greek tragedy, with the main material exposed in a technical way in the {\it episodes}, preceded by the {\it prologue} that the reader is just reading, and most importantly by the {\it parode}, the first song sung by the chorus, which anticipates the main themes in a self-contained manner. Throughout the play the chorus stays on stage as a constant interlocutor, so the reader should always keep in mind the voice of the parode, which is structured into a {\it stroph\^e} and a {\it antistroph\^e} with the same meter, where we present the older material and the newer one, in parallel ways. The parode ends with an {\it epode} on future perspectives related to our results, while more technical conclusions are drawn in the closing {\it exode}. The preceding {\it stasimon}, the final song sung by the choir, is in a diminished locrian mode. The reason why this story should a tragedy, rather than a comedy, is not clear to the authors.

\section{Parode: Enunciation of the main results}

Before dwelling into our theory in full detail, in this Parode we present a less technical, yet self-contained discussion of the main results, which can be considered as an independent letter on its own. This ``entrance ode'' is meant to provide the reader knowledgeable in the field with enough details to reproduce most of the results on his own, and the neophyte with an overview on the main lines of reasoning. Quoting (with minor adjustments.) A. V. Baez \cite{spiralling}, \guillemotleft the reader may go through this section rapidly, as it takes him through a round of the spiral; the treatment may strike him as light, even inadequate. But he/she should rest assured, however, that we are laying a good foundation for a more concise and mathematical treatment in the following chapters \guillemotright.

We will first introduce the known facts regarding a ``complete'' set of currents to make contact with established knowledge and lingo in the field. We then introduce our new results about marginal sets of currents, paralleling them to the older results. Finally we explain some of the main technical ingredients underlying our results, and draw conclusions.

\subsection{Stroph\^e: ``Complete'' fluctuations and response}
\label{strophe}

Macroscopic thermodynamics describes systems through which a certain number\footnote{\label{cardinality}Symbol $|\,\cdot\,|$ denotes both the cardinality of sets and the range of indexes. Wherever possible, we will omit to specify the range of indexes.} $|\alpha|$ of (steady) {\it currents}  $\phi_\alpha$ flow, powered by conjugate thermodynamic forces or {\it affinities} $\A_\alpha$. The system can be seen as an interface between several reservoirs, with the currents flowing through the system, to and from reservoirs. We take currents and affinities to be a ``complete'' set of core, irreducible observables, assuming that all conservation laws, e.g. of energy (First Law of thermodynamics), number of particles, etc. have already been taken care by gauging out certain reference reservoirs. We shall explain what it exactly means to be ``complete'' later in this Parode. Then the affinities usually are {\it differences} of inverse temperatures, chemical potentials, etc., and  currents are ``conserved on their own''. For this reason we drew them as in-out ``reservoir arrows'' in the illustration Fig.\,\ref{fig:illustration}. The macroscopic entropy production rate (EPR) is defined as the bilinear form
\begin{align}
\sigma := \sum_\alpha \phi_\alpha \A_\alpha.
\end{align}

When we go microscopic, because of thermal noise we need to allow for fluctuations. We thus make currents into  random variables. We consider a single realization, or {\it path} or {\it trajectory} $\traj$ of a hypothetical experiment in a time window $[0,t]$. The stochastic time-integrated currents $\tcur^t_\alpha := \tcur_\alpha[\traj]$ are functionals of such trajectory, typically time-extensive. Therefore
\begin{align}
\cur_\alpha := \lim_{t \to \infty}\frac{\left\langle \tcur^t_\alpha \right\rangle}{t}
\end{align}
converges and yields the mean steady currents, where the average is taken with respect to a probability measure over trajectories $\prob(\traj) \mathcal{D}\traj$, that we will describe in detail in  \S\,\ref{subsec:measure}.

Along a single realization of the process, we define the {\it entropy production} as 
\begin{align}
\ssum  \tcur^t_\alpha  \A_\alpha + O(1). \label{eq:ep}
\end{align}
Here $O(1)$ stands for contributions that do not add-up in time, which are due to the transient adjustment of the system's internal entropy. With a slight stretch of imagination, the entropy production can be considered as the amount of entropy delivered to the environment during the process; however, to be slightly pedantic, we must emphasize that in our approach there is no such thing like a state function ``entropy of the environment''.

The entropy production is the major actor in the so-called fluctuation relation (FR)  \cite{maes,kurchan,lebowitz,andrieux,faggionato,polettiniFT} 
\begin{align}
\frac{\prob (\{ \tcur_\alpha\})}{\prob (\{ - \tcur_\alpha\})} =  \exp \ssum \tcur_\alpha \A_\alpha  \label{eq:fr}
\end{align}
where $\prob(\{ \tcur_\alpha\})$ is the probability density function (p.d.f.) that the time-integrated currents $\tcur^t_\alpha$ take values in a neighborhood of $\tcur_\alpha$. In the rest of the paper we will adopt a strategy discussed in Ref.\,\cite{polettiniFT} by which we can deal with finite-time FRs on the same footing as with asymptotic ones, whereby the above relation is exact equality at all times, provided the initial configuration from which the trajectory departs is selected with an appropriate probability distribution. In all those cases where results only hold in the long-time limit, we will use the asymptotic equality $\asymp$. The reader not interested in these subtleties might just view all the FRs as asymptotic at $t\to+\infty$.

An immediate corollary of the FR is the integral fluctuation relation (IFR) \cite{jarzy} 
\begin{align}
\left\langle  \exp {- \ssum \tcur_\alpha \A_\alpha} \right\rangle = 1. \label{eq:ifr}
\end{align}
The IFR embodies and refines the Second Law of thermodynamics, which states that on average the EPR is non-negative,
\begin{align}
\sigma \geq 0,
\end{align}
an immediate consequence of Eq.\,(\ref{eq:ifr}), via Jensen's inequality for convex functions.

A system is said to be {\it detailed balanced} when all of the affinities vanish; in this case the steady state is an equilibrium, that is, all mean steady currents vanish:
\begin{align}
\mathrm{equilibrium:}\quad \A_\alpha = 0 , \forall \alpha \iff \phi_\alpha = 0 , \forall \alpha.
\end{align}

It is well known that the FR actually gives rise to a cornucopia of IFRs \cite{maesseminaire}. We notice here in passing, as a new result, that in the case where only two processes contribute to the total entropy production, $|\alpha| =2$, from the FR also follows
\begin{align}
\left\langle \exp {- \, \tcur_1 \A_1} \right\rangle = \left\langle \exp {- \, \tcur_2 \A_2} \right\rangle.  \label{eq:recIFR}
\end{align}
We dub this the {\it reciprocal IFR}. The interesting feature of this relation is that it resolves and relates the statistics of two currents that can be quite different in physical nature.

The FR allows to derive all known results of response theory close to equilibrium. To attack this problem, we introduce an explicit dependency of the structural properties of the system (viz. the transition rates of the underlying stochastic dynamics) on certain parameters $\bs{x} = \{x_\kappa\}_{\kappa = 1}^{|\kappa|}$, $|\kappa| \geq |\alpha|$, with the following requirement:

\label{A0}
\begin{itemize}
\item[A0]  The first $|\alpha|$ of these parameters are {\it thermodynamic}, meaning that there exist constants $x_\alpha^{\mathrm{eq}}$ such that
\begin{align}
\A_\alpha (\bs{x})  = x_\alpha - x_\alpha^{\mathrm{eq}}.
\end{align}
All other parameters $\{x_\kappa\}_{\kappa > |\alpha|}$, on which the affinities do not depend, are {\it kinetic}. 
\end{itemize}
Notice that while this is just a contrived way to say that we either perturb the affinities, or some property that does not alter the affinity, this subtlety will play an important role below. At $\bs{x} = \bs{x}^{\eq}$ the system satisfies detailed balance whereby all of the forces $\A_\alpha(\bs{x}^\eq) = 0$ vanish and so do the currents, $\cur_\alpha^\eq := \cur_\alpha(\bs{x}^\eq) = 0$, where the superscript ``$\,{}^{\eq}\,$'' means ``evaluated at $\bs{x} = \bs{x}^{\eq}$''.

For systems that are slightly perturbed out of equilibrium, two major results hold: the (symmetrized) fluctuation-dissipation relation (S-FDR) and the reciprocal relations (RR) \cite{onsager,miller}. Defining the response coefficients as
\begin{align}
\lcur_{\alpha; \kappa} := \frac{\partial  \cur_{\alpha}}{\partial x_{\kappa}},
\end{align}
these two near-equilibrium relations state respectively that the response to a variation of a thermodynamic parameter at $\bs{x}=\bs{x}^\eq$ satisfies
\begin{subequations}
\begin{align}
\lcur^{\mathrm{eq}}_{\alpha; \alpha'} + \lcur^{\mathrm{eq}}_{\alpha'; \alpha} & = \ccur^{\mathrm{eq}}_{\alpha\alpha'}  \label{eq:greenkubo} \\
\lcur^{\mathrm{eq}}_{\alpha; \alpha'} - \lcur^{\mathrm{eq}}_{\alpha'; \alpha} & = 0 \label{eq:onsager}
\end{align}
\end{subequations}
where
\begin{align}
\ccur_{\alpha\alpha'} & := \lim_{t\to \infty} \frac{1}{t} \left\langle \left(\tcur^t_\alpha - \big\langle \tcur^t_\alpha \big\rangle \right)  \left( \tcur^t_{\alpha'} - \big\langle \tcur^t_{\alpha'} \big\rangle \right) \right\rangle
\end{align}
is the steady-state variance of the currents, properly scaled with time. The S-FDR and the RR can be proven quite straightforwardly as corollaries of the FR. More importantly for this paper, the first result follows as a corollary of the IFR Eq.\,(\ref{eq:ifr}), and the second as a corollary of the reciprocal IFR Eq.\,(\ref{eq:recIFR}). At equilibrium, currents do not respond to a variation of the kinetic parameters:
\begin{align}
\lcur^{\mathrm{eq}}_{\alpha; \kappa} = 0 , \qquad \kappa > |\alpha|. \label{eq:activeres}
\end{align}
Furthermore, the FR can be employed to produce higher-order response relations \cite{andrieux2,andrieux3}, which constitute the most promising testing ground for our theory. For example, at third order, focusing on one single current, near equilibrium one obtains 
\begin{subequations}\label{eq:thirdordereq}
\begin{align}
\ccur^{\mathrm{eq}}_{\alpha\alpha\alpha} & = 0 \\ \ccur_{\alpha\alpha; \alpha}^{\mathrm{eq}} - \lcur_{\alpha; \alpha\alpha}^{\mathrm{eq}} & = 0, 
\end{align}
\end{subequations}
where $\ccur_{\alpha\alpha\alpha}$ is the scaled third-order cumulant. The first relation is due to the fact that at equilibrium the current p.d.f. is symmetric, $P^{\eq}(\{\tcur_\alpha\}) = P^{\eq}(\{- \tcur_\alpha\})$, and the second expresses the second-order response of the average current in terms of the first-order response of its variance.

\begin{center}
\includegraphics[width=.2\columnwidth]{ornament2.pdf}
\end{center}

Let us now sketch the mathematical framework and assumptions based on which the above results can be derived (full details in Sec.\,\ref{setup}). We consider a continuous-time, discrete configuration-space Markov ``jump'' process occurring on a finite network with configurations ({\it sites} in graph-theoretic language) $i,j,\ldots$ connected by transitions ({\it oriented edges}) $i\!j,\ldots$, where the $i$ is the final site and $j$ is the starting site, following the right-to-left physicists' convention. The dynamics can be described by an evolution equation for the probability $p_i(t)$ of being in configuration $i$ at time $t$, governed by the master equation
\begin{align}
\frac{d}{dt}\vec{p}(t)  = \W \vec{p}(t). \label{eq:ME} 
\end{align}
Here, $\vec{p} = (p_1(t))_i$ and $\W$ is a Markov-jump process generator (MJPG) with entries
\begin{align}
\W_{\ij} = \left\{\ba{ll} w_{\ij}, & i \neq j \\ - w_i & i = j. \ea\right. ,
\end{align}
where $w_i = \sum_l w_{ji}$ is the exit rate out of a configuration. We call $\W$ the {\it forward} generator. The dependence on the external parameters is encoded in the rates $w_{\ij}= w_{\ij}(\bs{x})$. Given the steady-state of the dynamics $\vec{p}$, satisfying $\W \vec{p} = 0$, one can construct the time-reversed generator $\overline{\W} = \Pp \W^T \Pp^{-1}$, or simply hidden time reversal (TR), where $\Pp = \mathrm{diag}\, (p_i)_i$. In a sense, time reversal ``runs steady-states back in time,'' with detailed-balanced (equilibrium) systems obeying time-reversal symmetry
\begin{align}
\W(\bs{x}^{\eq}) = \overline{\W}(\bs{x}^{\eq}). \label{eq:dbop}
\end{align}

Time-integrated {\it edge currents} $\tcur^t_{i\!j}$ count the net number of times a certain transition is performed along a single realization of the process;  all possible current-like observables $\tcur^t_\alpha$ are linear combinations of edge currents. It is usually more practical to study the currents' statistics via their cumulants, properly scaled in time. An important result in the theory of Markov processes allows to obtain the currents' scaled cumulant generating function (SCGF)\footnote{We actually adopt the sign convention of Touchette \cite{touchetterep} rather than that of Lebowitz and Spohn \cite{lebowitz} on the definition of the tilted operator, and thus on the SCGF -- which in our case is actually a {\it signed} SCGF.}  $\lambda(\{q_\alpha\})$ as the dominant eigenvalue of a suitably defined ``tilted'' operator, see Sec.\,\ref{subsec:scgf}. Then the FR Eq.\,(\ref{eq:fr}) translates into the following fluctuation symmetry \cite{lebowitz} 
\begin{align}
\lambda(\{q_\alpha\}) = \lambda(\{\A_\alpha - q_\alpha\}), \label{eq:lebspo}
\end{align}
while the the IFR Eq.\,(\ref{eq:ifr}) reads
\begin{align}
\lambda(\{\A_\alpha\}) =0
\end{align}
and the reciprocal IFR Eq.\,(\ref{eq:recIFR}) reads
\begin{align}
\lambda(\{\A_1,0\}) = \lambda(\{0,\A_2\}).
\end{align}

All of the above results hold in the following ``complete'' setups:
\begin{itemize}
\item[A1] Index $\alpha$ spans through independent {\it cycles} in the network. The affinities are computed as the sum of the log-ratio of the rates $\log w_{\ij}/w_{\ji}$ along such cycles, and their conjugate currents $\tcur^t_\alpha$ are edge currents associated to certain preferred transitions in the network, in the light of Schnakenberg's cycle analysis of steady states, which is the analog of Kirchhoff's mesh analysis of electrical circuits applied to Markov processes \cite{schnak,andrieux,polettini2}.
\item[A2] Index $\alpha$ ranges through a smaller number of {\it phenomenological} currents, which in realistic physical models are associated to several transitions in the configuration network of a system. Provided such transitions cover at least a basis of cycles, correspondingly, for the FR to hold, the cycle affinities must enjoy certain symmetries, a condition called {\it local detailed balance} or {\it thermodynamic consistency} -- or, simply consistency, systematically analyzed in Refs.\,\cite{bridging,rao}.
\end{itemize}
For example, if measured independently, the set of currents denoted by double arrows in the following network form a minimal ``complete'' set according to setup A1,
\begin{align}
\mathrm{A1)} \quad \Scale[2]{\ba{c}\xymatrix{
\bullet \ar@{<<-}[r]^{\Scale[.5]{1}} \ar@{-}[d] & \ar@{-}[dl]\ar@{-}[d]  \bullet  & \bullet  \ar@{->>}[l]_{\Scale[.5]{2}}   \ar@{-}[d]   \\
\bullet \ar@{<<-}[r]_{\Scale[.5]{3}} &  \bullet  &  \bullet \ar@{-}[ul]  \ar@{->>}[l]^{\Scale[.5]{4}}  
}\ea}. \nonumber
\end{align}
As detailed in Sec.\,\ref{sec:cycle}, the corresponding cycle affinities are calculated as the log-ratio of the product of the rates on a basis of fundamental cycles. We can represent them diagrammatically as
\newcommand{\scale}{.7}
\begin{equation}
\label{eq:affgraph}
\begin{aligned}
\A_1 & = \log \frac{\Scale[\scale]{\xymatrix{
\bullet \ar@{<-}[r] \ar@{->}[d] & \ar@{<-}[dl]  \bullet   & \ar@[Gray]@{-}[l] \ar@[Gray]@{-}[d]   \\
\bullet &  \ar@[Gray]@{-}[u] \ar@[Gray]@{-}[l] & \ar@[Gray]@{-}[ul] \ar@[Gray]@{-}[l]
}}}
{\Scale[\scale]{\xymatrix{
\bullet \ar@{->}[r] \ar@{<-}[d] & \ar@{->}[dl]  \bullet   & \ar@[Gray]@{-}[l] \ar@[Gray]@{-}[d]   \\
\bullet &  \ar@[Gray]@{-}[u] \ar@[Gray]@{-}[l] & \ar@[Gray]@{-}[ul] \ar@[Gray]@{-}[l]
}}},
& \A_2 & = \log \frac{\Scale[\scale]{\xymatrix{  \ar@[Gray]@{-}[r] & \bullet  & \bullet  \ar@{->}[l]   \ar@{<-}[d]  \\  \ar@[Gray]@{-}[ur] \ar@[Gray]@{-}[u] \ar@[Gray]@{-}[r] &  \ar@[Gray]@{-}[u]  \ar@[Gray]@{-}[r] &  \bullet \ar@{<-}[ul] 
}}}{\Scale[\scale]{\xymatrix{  \ar@[Gray]@{-}[r] & \bullet  & \bullet  \ar@{<-}[l]   \ar@{->}[d]  \\  \ar@[Gray]@{-}[ur] \ar@[Gray]@{-}[u] \ar@[Gray]@{-}[r] &  \ar@[Gray]@{-}[u]  \ar@[Gray]@{-}[r] &  \bullet \ar@{->}[ul] 
}}}, \\
\A_3 & = \log \frac{\Scale[\scale]{\xymatrix{\ar@[Gray]@{-}[d] 
& \ar@{<-}[dl]\ar@{->}[d]  \bullet \ar@[Gray]@{-}[l] & \ar@[Gray]@{-}[l] \ar@[Gray]@{-}[d]   \\
\bullet \ar@{<-}[r] &  \bullet  & \ar@[Gray]@{-}[ul] \ar@[Gray]@{-}[l]
}}}
{\Scale[\scale]{\xymatrix{\ar@[Gray]@{-}[d] 
& \ar@{->}[dl]\ar@{<-}[d]  \bullet \ar@[Gray]@{-}[l] & \ar@[Gray]@{-}[l] \ar@[Gray]@{-}[d]   \\
\bullet \ar@{->}[r] &  \bullet  & \ar@[Gray]@{-}[ul] \ar@[Gray]@{-}[l]
}}}, & 
\A_4 & = \log \frac{\Scale[\scale]{\xymatrix{ \ar@[Gray]@{-}[r]  & \ar@[Gray]@{-}[r]  \ar@{<-}[d]  \bullet  & \ar@[Gray]@{-}[d]  \\
  \ar@[Gray]@{-}[ur] \ar@[Gray]@{-}[u] \ar@[Gray]@{-}[r]  & \bullet  &  \bullet \ar@{<-}[ul]  \ar@{->}[l]  
}}}{\Scale[\scale]{\xymatrix{ \ar@[Gray]@{-}[r]  & \ar@[Gray]@{-}[r]  \ar@{->}[d]  \bullet  & \ar@[Gray]@{-}[d]  \\
  \ar@[Gray]@{-}[ur] \ar@[Gray]@{-}[u] \ar@[Gray]@{-}[r]  & \bullet  &  \bullet \ar@{->}[ul]  \ar@{<-}[l]  
}}},
\end{aligned}
\end{equation}
where the arrows in the diagrams imply multiplication of the corresponding rates.

Instead, if the observer is only capable of measuring a linear combination of the above currents, then we fall within setup A2. This might be the case when several transitions in configuration space correspond to the exchange of the same physical quantity with one particular reservoir, a situation that we illustrate by a curly ``reservoir arrow,'' borrowing from the chemistry literature. For example, the system
\begin{align}
\mathrm{A2)} \quad  \Scale[2]{\ba{c}\xymatrix{
\bullet \ar@{<<-}^{\hspace{.2cm} \substack{\rotatebox[origin=c]{180}{$\curvearrowright$} \\  \vspace{-.4cm} }}[r] \ar@{-}[d] & \ar@{-}[dl]\ar@{-}[d]  \bullet  \ar@{<<-}^{\hspace{.2cm} \substack{\rotatebox[origin=c]{180}{$\curvearrowright$} \\  \vspace{-.4cm} }}[r]    & \bullet \ar@{-}[d]   \\
\bullet \ar@{<<-}^{\hspace{.2cm} \substack{\rotatebox[origin=c]{180}{$\curvearrowright$} \\  \vspace{-.4cm} }}[r] &  \bullet \ar@{<<-}^{\hspace{.2cm} \substack{\rotatebox[origin=c]{180}{$\curvearrowright$} \\  \vspace{-.4cm} }}[r]   &  \bullet \ar@{-}[ul]  
}\ea} \nonumber
\end{align}
corresponds to the situation where each of the four transitions contributes one unit to the phenomenological current counter, but the observer would not be able to tell which one of the transitions happened. Then, to grant consistency for this particular example the affinities along the cycles depicted above must all take the same value (otherwise, by a calorimetric experiment the observer {\it would} be able to tell the difference!).
 
Let us give an insight on the physical interpretation of transition rates and on their thermodynamically consistent parametrization assumed in A0. For detailed balanced systems subject to conservative forces, the most general form that the rate of hopping from site $j$ to site $i$ can take is
\begin{align}
w^{\mathrm{eq}}_{\ij} = v_{\ij} \, \exp - u_j
\end{align}
where $v_{\ij} = v_{\ji} > 0$ is symmetric. From a physical standpoint, in view of e.g. the Arrhenius law \cite{hanggi},  one can portray the configuration space of a system as a landscape with the sharpest minima at the configuration sites, separated by activation barriers. The configuration function $u_i$ and the symmetric term $v_{\ij}$ fully describe such an {\it internal landscape}. For systems that do not satisfy detailed balance an asymmetric term $a_{\ij} = - a_{\ji}$ appears and we can generally write \cite{rao}
\begin{align}
w_{\ij} = w^{\mathrm{eq}}_{\ij} \, \exp a_{\ij}/2.
\end{align}
The intuition is that the non-conservative term $a_{\ij}$ is a relic of the interaction of the system with the degrees of freedom of an external reservoir that influences the transition. For example, in the procedure of obtaining an open irreversible chemical network from a closed one by chemostatting chemical species described in Ref.\,\cite{polettiniCN}, the internal landscape is fully encoded in the reaction rates, while the terms $a_{\ij}$ correspond to the concentrations of the external chemostats. Importantly, thermodynamic affinities only depend on the latter: the transformation of one particular affinity $d \A_{\alpha}$, at fixed values of all other affinities, only involves: A0.i) a {\it local} transformation of the external ``relic'' terms along the network's edges that are peculiar to that particular mechanism;  A0.ii) a {\it global} transformation of the internal energy landscape.  We will detail this issue in Sec.\,\ref{parametrizations}.

We will call a transformation of the form
\begin{align}
w_{\ij} \to w'_{\ij} = w_{\ij} \,e^{-a_j} \label{eq:gaugetrans}
\end{align}
a {\it gauge} transformation. Under such a transformation, the log-ratio of the rates $\log w_{\ij}/w_{\ji}$ transforms like an inhomogeneous gauge connection, but the affinities are invariant. In a sense, they are the Wilson loops of the theory. This nomenclature is, in fact, more than an analogy. Gauge invariance of nonequilibrium thermodynamics is a concept put forward by one of the authors in Refs.\,\cite{polettinigauge,dice}. There, it is argued that it is a necessary property if one wants to make sense of thermodynamics as a science of information and ignorance \cite{bennaim}, as the corresponding continuous symmetry corresponds to a modification of prior probabilities. Therefore gauge invariance allows to deal with biases encoded in prior information, often perceived as a threat to the ``objectivity'' of the theory.

\subsection{Antistroph\^e: Marginal fluctuations and response}
\label{antistrophe}

We now focus on a subset of $|\mu|< |\alpha|$ currents $\{\tcur^t_\mu\}$ and consider their marginal p.d.f.
\begin{align}
\prob(\{\tcur_\mu\}) := \int \prod_{\alpha> |\mu|} d \tcur_\alpha \, \prob(\{\tcur_\alpha\})  .
\end{align}
The questions we address are: which of the above relations survive, what new results emerge, and under which (presumably stricter) conditions?

The central result in this paper (\Th{th:ift1}, \Th{th:ift2}) is that there exist {\it effective affinities} $\Q_\mu$ such that a {\it marginal IFR} holds 
\begin{align}
\left\langle  \exp {- \sum \tcur_\mu \Q_\mu} \right\rangle = 1, \label{eq:MIFR}
\end{align}
despite the fact that the full-fledged FR does not,
\begin{align}
\frac{\prob(\{\tcur_\mu\})}{\prob(\{-\tcur_\mu\})} \neq \exp {\sum \Q_\mu \tcur_\mu}, \label{eq:noFR}
\end{align}
and, provided there is at least one additional unobserved current, neither does the reciprocal IFR,
\begin{align}
\left\langle  \exp {- \tcur_1 \Q_1} \right\rangle \neq \left\langle  \exp {- \tcur_2 \Q_2} \right\rangle,
\end{align}

Here and below $\neq$ loosely means ``generally not,'' keeping into consideration that one can always fabricate systems whose marginal currents do obey the marginal FR (e.g. systems with statistically independent currents because of a special topology of the network). That the FR does not generally hold can already be deduced by the analysis of specific examples, see e.g. Ref.\,\cite{lacoste}. Again, our marginal IFR holds asymptotically, for any given initial ensemble, or at all times provided the trajectory's initial configuration is sampled from a special state $\vec{p}^{\,\mathrm{st}}$ that we will describe shortly.

However, a moment of reflection leads to the conclusion that, {\it per se}, the existence of values of the $\{\Q_\mu\}$ that make Eq.\,(\ref{eq:MIFR}) true should be no surprise. If we are allowed to tune such values at will, the average of the exponential can definitely range anywhere from $0$ to $+\infty$.  As we will discuss later in this Parode, for $|\mu|>1$ there actually is a continuum of candidate effective affinities fulfilling the marginal IFR. Thus, what is important is not that there exist such values, but that they can be given an operational interpretation\footnote{Interestingly, the same emphasis on this operational aspect is found in the already mentioned textbook by Baez: ``An understanding of concepts requires, however, much more than the ability to recite the associated words and their dictionary definitions. It is necessary to study, and preferably to experience, the {\it operations} that give meaning to the words.''}. We reserve the expression ``effective affinities'' and the notation $\{\Q_\mu\}$ to one particular choice of those values, identified by a constructive procedure that we will soon detail, and that most importantly grants that they are {\it marginally thermodynamic} in the sense that
\begin{align}
\frac{\partial} {\partial x_{\mu'}} \Q_\mu(\bs{x}) = \delta_{\mu,\mu'} = \frac{\partial} {\partial x_{\mu'}} \A_\mu(\bs{x}).\label{eq:bmu}
\end{align}
This is crucial if we want to produce a response theory. However, this latter fact requires to slightly reduce the scope of assumption A0:
\begin{itemize}
\item[B0] Thermodynamic parameters only affect the rates of the networks' edges that support the current of observational interest.
\end{itemize}
We will investigate at length the difference between A0 and B0 in Sec.\,\ref{parametrizations}. Let us already give a piece of good news, in the light of the physical parametrization of the transition rates discussed in the previous section. The only difference with respect to the parametrization of the ``real'' affinities of the ``complete'' theory is that modifications of the internal landscape might affect effective affinities. Therefore we can only afford A0.i) the same {\it local} transformation of the external antisymmetric terms along the network's edges that are peculiar to that particular mechanism. Instead, we need to replace A0.ii with B0.ii) a {\it local} transformation of the internal energy landscape. This is not a dramatic restriction. As a matter of fact, the workings of \Th{th:mardiss1} basically show that there is not much more to ``thermodynamic parametrization'' than there is in ``local parametrization,'' so that this whole discussion can be safely dismissed: the whole point of this discourse is to show that the parametrization does not really affect the theory, unless one plays devil's advocate by picking a very nonlocal and contrived parametrization. As far as we only modify reservoir properties (e.g. temperatures, chemical potentials), we are on safe grounds.

From the marginal IFR follows that the marginal EPR $\ssum \Q_\mu \cur_\mu$ is positive, while notice that in general the ``piece'' of EPR $\ssum \A_\mu  \cur_\mu$ might be not, due to the phenomenon of transduction by which some currents can be made to run against their conjugate thermodynamic forces by a conjure of the other currents and forces \cite{hill}. More interestingly, we prove in \Th{th:hierarchy} that the marginal EPR is always smaller than the ``complete'' EPR
\begin{align}
0 \leq \ssum \Q_\mu \cur_\mu \leq \ssum \A_\alpha \cur_\alpha,
\end{align}
This generalizes the results of Ref.\,\cite{gili}, which deals with the case $|\mu| =1$. These considerations open up the question in what sense $\ssum_\mu  \Q_\mu  \cur_\mu$ can actually be interpreted as EPR from a marginal point of view. As shown in Ref.\,\cite{polettiniobs}, and recapitulated in \S \ref{sec:marepr}, in the case of a single current supported on one edge, indeed this quantity represents the putative EPR evaluated by a local observer that can only access information about a specific transition of the system, and who formulates a minimal steady-state model of the hidden sector of the system. We lack the generalization of this latter argument to $|\mu| > 1$.

In fact, there exists entire {\it hierarchies} of marginal theories, according to whether one measures $|\mu|=1,2,\ldots, |\alpha|$ currents, up to a ``complete'' set. For example, the above case study admits $|\alpha|! = 24$ hierarchies, among which
\begin{equation}
\label{eq:24}
\begin{aligned}
|\mu| = |\alpha| & = 4  &   & \Scale[1.0]{\ba{c}\xymatrix{
\bullet \ar@{<<-}[r] \ar@{-}[d] & \ar@{-}[dl]\ar@{-}[d]  \bullet  & \bullet  \ar@{->>}[l]   \ar@{-}[d]   \\
\bullet \ar@{<<-}[r] &  \bullet  &  \bullet \ar@{-}[ul]  \ar@{->>}[l]  
}\ea} \nonumber \\
|\mu| & = 3  &   & \Scale[1.0]{\ba{c}\xymatrix{
\bullet \ar@{<<-}[r] \ar@{-}[d] & \ar@{-}[dl]\ar@{-}[d]  \bullet  & \bullet  \ar@{->>}[l]   \ar@{-}[d]   \\
\bullet \ar@{-}[r] &  \bullet  &  \bullet \ar@{-}[ul]  \ar@{->>}[l]  
}\ea} \nonumber \\
|\mu| & = 2  &   & \Scale[1.0]{\ba{c}\xymatrix{
\bullet \ar@{<<-}[r] \ar@{-}[d] & \ar@{-}[dl]\ar@{-}[d]  \bullet  & \bullet  \ar@{-}[l]   \ar@{-}[d]   \\
\bullet \ar@{-}[r] &  \bullet  &  \bullet \ar@{-}[ul]  \ar@{->>}[l]  
}\ea} 	\nonumber \\
|\mu| & = 1  &   & \Scale[1.0]{\ba{c}\xymatrix{
\bullet \ar@{<<-}[r] \ar@{-}[d] & \ar@{-}[dl]\ar@{-}[d]  \bullet  & \bullet  \ar@{-}[l]   \ar@{-}[d]   \\
\bullet \ar@{-}[r] &  \bullet  &  \bullet \ar@{-}[ul]  \ar@{-}[l]  
}\ea}. \nonumber
\end{aligned}
\end{equation}
Within any one such hierarchy, we will be able to show that the mean EPR estimated at each level is smaller than that estimated at the subsequent level,
\begin{align}
0 \leq \ldots \leq \sum_{\mu = 1}^{|\mu|} \Q_\mu^{1,\ldots,|\mu|} \cur_\mu \leq \ldots \leq \sum_{\alpha = 1}^{|\alpha|} \Q_\alpha^{1,\ldots,|\alpha|} \cur_\alpha,
\end{align}
where we now added a superscript as a further specification of the effective affinities, to highlight the fact that they are associated with the $|\mu|$-th theory in the hierarchy. That is because effective affinities associated to the same current, but referring to different levels in the hierarchy, are generally different. The ``real'' affinities $\A_\alpha = \Q_\alpha^{1,\ldots,|\alpha|}$ are the last in the hierarchy.

(We now go back to dropping the hierarchy specification superscript $1,\ldots,|\mu|$.) A system for which all marginal currents vanish, $\cur_\mu^\st = 0, \forall \mu$, is said to be at {\it stalling}, where it stalls. We will show (\Th{th:stalling1}, \Th{iffstalling}) that one achieves stalling if and only if all of the effective affinities vanish:
\begin{align}
\mathrm{stalling:}\quad \Q_\mu = 0 , \forall \mu \iff \phi_\mu = 0 , \forall \mu.
\end{align}
While the marginal currents vanish at stalling, all other currents need not vanish. Hence stalling steady states are generally far from equilibrium, and can be interpreted as states of ``local equilibrium'' with respect to our hypothetical marginal observer. Clearly, the variety of stalling values $\bs{x}^{\st}$ includes that of equilibrium values $\bs{x}^{\mathrm{eq}}$, and typically the latter is a set of zero measure in the former.

Let us now consider response to perturbations out of stalling. The IFR alone grants the validity of the S-FDR, but not of the RR:
\begin{subequations}\label{eq:gko}
\begin{align}
\lcur^\st_{\mu; \mu'} + \lcur^\st_{\mu'\!; \mu} & = \ccur_{\mu\mu'}^{\st}
\label{eq:symmgk}  \\
\lcur^\st_{\mu; \mu'} - \lcur^\st_{\mu'\!; \mu}  & \neq 0.
\label{eq:onsagerp}
\end{align}
\end{subequations}
This is a clear-cut experimental prediction of our theory: While the S-FDR relation is common to response out of equilibrium and out of stalling, the violation of the RR is a signature of stalling. Notice that  Eq.\,(\ref{eq:bmu}) is a guarantee that perturbations with respect to the effective affinities are the same as  perturbations with respect to the ``real'' affinities, so that to test response at stalling no specific experimental protocol has to be devised that is inherently different than that at equilibrium, provided assumption B0 is met. This is crucially important: We want the experimental apparatuses of the ``complete'' and the marginal theories to be the same, because in principle there is no a priori assurance that the system we are going to measure is actually complete.

Let us now look at other signatures of stalling. Differing from the ``real'' affinities, the effective affinities might still depend on the rest of the parameters $\rest := \{x_\kappa\}_{\kappa > |\mu|}$
\begin{align}
\Q_\mu(\bs{x}) = x_\mu - x_\mu^\st(\rest). 
\end{align}
These include the kinetic ones. In other words, the effective affinities might be sensible to modifications of the internal landscape even in the hidden sector of the system. This gives rise, in spite of Eq.\,(\ref{eq:activeres}), to the response formula
\begin{align}
\left( \sum_\mu \lcur_{\mu; \kappa} \frac{d x^\st_\mu}{dx_\kappa} + \frac{1}{2} \sum_{\mu,\mu'} \ccur_{\mu\mu'} \frac{d x^\st_\mu}{dx_\kappa}  \frac{d x^\st_{\mu'}}{dx_\kappa} \right)^\st = 0, \label{eq:pertkin}
\end{align}
equipped with orthogonality relation
\begin{align}
\sum_\mu \cur_\mu \frac{d x_\mu^{\st}}{dx_\kappa} = 0.
\end{align}
Hence, perturbations of kinetic parameters, even far from the observable configurations, might lead to a perturbation of the steady state out of stalling. Local equilibrium is more fragile than equilibrium, as intuitive.

Considering higher cumulants, at third order we find, in spite of the two equations in Eq.\,(\ref{eq:thirdordereq}), that
\begin{align}
 \lcur_{\mu; \mu\mu}^\st - \ccur_{\mu\mu; \mu}^\st  =\frac{1}{3} \ccur_{\mu\mu\mu}^{\st},
\end{align}
while in general
\begin{align}
\ccur_{\mu\mu\mu}^{\st} \neq  0. 
\end{align}
This is due to the fact that, at stalling, the marginal p.d.f. is not necessarily symmetric, $\prob^\st(\{\tcur_\mu\}) \neq \prob^\st(\{-\tcur_\mu\})$. Therefore, the skewness of the p.d.f. is a signature of a stalling steady state.

\begin{center}
\includegraphics[width=.2\columnwidth]{ornament2.pdf}
\end{center}

In the framework of Markov jump-processes on a network briefly described above, marginal currents and their conjugate effective affinities can either be
\begin{itemize}
\item[B1] Currents flowing along single edges, but with some cycles left out from the accounting. Effective affinities are uniquely identified by the theory.
\item[B2] Phenomenological currents. In this case, our construction only holds provided that effective affinities satisfy a condition of {\it marginal consistency} that with simple examples can be shown to be stricter than ``complete'' consistency. 
\end{itemize}
An example of setup B1 is the following network
\begin{align}
\mathrm{B1)} \quad \Scale[2]{\ba{c}\xymatrix{
\bullet \ar@{<<-}[r] \ar@{-}[d] & \ar@{-}[dl]\ar@{-}@[Gray][d]  {\color{Gray} \bullet}  & {\color{Gray} \bullet}  \ar@{-}@[Gray][l]   \ar@{-}@[Gray][d]   \\
\bullet \ar@{<<-}[r] & {\color{Gray} \bullet}  & {\color{Gray} \bullet}   \ar@{-}@[Gray][ul]  \ar@{-}@[Gray][l]  
}\ea} \nonumber
\end{align}
where arrowed transitions are observable, and grey transitions are hidden. The remaining two transitions are kept black because, at a steady state, the current flowing through them is known. Let us denote the observable transitions $i_\mu j_\mu$. The effective affinities are identified according to the following recipe: Remove all the observable transitions (we may assume for simplicity that the network remains connected, but this is not mandatory):
\begin{align}
\Scale[2]{\ba{c}\xymatrix{
\bullet \ar@{-}[d] & \ar@{-}[dl]\ar@{-}@[Gray][d]  {\color{Gray} \bullet}  & {\color{Gray} \bullet}  \ar@{-}@[Gray][l]   \ar@{-}@[Gray][d]   \\
\bullet & {\color{Gray} \bullet}  & {\color{Gray} \bullet}   \ar@{-}@[Gray][ul]  \ar@{-}@[Gray][l]  
}\ea}. \nonumber
\end{align}
On such a reduced configuration space, let us consider the dynamics described by the {\it hidden generator} $\W_{\mathrm{hid}}$ obtained by setting the rates of the observable transitions to zero, $w_{i_\mu j_\mu} = w_{j_\mu i_\mu} = 0$. Let the system relax to the {\it stalling} steady state $\vec{p}^{\,\st}$ of the hidden dynamics, $\W_{\mathrm{hid}} \vec{p}^{\,\st} = 0$. Then the effective affinities are given by
\begin{align}
\Q_\mu = \log \frac{w_{i_\mu j_\mu} p^{\st}_{j_\mu}}{w_{j_\mu i_\mu} p^{\st}_{i_\mu}}.
\end{align}
From an operational perspective, to find the effective affinities there is no need to know the actual rates, as one can just tune parameters to make the currents stall $\{\phi_\mu\} = 0$. In fact, notice that a local parametrization of the rates such as
\begin{align}
\frac{w_{i_\mu j_\mu}(\x)}{w_{j_\mu i_\mu}(\x)} = \exp x_\mu
\end{align}
yields
\begin{align}
\Q_\mu(\x) = x_\mu - x_\mu^{\st}
\end{align}
where $x_\mu^\st = \log (p^{\st}_{i_\mu}/p^{\st}_{j_\mu})$. Furthermore, we can show that there is no difference between ``tuning to stalling'' and ``removing'' as far as the stalling steady state is concerned (see \Th{th:tuning}). This latter expression is the fundamental link between the mathematical and the operational definitions of the effective affinities, thus the cornerstone of the physical interpretation of our theory. The effective affinities have a twofold characterization. On the one hand, they can be interpreted as the forces exerted on the observable edges at a quench, that is, by preparing the system in steady state $\vec{p}^{\,\st}$ and then suddenly switching on the transition rates. On the other, they can be obtained by tuning the controllable parameters to the stalling values that make currents stall. Furthermore, the effective affinities will be shown to be gauge invariant under transformation Eq.\,(\ref{eq:gaugetrans}), exactly like their ``complete'' counterparts, thus granting the compatibility of the theory with foundational requirements.

Notice that, like we briefly mentioned above, when considering an observer who adds more and more currents to his/her observational basket, different marginal theories are generated. Now the reason is clear: the stalling steady state obtained by removing $|\mu|$ edges is different from that obtained by removing a subset $|\mu'| \subset |\mu|$. As observed, the effective affinities conjugate to one particular current differ among themselves at different levels of such hierarchy, which implies that the stalling values $x_\mu^{|\mu|,\st}$ differ as well. From an operational point of view, this phenomenon has a simple interpretation: by virtue of Eq.\,(\ref{eq:pertkin}), once the first $|\mu'|$ currents stall, tuning the other $|\mu|-|\mu'|$ to stalling will also perturb the first, thus disrupting the stalling steady state achieved before. This creates space for an interesting question, whether there exists a smart iterative procedure to tune to stalling.

Like with affinities, we can give a graphical representation of effective affinities. For example, for level $|\mu|=3$ in the hierarchy illustrated above we have three effective affinities. Let's consider only the first, which reads:
\begin{align}
\Q^{1,2,4}_1 & = \log \frac{
\ba{c} \Scale[\scale]{\xymatrix{
\bullet \ar@{<-}[r] \ar@{->}[d] & \ar@{<-}[dl]  \bullet   & \ar@[Gray]@{-}[l] \ar@[Gray]@{-}[d]   \\
\bullet &  \ar@{->}[u] \ar@[Gray]@{-}[l] & \ar@[Gray]@{-}[ul] \ar@[Gray]@{-}[l]}}\ea + \ba{c} \Scale[\scale]{
\xymatrix{\bullet \ar@{<-}[r] \ar@{->}[d] & \ar@{<-}[dl]  \bullet   & \ar@[Gray]@{-}[l] \ar@[Gray]@{-}[d]   \\
\bullet &  \ar@[Gray]@{-}[u] \ar@{->}[l] & \ar@[Gray]@{-}[ul] \ar@[Gray]@{-}[l]}} \ea + \ba{c} \Scale[\scale]{\xymatrix{
\bullet \ar@{<-}[r] \ar@{->}[d] & \ar@[Gray]@{-}[dl]  \bullet   & \ar@[Gray]@{-}[l] \ar@[Gray]@{-}[d]   \\
\bullet &  \ar@{->}[u] \ar@{<-}[l] & \ar@[Gray]@{-}[ul] \ar@[Gray]@{-}[l]}}\ea} {\ba{c}
\Scale[\scale]{\xymatrix{
\bullet \ar@{->}[r] \ar@{<-}[d] & \ar@{->}[dl]  \bullet   & \ar@[Gray]@{-}[l] \ar@[Gray]@{-}[d]   \\
\bullet &  \ar@{->}[u] \ar@[Gray]@{-}[l] & \ar@[Gray]@{-}[ul] \ar@[Gray]@{-}[l]}}\ea + \ba{c} \Scale[\scale]{\xymatrix{ \bullet \ar@{->}[r] \ar@{<-}[d] & \ar@{->}[dl]  \bullet   & \ar@[Gray]@{-}[l] \ar@[Gray]@{-}[d]   \\
\bullet &  \ar@[Gray]@{-}[u] \ar@{->}[l] & \ar@[Gray]@{-}[ul] \ar@[Gray]@{-}[l]}} \ea + \ba{c} \Scale[\scale]{\xymatrix{
\bullet \ar@{->}[r] \ar@{<-}[d] & \ar@[Gray]@{-}[dl]  \bullet   & \ar@[Gray]@{-}[l] \ar@[Gray]@{-}[d]   \\
\bullet &  \ar@{<-}[u] \ar@{->}[l] & \ar@[Gray]@{-}[ul] \ar@[Gray]@{-}[l]}}\ea
}. \end{align}
In a way, like each ``real'' affinity in Eq.\,(\ref{eq:affgraph}) is defined along one fundamental cycle, cycles (more than one) still play a role in the definition of the effective affinity, though in a more involved way. The effective affinity includes all of the cycles that pass through the observable transition and that are not already ``taken care for'' by other observable transitions. This ``dressing'' of the affinity by resumming diagrams is somewhat reminiscent of the paradigm of the renormalization of particles's masses and charges in Quantum Field Theory. It's interesting to compare these effective affinities to the ``real'' affinities  of such cycles $\C\owns \mu , \niton \mu' \neq \mu$ that contain the observable edge under scrutiny but not all others. We can prove as a consequence of the log-sum inequality in information theory that
\begin{align}
\Q_\mu & \leq \frac{\sum_{\C} \phi_\C \A_\C} {\sum_\C \phi_\C},
\end{align}
where the $\phi_\C$ are the so-called Hill cycle currents \cite{hill,hill66}, which provide a fundamental cycle decomposition of the stochastic process.


If instead we are in framework B2, and for example the two observable currents considered above are associated to the same reservoir so that the observer measures the sum of their values
\begin{align}
\mathrm{B2)} \quad 
\Scale[2]{\ba{c}\xymatrix{
\bullet \ar@{<<-}^{\hspace{.2cm} \substack{\rotatebox[origin=c]{180}{$\curvearrowright$}\\  \vspace{-.4cm} }}
[r] \ar@{-}[d] & \ar@{-}[dl]\ar@{-}@[Gray][d]  {\color{Gray} \bullet}  & {\color{Gray} \bullet}  \ar@{-}@[Gray][l]   \ar@{-}@[Gray][d]   \\
\bullet \ar@{<<-}^{\hspace{.2cm} \substack{\rotatebox[origin=c]{180}{$\curvearrowright$} \\  \vspace{-.4cm} }}[r] & {\color{Gray} \bullet}  & {\color{Gray} \bullet}   \ar@{-}@[Gray][ul]  \ar@{-}@[Gray][l]  
}\ea}, \nonumber
\end{align}
then the theory stands on the major assumption of marginal consistency, which in this particular case requires the two effective affinities to be identical. Systems that fail to meet this condition exhibit a violation of the IFR and of the  S-FDR. At stalling, where the phenomenological current vanishes (e.g. the sum of the two currents in the current example), the condition of marginal consistency has an intuitive physical interpretation: we will show in \Th{th:mtctc} that it is met if all of the microscopic edge currents contributing to a phenomenological current also vanish. For example, the following configuration of currents makes the phenomenological current vanish, but it is internally lively, which would lead the observer to estimate a vanishing effective EPR where, instead, there is effective dissipation:
\begin{align}
\Scale[2]{\ba{c}\xymatrix{
\bullet \ar@_{->}@<-0.4mm>[r] \ar@^{-}@<0.0mm>[r]  \ar@^{->}@<0.4mm>[r] &  {\color{Gray} \bullet}  \ar@{--}[dl] 
  \ar@_{->}@[Gray]@<-0.2mm>[r]   \ar@^{->}@[Gray]@<0.2mm>[r]  
      \ar@_{->}@[Gray]@<-0.2mm>[d]   \ar@^{->}@[Gray]@<0.2mm>[d] &  {\color{Gray} \bullet}
    \ar@_{->}@[Gray]@<-0.2mm>[d]   \ar@^{->}@[Gray]@<0.2mm>[d] \\ 
\bullet  \ar@_{->}@<-0.4mm>[u] \ar@^{-}@<-0.0mm>[u] \ar@^{->}@<0.4mm>[u] 
 & {\color{Gray} \bullet}  \ar@_{->}@<-0.4mm>[l] \ar@^{-}@<0.0mm>[l] \ar@^{->}@<0.4mm>[l] & {\color{Gray} \bullet}  \ar@{->}@[Gray]@<0.0mm>[ul]   \ar@{->}@[Gray]@<0.0mm>[l]  
 }\ea} \nonumber
\end{align}
Here, every arrow depicts a ``quantum'' of current;  notice that Kirchhoff's current law is satisfied at each site of the network. For such a system, our theory will not work. \Th{th:reverseconsistency} makes the point that, if for all possible values of the microscopic currents there is no internal dissipation, then the theory is marginally consistent.

\begin{center}
\includegraphics[width=.2\columnwidth]{ornament2.pdf}
\end{center}

The main instrument we will employ is the SCGF of the marginal currents $\lambda(\{q_\mu\})$, which by the {\it contraction principle} in Large Deviation Theory can be obtained from that of the ``complete'' currents by setting $q_\alpha = 0$, for all unobserved currents $\alpha > |\mu|$. As briefly mentioned above, it is well known from Large Deviation Theory that the SCGF is the dominant Perron-Froebenius eigenvalue of the so-called {\it tilted operator} $\M(\{q_\mu\})$, obtained from the MJPG $\W$ by augmenting the off-diagonal entries corresponding to the transitions of interest with exponential factors that depend on the {\it counting variables} $q_\mu$. In general the tilted operator is not a MJPG. Nevertheless, the central result in our paper, stated in \Th{th:unitary} (for one single edge current) and \Th{th:unitary2} (for several edge currents), shows that, letting  $\Past$ be the diagonal positive-definite matrix $\Past = \mathrm{diag}\,\{p^\st_i\}_i$, then the operator
\begin{align}
\widetilde{\W} := \Past \, \M(\{\Q_\mu\})^T \, {\Past}^{-1} \label{eq:marginalgenerator}
\end{align}
is indeed a MJPG, which we call the {\it hidden time-reversal} (hidden TR) generator. Since $\widetilde{\W}$ and $ \M(\{\Q_\mu\})^T$ are similar, their common Perron eigenvalue vanishes and we obtain the {\it marginal IFR} for the SCGF
\begin{align}
\lambda(\{\Q_\mu\}) = 0. \label{eq:iftlap}
\end{align}
For $|\mu|>1$ there actually exists a continuum of values $q_\mu \equiv q^\ast_\mu$ for which $\lambda(\{q^\ast_\mu\}) = 0$. Consider for example the case $|\mu|=2$. In this case, unless the system displays critical behavior (nonequilibrium phase transitions \cite{companion}), the SCGF is a paraboloid-like curve. Its locus of zeroes is a closed convex curve that includes $(0,0)$, and the effective affinities $(\Q^{1,2}_1,\Q^{1,2}_2)$. Also, where it meets with the two axis, it also includes the two effective affinities $(0,\Q^{2}_2)$ and $(\Q^{1}_1,0)$ (some illustrative figures can be found on p.\,\pageref{fig:scgf}). All these values have a special physical interpretation, and in particular they are well-behaved as comes to thermodynamically consistent parametrizations of the rates. All other values $q^\ast_1,q^\ast_2$ are not representative of anything physical, to the best of our understanding.

A compelling question is what kind of process evolves by the hidden TR generator. We collect evidence that  hidden TR  ``tends to preserve'' the dynamics in the observable sector of the configuration space, while it ``tends to invert'' it in the hidden sector. We show that (\Th{th:delsingle}, \Th{th:unitary2}):
\begin{subequations}\label{eq:delcon}
\begin{align}
\W & =: \W_{\mathrm{mar}}  + \W_{\mathrm{hid}}  \label{eq:dede} \\ 
\widetilde{\W} & = \W_{\mathrm{mar}} + \overline{\W_{\mathrm{hid}}}
\end{align}
\end{subequations}
The first equation actually defines the {\it marginal generator} $\W_{\mathrm{mar}}$ on the marginal edge set as that which has vanishing entries for all edges that do not belong to the observable configuration space. In the second expression there appears the time reversal of the hidden generator $ \overline{\W_{\mathrm{hid}}} = \Past \W_{\mathrm{hid}}^T \Past^{-1}$. Notice that, differing from the time reversal of the full generator $\overline{\W}$, here inversion needs to be taken with respect to the stalling steady state (which in fact is the steady state of $\W_{\mathrm{hid}})$. Then, the hidden TR $\widetilde{W}$ only reverses the dynamics in the hidden configuration space. Furthermore, the hidden TR construction is involutive:
\begin{align}
\widetilde{\widetilde{\W}} = \W.
\end{align}

Marginal and hidden degrees of freedom are intertwined. In particular, as a simple consequence of Kirchhoff's current law, one cannot modify hidden currents without affecting the observable currents. So, for example, if this is a steady configuration of currents in the forward dynamics,
\begin{align}
\Scale[2]{\ba{c}\xymatrix{
\bullet \ar@_{->}@<-0.6mm>[r] \ar@^{-}@<0.2mm>[r] \ar@^{-}@<-0.2mm>[r] \ar@^{->}@<0.6mm>[r] &  {\color{Gray} \bullet}  \ar@{->}[dl] 
  \ar@_{->}@[Gray]@<-0.2mm>[r]   \ar@^{->}@[Gray]@<0.2mm>[r]  
      \ar@_{->}@[Gray]@<-0.2mm>[d]   \ar@^{->}@[Gray]@<0.2mm>[d] &  {\color{Gray} \bullet}
    \ar@_{->}@[Gray]@<-0.2mm>[d]   \ar@^{->}@[Gray]@<0.2mm>[d] \\ 
\bullet  \ar@_{->}@<-0.6mm>[u] \ar@^{-}@<0.2mm>[u] \ar@^{-}@<-0.2mm>[u] \ar@^{->}@<0.6mm>[u] 
 & {\color{Gray} \bullet}  \ar@_{->}@<-0.4mm>[l] \ar@^{-}@<0.0mm>[l] \ar@^{->}@<0.4mm>[l] & {\color{Gray} \bullet}  \ar@{->}@[Gray]@<0.0mm>[ul]   \ar@{->}@[Gray]@<0.0mm>[l]  
 }\ea} \nonumber
\end{align}
then the corresponding steady configuration according to the hidden TR dynamics might look something like this (we emphasize that these are just pictorial illustrations):
\begin{align}
\Scale[2]{\ba{c}\xymatrix{
\bullet \ar@_{->}@<-0.4mm>[r] \ar@^{-}@<0.0mm>[r]  \ar@^{->}@<0.4mm>[r] &  {\color{Gray} \bullet}  \ar@{->}[dl] 
  \ar@_{<-}@[Gray]@<-0.2mm>[r]   \ar@^{<-}@[Gray]@<0.2mm>[r]  
      \ar@_{->}@[Gray]@<-0.4mm>[d]   \ar@{-}@[Gray]@<0.0mm>[d]  \ar@^{->}@[Gray]@<0.4mm>[d] &  {\color{Gray} \bullet}
    \ar@_{<-}@[Gray]@<-0.2mm>[d]   \ar@^{<-}@[Gray]@<0.2mm>[d] \\ 
\bullet  \ar@_{->}@<-0.4mm>[u] \ar@^{-}@<0.0mm>[u]  \ar@^{->}@<0.4mm>[u] 
 & {\color{Gray} \bullet}  \ar@_{->}@<-0.2mm>[l] \ar@^{->}@<0.2mm>[l] & {\color{Gray} \bullet}  \ar@{<-}@[Gray]@<0.0mm>[ul]   \ar@{<-}@[Gray]@<0.0mm>[l]  
 }\ea} \nonumber
\end{align}
Notice that the observable cycle currents maintain the same direction, while the hidden currents ``tend to be reversed,'' though such inversion cannot be exact otherwise Kirchhoff's current law would be violated.

We can also consider the behavior of the other law of Kirchhoff, the loop (or cycle) law prescribing the values of the ``real'' affinities. We show (\Th{th:affaff}) that the hidden TR generator reverses all of the hidden ``real'' affinities, while it ``tries to preserve'' the marginal ones: 
\begin{subequations}
\begin{align}
\widetilde{\A}_\mu & = 2 \Q_\mu - \A_\mu, & & \mu \leq |\mu|,\\
\widetilde{\A}_\alpha & = - \A_\alpha,  & & \alpha > |\mu| .
\end{align}
\end{subequations}
Notice that at stalling all of the affinities are exactly reversed. In fact, at stalling the hidden TR generator coincides with the forward TR generator
\begin{align}
\widetilde{\W}(\bs{x}^\st) = \overline{\W}(\bs{x}^\st),
\end{align}
which is the analog of the detailed-balance condition Eq.\,(\ref{eq:dbop}). 

Associated to the marginal dynamics is a marginal path measure $\widetilde{\prob}$, in terms of which we can prove the {\it marginal FR}
\begin{align}
\frac{\prob(\{\tcur_\mu\})}{\widetilde{\prob}(\{-\tcur_\mu\})} = \exp {\sum \Q_\mu \tcur_\mu}, \label{eq:marfparode}
\end{align}
which can be equivalently stated as a marginal fluctuation symmetry as
\begin{align}
\lambda(\{q_\mu\}) = \widetilde{\lambda}(\{\Q_\mu - q_\mu\}).
\end{align}
Notice the crucial difference with respect to its ``complete'' counterpart Eq.\,(\ref{eq:noFR}): in this case we are comparing different probability distributions, which opens up the question whether the hidden TR dynamics can be operationally defined, just like effective affinities were. From Eq.\,(\ref{eq:marfparode}) we can restore the generalized RR at stalling
\begin{align}
\lcur^\st_{\mu; \mu'} = \widetilde{\lcur}^\st_{\mu'\!; \mu}. \label{eq:axial}
\end{align}
More in general, all of the higher-order response relations that characterize ``complete'' systems can be restored upon appropriate hidden TR. Surprisingly, we can even prove, only at long times, a inter-hierarchical FR
	\begin{align}
	\frac{\widetilde{\prob}^{|\mu|}\left(\{\tcur_\mu\}_{\mu = 1}^{|\mu|}\right)}{\widetilde{\prob}^{|\mu'|}\left(\{\tcur_\mu\}_{\mu = 1}^{|\mu|}\right)} \asymp \exp \left( \sum_{\mu=1}^{|\mu|} \Q^{1,\ldots,|\mu|}_\mu \tcur_\mu - \sum_{\mu=1}^{|\mu'|} \Q^{1,\ldots,|\mu'|}_\mu \tcur_\mu \right),
	\end{align}
where $|\mu'|>|\mu|$ and where now $\widetilde{\prob}^{|\mu|}$ is the hidden TR associated to the $|\mu|$-th marginal theory in the hierarchy.

Finally, we also stack one negative result to the pile. Recently an uncertainty relation connecting a current's error and total dissipation has been proven \cite{baratounc,pietzonkaunc,gingrichunc,lazarescu}. However, the bound is not strict and it is only significant when the current quantifies the full dissipation. For marginal currents, it makes for a natural conjecture to speculate that the effective affinity would enter the bound in place of the ``real'' ones. We show in Sec.\,\ref{sec:stasimon} that this is not the case.

\subsection{Epode: Discussion and perspectives}
\label{epode}


Let us draw some general conclusions. More technical perspectives will be discussed in Sec.\,\ref{exode}.

In this paper we present a rather general theory of fluctuation relations and response formulas for an observer that only measures and controls a marginal subset of currents. The context is that of the stochastic thermodynamic analysis of continuous-time, discrete configuration-space autonomous Markov ``jump'' processes. The theory makes some clear-cut experimental predictions, in particular the integral fluctuation relation with respect to the effective affinities, the violation of the reciprocal relations at stalling steady states, and the validity of the symmetrized fluctuation-dissipation relation. The theory is fairly complete as regards currents that account for single transitions in configuration space, and it also holds for phenomenological currents that are linear combinations of edge currents, provided the additional requirement of marginal consistency is met, which is analogous to local detailed balance, but stricter. Therefore, the most imminent open question left aside is what kind of systems satisfy marginal consistency, and if a system does not, how does the surplus of entropy production at the stalling states affect response.

Central objects in the theory are the effective affinities. While they are mathematically expressed in terms of the rates of the Markovian dynamics all over the configuration space, an operational procedure allows to evaluate them without full knowledge of the transition rates, provided the parameters that the observer controls are known to only affect the rates corresponding to the measurable degrees of freedom. If this is not the case, then one can turn the story the other way around, and use the predictions of our theory as a test of locality of the physical parameters.

The full fluctuation relation and the reciprocal relations can be reinstated upon the identification of a suitable hidden time-reversed Markovian dynamics. The question is open whether to obtain such dynamics one needs to be able to micro-engineer all rates, in which case the latter relations remain only formal, or else, as is the case for the effective affinities, whether there exists a phenomenological procedure to determine the dynamics. This would unlock a new set of experimental predictions of our theory. A test-bed for this possibility is that of a computational experiment: is it possible to program a Gillespie simulation of the hidden time-reversed dynamics without specifying all of the rates as an input, but rather by performing a smaller transformation of the known parameters with respect to the simulations of the forward dynamics? We are not yet in the position to give a definitive answer to this question.

An important consideration is that ours is not a kinetic theory, that is, it does not provide a procedure to coarse-grain the dynamics in the hidden sector of the configuration space in order to obtain an effective dynamics in the marginal configuration space. While the {\it gedanken}-observer described in Ref.\,\cite{polettiniobs} and Sec.\,\ref{sec:marepr} does cook up a marginal dynamics that explains his steady-state observations, in no way this dynamics is representative of the finite-time behavior, including such questions as the rate of convergence to the steady state, first exit times out of the hidden sector etc. Furthermore, our theory does not involve an exquisitely dynamical, but physically relevant limiting situation, that of time-scale separation between the marginal and the hidden degrees of freedom. The relationship between our theory and various other approaches, such as those described in Refs.\,\cite{bravi,bo}, is an interesting territory to explore.

The results that we presented are amenable to several generalizations. To lattice gas models, where response theory is enriched by all aspects regarding the spatial disposition of particles \cite{leitmann}. To Markov jump processes on infinite configuration spaces, in particular population dynamics, chemical reaction kinetics, and reaction-diffusion theory. To diffusion processes on continuous configuration spaces. To time-periodic processes, rather than stationary, and more generally to time-dependent perturbations, to systems with resetting \cite{pal}. To finite-time response to a sudden perturbation $x_\mu(t) = \theta(t) (x_\mu - x_\mu^\st)$ ($\theta$ being Heaviside's step-function), or to perturbations that are modulated in a finite-time interval. Periodicity calls for a study in Fourier space, where response relations incarnate into susceptibilities and spectral response functions, and where it is already known that far from equilibrium several of the equilibrium results are violated \cite{harada}.

Further questions on marginal and effective theories are genuinely thermodynamic. As a matter of fact, any question addressed in recent years in the field can be turned marginal: the study of efficiency and efficiency fluctuations, of the linear regime where (marginal) currents are approximately linear in the (effective) affinities, of transduction \cite{hill}, of variational principles such as the minimum and maximum entropy production principles \cite{polettiniminEP,polettiniziegler}. Recent models cope with strong system-environment interactions by envisaging the system as a subsystem of a larger system-environment complex, itself weakly interacting with its surrounding. Again, such system-environment complex could be analyzed in terms of our theory. Systems that have irreversible transitions, such as stochastic processes with resetting, always posed a challenge, because the thermodynamic force diverges along irreversible transitions;  one way out, among others \cite{yuto}, could be to dump the irreversible transitions into the hidden trash bin.

The observables that we consider, the currents, are antisymmetric under time reversal. A new central paradigm is that response out of equilibrium depends in a crucial way on the {\it activity} of the system, i.e. some measure of the gross amount of stuff flowing, in opposition to the current that measures the net amount of stuff delivered. Many recent results regarding currents have been generalized to flows and other symmetric quantities, e.g. the uncertainty relation briefly mentioned above \cite{garrahan}, fluctuation relations  \cite{falasco}, and response formulas  \cite{baiesi,baiesi2,prost}.

Violations of the reciprocal relations are often associated with broken time-reversal symmetry, e.g. the microscopic dynamics involves axial fields that are antisymmetric under time reversal, such as magnetic fields, Coriolis forces, the momentum variable in underdamped Brownian motion etc. It is well known that in these cases the Onsager symmetry can be restored upon inversion of the axial fields. With an eye on Eq.\,(\ref{eq:axial}), it is tempting to speculate that the marginal dynamics might be the discrete analog of the axial-field inversion. The analysis of proper time reversal of continuous noisy systems with even and odd  variables with respect to the FR has been broadly studied \cite{spinney}. If our speculation is fruitful, it would allow to include even and odd variables within the formalism of Markov jump-processes without additional requirements.

Local observers are reminiscent of the theory of relativity. Some authors have considered \cite{speck,gawedzki} fluctuation relations in moving frames where a ``local equilibrium'' can be attained by a privileged observer. It would be interesting to inspect whether such transformations could be framed within our theory of a marginal observer.

One interesting aspect that is completely missing is that of {\it duality}, whereby one swaps the role of the marginal and the hidden state spaces. Is there any relation between the theories so obtained? We notice in passing that by first performing the hidden time reversal, and then the dual hidden time reversal, one does {\it not} obtaine the ``complete'' time reversal of the forward dynamics. Thus, if a relation exists, it might be subtle.

Let us conclude with some more epistemological remarks. At all stages we insisted on its operational character. We also revived gauge invariance as a solution to the ``dilemma of the observer''. This is because we strongly believe that physics is not about properties of some absolute ``thing in itself,'' but it rather deals with relations and processes, and about how an idealized observer interprets his/her observations. Furthermore, a priori there is no reason to presume that a system is ``complete''. In our view, theories are always marginal to some extent -- in particular statistical physics is intrinsically a theory of incomplete information. For this reason we always comment the words ``complete'' and ``real''.

\section{Episode 1: Setup}

Here we provide a compendium of the  thermodynamic analysis of ``complete'' systems evolving by a Markovian jump-process dynamics on a graph, to set the notation, introduce the basic techniques, and provide numerous references to more in-depth studies. While the knowledgeable reader might safely skip this section, the newbie should not be discouraged either, as in the following sections we will attempt to construct our theory in a pedagogical and self-contained manner. The ornament on p.\,\pageref{victorian} divides the old material from the new one.

\subsection{Algebraic graph theory in a pistachio-shell} 

The system's configuration space is a finite oriented graph $\G = (\I,\E,\partial)$ with a number 
$|\I|$ of sites $i,j,\ldots \in \I$ connected by $|\E|$ oriented edges $\ij\in \E$, corresponding to the possible transitions between sites. We assume that the graph is connected, without loops nor multiple edges between two sites\footnote{This assumption excludes the possibility of resolving multiple transitions, which is crucial in stochastic thermodynamics, especially in the light of the assumption of local detailed balance \cite{ldbmassi} whereby different reservoirs enhance transitions. The generalization of all of our results is straightforward, but it makes the notation overly baroque, to the detriment of clarity. We discuss it in Sec.\,\ref{sec:multipledges}.}. We assign an orientation to the edges $\ij = i \gets j$ by prescribing an arbitrary order relation\footnote{This is just one way to introduce an arbitrary orientation of the edges. Not all orientations come from an order relation.} $i \prec j$.

The {\it incidence matrix} $\partial: \mathbb{R}^{|\E|} \to \mathbb{R}^{|\I|}$, prescribing which sites are boundaries of which edges, has entries $\partial_k^{\ij} = \delta^i_k - \delta^j_k$. Square matrices defined on the configuration space $\I$ of a system are denoted $\matrix{A}: \mathbb{R}^{|\I|} \to  \mathbb{R}^{|\I|}$, and they act on vectors $\vec{v} \in \mathbb{R}^{|\I|}$. All other vectors, including those living in the linear space generated by edges $\mathbb{R}^{|\E|}$, live in the linear space generated by edges, are denoted in bold $\bs{v}$. 

A spanning tree $\T \subseteq \E$ is a collection of $|\I|-1$ (unoriented) edges that connect all sites. In a rooted spanning tree $\T_{\!i}$ edges are oriented in such a way that there is a unique directed path leading from any site to $i$. An oriented cycle is a succession of oriented edges such that at every site there are as many incoming edges as outgoing ones. A cycle is {\it simple}, and it is denoted $\C$, when it has no crossings. A cycle can be algebraically identified as a right-null vector $\bs{c}$ of the incidence matrix, $\partial \bs{c}= 0$.

\subsection{Master equation dynamics}
\label{setup}

\subsubsection{Master equation}
\label{subsec:flux}
We assign to each edge time-independent positive transition rates $w_{\ij}$ and $w_{\ji}$ of jumping respectively from $j$ to $i$ and from $i$ to $j$, and we let $w_i := \sum_j w_{\ji}$ be the total escape rate out of site $i$. Let $\vec{p}(t) = (p_i(t))_i$ be the vector of probabilities of being at a site at a given time $t$, sometimes called {\it ensemble}. Given the initial ensemble $\vec{p}(0) = \vec{p}^{\,0}$, $\vec{p}(t)$ obeys the master equation Eq.\,(\ref{eq:ME}). Entries along columns of $\W$ add up to zero, $\W^T \vec{1} =0$, where $\cdot^T$ is matrix transposition and $\vec{1}$ is the vector with all entries equal to unity. The master equation can be cast as a continuity equation
\begin{align}
\frac{d}{dt}\vec{p}(t) + \partial \bs{\cur}(t) = 0
\end{align}
in terms of the vector of currents $\bs{\cur}(t) = (\cur_{\ij}(t))_{\ij \in \E}$ with entries  $\cur_{\ij}(t) = \psi_{\ij}(t) - \psi_{ji}(t)$, where
\begin{align}
\psi_{\ij}(t) := w_{\ij} p_j(t)
\end{align}
is sometimes called the mean flux from site $j$ to $i$.

\subsubsection{Steady ensemble} 
Assuming that the graph is connected and that rates are non-negative, then the system tends to a steady ensemble $\vec{p} =\lim_{t \to \infty} \vec{p}(t)$ that is the unique right-null vector of the generator, and that makes the steady currents $\bs{\cur} = \lim_{t\to \infty} \bs{\cur}(t)$ ``divergenceless'':
\begin{align}
0 = \W \vec{p} = - \partial \bs{\cur} . \label{eq:steady}
\end{align}
The right-hand side of this equation is Kirchoff's current law. It is well known that, up to a normalization factor, the steady ensemble can be found in terms of minors of the MJPG  \cite{schnak,gaveau}
\begin{align}
p_i \propto (-1)^{i+j} \det \W{(j|i)} \label{eq:matrix}
\end{align}
where $\W{(j|i)}$ is the matrix obtained by removing the $j$-th row and the $i$-th column; the above expression holds independently of $j$. A one-line proof of this fact is as follows: since the determinant of $\W$ vanishes (its null eigenvector being the steady state), expanding with the Laplace formula along the $j$-th column $0 = \det \W = \sum_{j\neq i} w_{\ij} (-1)^{i+j} \W{(j|i)} - (-1)^{i+j} w_{i}\W{(i|i)}$, and we conclude $\Box$.

The steady ensemble can be expressed in terms of rooted oriented spanning trees as
\begin{align}
p_i  =  \frac{\ST_{\!i}(\G) }{\ST(\G)} \label{eq:tree}
\end{align}
where
\begin{align}
\ST_i(\G) := \sum_{\T_{\!i} \subseteq \E} \prod_{\ij \in \T_{\!i}} w_{\ij},
\end{align}
is the so-called spanning-tree polynomial, where $\T_{\!i}$ ranges over oriented spanning trees with root in $i$, and $\ST(\G) = \sum_i \ST_{\!i}(\G)$ is the normalization. Since rates $w_{\ij}$ have dimensions of an inverse time, we have a liberty in the choice of the time unit. We choose to spend this liberty by setting, unless otherwise stated,
\begin{align}
\ST(\G)  \stackrel{!}{=} 1. 
\end{align}
The equivalence between Eqs.\,(\ref{eq:matrix}) and (\ref{eq:tree}) is an instance of the matrix-tree theorem in algebraic graph theory, an important paradigm that will play a major role in the physical interpretation of our results, in particular when we will consider portions of the configuration space, a case that is covered by the crucial all-minors matrix-tree theorem for weighted oriented graphs \cite{chaiken}.

A steady state is said to be an equilibrium if it satisfies the condition of detailed balance
\begin{align}
\psi_{\ij}^{\eq} = {w_{\ij} p^{\eq}_j} = {w_{\ji} p^{\eq}_i} = \psi_{\ji}^{\eq}. \label{eq:db}
\end{align}
Hence at equilibrium the steady currents vanish. Equilibrium steady states admit a particularly simple expression in terms of ratio of the rates:
\begin{align}
p^{\eq}_k = \left[ \sum_{l} \prod_{\ij \in \gamma_{k \gets l}} \frac{w_\ij}{w_{\ji}} \right]^{-1}, \label{eq:gaugess}
\end{align}
where $\gamma_{k\gets l}$ is an arbitrary connected path leading from state $l$ to $k$.

\subsubsection{Time reversal} The generator time-reversed dynamics is defined as follows. Given a forward generator $\W$, compute its steady state $\vec{p}$, construct the diagonal matrix $\Pp$ whose diagonal entries are the steady-state probabilities, $\Pp := \mathrm{diag}\,\{p_i\}_i$. Then the TR generator is
\begin{align}
\overline{\W} := \Pp\,\W^T\,\Pp^{-1} \label{eq:timereversed}
\end{align}
It can easily be shown that $\overline{\W}$ is indeed a \MJPG, with, among others \cite{polettiniconvex}, the following properties: same spectrum as $\W$, same exit frequencies out of configurations, same steady state, all inverted steady-state currents. A system satisfies detailed balance if and only if $\W = \overline{\W}$.

\subsection{Master equation thermodynamics}

We hereby consider mean currents. The stochastic counterpart is covered in the next subsection.

\subsubsection{Observational currents}

Currents of observational interest are linear combinations of the edge currents,
\begin{align}
\cur_\alpha = \sum_{\ij} \phys_\alpha^{\ij} \, \cur_{\ij},
\end{align}
where $\sum_{\ij}$ sums over edges, while $\sum_{i,j}$ sums over couple of sites. By the handshaking lemma in graph theory, $\sum_{i,j} = 2 \sum_{\ij}$ for any summand.

The antisymmetric weight factor $\phys_\alpha^{\ij} = - \phys_\alpha^{\ji} $ prescribes by what amount the $\alpha$-th current increases [decreases] when transition $i \gets j$ [$ j\gets i$] is performed. 
We will describe the conditions upon which such currents form a complete set later in this section. Observational currents might either be subsets of the edge currents, in which case there exists some particular edge $\ij_\alpha$ such that $\phys_\alpha^{\ij} = \delta_{\ij,\ij_\alpha}$, or otherwise they are {\it phenomenological}, which means that at least one such current is supported on more than one edge. The first will be the subject study of Secs.\,\ref{sec:singleedge} and \ref{sec:multiedge}, the second of Sec.\,\ref{sec:phenomenological}. While we use a unified notation for all observational currents, the treatment of phenomenological currents poses specific problems. Of course, the observer might stipulate that {\it all} edge currents are of observational interest, in which case $\alpha$ is a multi-index $\alpha = i'\!j'$, with $\phys_\alpha^{\ij} = \delta^i_{i'} \delta^j_{j'}$.

\subsubsection{Forces, entropy production rate and local detailed balance}

Thermodynamic reasoning ensues when one complements the dynamical information contained in currents with conjugate {\it forces} that quantify the cost of transitions. The steady-state force associated to a particular transition is given by
\begin{align}
\F_{\ij} := \log \frac{w_{\ij} p_j}{w_{\ji} p_i} = \log \frac{\psi_{\ij}}{\psi_{\ji}}. \label{eq:ssforce}
\end{align}
Clearly, a system has all vanishing steady forces if and only if it satisfies the condition of detailed balance Eq.\,(\ref{eq:db}).

The mean steady-state EPR is defined as
\begin{align}
\sigma := \ssum \cur_{\ij} \F_{\ij} = \ssum (\psi_{\ij} - \psi_{\ji}) \log \frac{\psi_{\ij}}{\psi_{\ji}} \geq 0. \label{eq:EPR}
\end{align}

When working with phenomenological currents, to achieve a thermodynamically consistent treatment the corresponding thermodynamic forces cannot be arbitrary, rather they need to enjoy certain symmetries, in such a way that ultimately the steady entropy production rate can be expressed in terms of the phenomenological currents and forces only. Thermodynamic consistency (or simply, consistency) is realized if the following condition of {\it local detailed balance} holds
\begin{align}
\F_{ij} =  \ssum \phys^\alpha_{\ij} \F_\alpha +  a_j - a_i  \label{eq:ldb}
\end{align}
where the $\F_\alpha$ are observational thermodynamic forces. For the sake of generality we included an arbitrary function of the configuration $a_i$ that takes into account the liberty offered by Kirchhoff's current law at steady states (see next section) and the steady state. Under this condition the steady EPR writes
\begin{align}
\sigma = \ssum \phi_\alpha \F_\alpha.
\end{align}
Notice that any dependence on $a_i$ is lost. Eq.\,(\ref{eq:ldb}) is not just a convenient physically meaningful parametrization of the rates, it actually imposes constraints on the space of possible rates, that can be interpreted as {\it symmetries} that the edge forces must satisfy. The number of these symmetries eventually increases if phenomenological currents are not linearly independent, i.e. if $\ssum \ell_\alpha \phys_{\ij}^\alpha = 0$ for some vector $(\ell_\alpha)_\alpha$. Symmetries and conservation laws arise from the interplay between the definition of the phenomenological currents, an information contained in $\phys^{\ij}_\alpha$, and the structure of the network, an information contained in $\partial$ \cite{bridging}. In this work we exclude the possibility of linearly dependent currents, though it would make for an interesting problem to investigate our results in those situations where a conserved current flows across the marginal/hidden configuration space.

\subsection{Cycle analysis} 
\label{sec:cycle}

Cycles are ubiquitous in thermodynamics. For example, cycles solve Kirchoff's current law, hence they are useful to describe steady states. While the focus is usually on the methods described in Schnakenberg's review \cite{schnak}, where the freedom in the choice of a cycle basis of $\mathrm{ker}\,\partial$ is broken to provide a compact expression for the EPR, there exists another cycle decomposition that is less compact but more general, and which will turn out to play an important role. In the following we will (somewhat improperly) refer respectively to Schnakenberg's and Hill's analysis.

\subsubsection{Schnakenberg analysis}

Consider an arbitrary spanning tree $\T$. There are $|\E| - |\I| +1$ edges $\ij_\alpha= i_\alpha \gets j_\alpha \in \E\setminus \T$, called {\it chords}, that do not belong to the spanning tree. Adding a chord to $\T$ generates a unique simple cycle $\C_{\alpha}$, that can be oriented along the direction of $\ij_\alpha$. To such cycle we associate a vector $\bs{\cyc}_\alpha$ with entries
\begin{align}
\cyc_\alpha^{\ij} =  \left\{\ba{ll}
+1, & \mathrm{if}\, \ij \in \C_\alpha \\
-1, & \mathrm{if}\, \ji \in \C_\alpha \\
0 & \mathrm{otherwise}
\ea \right.   .
\end{align}
The set of simple oriented cycles so generated forms a basis for the right null-space of the incidence matrix, $\partial \bs{\cyc}_\alpha = 0$. Hence, in view of Eq.\,(\ref{eq:steady}), cycles describe steady states. Transient states can be studied in terms of cocycles \cite{polettini2} (see below). In particular, we define {\it chord currents}, obtained by setting
\begin{align}
\phys^{\ij}_\alpha = \left\{\ba{ll}
+1, & \mathrm{if}\; \ij = \ij_\alpha \\
-1, & \mathrm{if}\; \ji = \ij_\alpha \\
0 & \mathrm{otherwise}
\ea \right. ,  
\end{align}
and their conjugate cycle affinities,  defined as
\begin{align}
\A_\alpha := \A(\C_\alpha) := \sum_{\ij \in \C_\alpha} \F_{\ij} =  \log \prod_{\ij \in \C_\alpha} \frac{w_{\ij}}{w_{\ji}}.
\end{align}
Notice that the steady state disappears from the final expression. Also, because $\sum_{\ij} \cyc_\alpha^{\ij} \,  \phys^{\ij}_{\alpha'} = \delta_{\alpha,\alpha'}$, consistently with the condition of local detailed balance Eq.\,(\ref{eq:ldb}) we obtain $\A_\alpha = \F_\alpha$ and thus the EPR only writes in terms of cycle observables $\sigma = \ssum \cur_\alpha \A_\alpha$. 

We conclude this section by providing an interesting expression for the cycle affinities. Consider the system obtained by removing all of the chords, {\it ossia} by setting their rates to zero. Because it has no cycles, it satisfies detailed balanced. Let $\vec{p}^{\,\eq}$ be its equilibrium steady state. Then, it is easy to show that
\begin{align}
\A_\alpha = \log \frac{w_{\ij_\alpha} p_{\!j_\alpha}^{\eq}}{w_{\ji_\alpha} p_{i_\alpha}^{\eq}}.
\label{eq:cycleaff}
\end{align}
It follows from the fact that, due to the property of detailed balance, by Eq.\,(\ref{eq:gaugess}) the equilibrium state obeys
\begin{align}
\frac{p_{k}^{\eq}}{p_l^{\,\eq}} = \prod_{\ij \in \gamma_{k \gets l}} \frac{w_{\ij}}{w_{\ji}}. 
\end{align}
where $\gamma_{k\gets l} \subseteq \T$ is now unique.

\subsubsection{Hill analysis} As a consequence of the spanning-tree expression for the steady state Eq.\,(\ref{eq:tree}), the steady-state current along edge $\ij$ can be written as a sum over all simple oriented cycles that contain edge $\ij$
\begin{align}
\cur_{\ij} = \sum_{\C \owns \ij} \left( \psi_{+\C} - \psi_{-\C}  \right), \label{eq:hill} 
\end{align}
where the {\it cycle fluxes} are given by
\begin{align}
\psi_{\pm\C} = \theta_\C \prod_{\ij \in \pm\C} w_{\ij} .
\end{align} 
The factor $\theta_\C = T_{\C}(\mathcal{G}/\C)$ is an rooted oriented spanning tree polynomial over the graph obtained by contracting cycle $\C$ to a unique vertex. Importantly, as observed in Ref.\,\cite{zia1}, it is symmetric under inversion of the cycle's orientation. Therefore the cycle affinity can be written as
\begin{align}
\A(\C) = \log \frac{\psi_{+\C}}{\psi_{-\C} } .
\end{align}
With a few passages one obtains for the entropy production rate
\begin{align}
\sigma =  \sum_{\C} (\psi_{+\C}- \psi_{-\C}) \log \frac{\psi_{+\C}}{\psi_{-\C} } 
\label{eq:hillepr}
\end{align}
where it is stipulated that each cycle is summed over only once, along one arbitrary choice of its orientation (otherwise a factor $1/2$ should be included).

We call the quantity
\begin{align}
\sigma_{\ij}^{\mathrm{Hill}} = \sum_{\C \owns \ij} (\psi_{+\C}- \psi_{-\C}) \log \frac{\psi_{+\C}}{\psi_{-\C} }   \label{eq:hillEPR}
\end{align}
the {\it local EPR} associated to Hill's cycle decomposition.

\subsubsection{Cocycles}
\label{par:cocycle}

By the rank-nullity theorem, the edge vector space $\matrix{R}^{|\E|}$ of a graph can be decomposed in a basis of cycles $\bs{c}$, which span the null vectors of the incidence matrix $\partial$, and of {\it cocycles} (also known in graph theory as {\it cuts} or {\it bonds}) $\bs{\cocyc}_{\alpha^\star}$, which span the image of $\partial$. By construction cocycles are orthogonal to cycles
\begin{align}
\sum_{\ij} c^{\ij}_{\alpha^\star} c^{\ij}_{\alpha} = 0. 
\end{align}
{\it Simple} cocycles are minimal sets of edges whose removal disconnects the graph into two subgraphs. The algebra of cycles and cocycles and their relationship to thermodynamics has been studied to great extent by one of the Authors in Ref.\,\cite{polettini2}. Cocycles play a role in characterizing transient states. In our work, they will play a minor role related to gauge invariance, see Sec.\,\ref{sec:gauge}.

\subsection{Stochastic thermodynamics}
\label{stocthermo}

\subsubsection{Trajectories and their measure}
\label{subsec:measure}
 We consider a single realization of a jump-process\footnote{For an exact and explicit mathematical characterization of the process as a solution of a stochastic equation in terms of Poisson process with intensity depending on the configuration the process visits, see \cite{andersonbook} in the specific case of the chemical master equation and \cite{kurtzbook} for general Markov jump processes}, described by the trajectory
\begin{align}
\traj := i_0 \stackrel{t_0}{\longrightarrow} i_1 \stackrel{t_1}{\longrightarrow} \ldots \stackrel{t_{N-1}}{\longrightarrow} i_N \stackrel{t_N}{\longrightarrow}, \label{eq:traj}
\end{align}
which performs $N$ jumps in the time arc $[0,t]$. It is described by the succession of sites visited\footnote{Notice that, had we resolved multiple edges, we should characterize a trajectory not by the sites visited but rather by the edges}  and by that of waiting times at sites, which add up to $\sum_{n=0}^N  t_n =t$. The trajectory p.d.f. is given by\footnote{More precisely, what follows is a family of p.d.f.'s labeled by $N$, the number of arguments which depends on total number of jumps $N$, itself a stochastic variable} 
\begin{align}
\prob[\traj] =  e^{-w_{i_N} t_N} \prod_{n=0}^{N-1}\left( w_{i_{n+1},i_n} \, e^{-w_{i_n} t_n}\right) \,p^0_{i_0}  \label{eq:density}
\end{align}
with respect to the trajectory integration measure
\begin{align}
\int \mathcal{D}\traj = \sum_N \sum_{i_0,\ldots,i_N}  \int_0^t  \prod_{n=0}^N dt_n \,\delta\Big( t - \sum_n t_n \Big).
\end{align}
The remarkable review Ref.\,\cite{weber} proposes two derivations of the measure over realizations of Markov jump processes, one heuristic, based on the intuitive mechanisms of the Gillespie algorithm, and one exact, based on the Laplace transform of the propagator (see also \cite{BEST} for analytical inversion formulas). The time-reversed trajectory is
\begin{align}
\invtraj := i_N \stackrel{t_N}{\longrightarrow} i_{N-1} \stackrel{t_{N-1}}{\longrightarrow} \ldots \stackrel{t_{1}}{\longrightarrow} i_0 \stackrel{t_0}{\longrightarrow}. \label{eq:trtraj}
\end{align}

The following fact, which is a fundamental (and often underestimated) passage, will play an important role in our discussion, because it is one of the few missing links that differentiates complete and marginal thermodynamics. Time reversal of the dynamics induces a time-reversed measure over trajectories $\overline{\prob}$, such that the time-reversed p.d.f. of a trajectory, conditioned on the initial configuration, coincides with the forward probability density of the time-reversed trajectory, conditioned on the final configuration:
\begin{align}
\overline{\prob}(\traj|i_0) = \prob(\invtraj|i_N). \label{eq:equipp}
\end{align}

\subsubsection{Time-integrated currents} We now introduce the stochastic observables of crucial interest, defined along a trajectory $\traj$. Letting $\tau_n = \sum_{n' < n} t_{n'}$ be the time elapsed before the $n$-th jump, the time-integrated {\it edge fluxes}
\begin{align}
\Psi^t_{\ij} := \Psi_{\ij}[\traj] & : = \int_0^t d\tau \sum_n \delta\left(\tau- \tau_n \right) \delta_{i_{n+1},i} \, \delta_{i_n,j}   \label{eq:integrand}
\end{align}
count the net number of times transition $j \to i$ occurs. The integrand in Eq.\,(\ref{eq:integrand}) is a spiking function of time \cite{andrieux}. Edge currents are defined as
\begin{align}
\tcur^t_{\ij} := \Psi^t_{\ij} - \Psi^t_{\ji}
\end{align}
Currents are antisymmetric under time reversal ${\tcur}[\invtraj] = - \tcur[\traj]$. The corresponding symmetrized quantities are sometimes called {\it activities}. Activities will not play a role in our theory.

\subsubsection{Currents' p.d.f. and scaled cumulants}
\label{subsec:scgf}

We will denote by $\prob^t(\{\tcur_\alpha\})$ the density function of the probability that the currents take values $\tcur^t_\alpha \in [\tcur_\alpha,\tcur_\alpha+d\tcur_\alpha]$ (we will drop the superscript $^t$ from now on). Formally,
\begin{align}
\prob(\{\tcur_\alpha\}) = \int \mathcal{D}\traj  \,P(\traj) \prod_\alpha \delta\Big(\tcur_\alpha[\traj] - \tcur_\alpha \Big).
\end{align}
It can be shown that cumulants of the time-integrated currents grow linearly in time. Let us then introduce the {\it scaled-cumulant generating function} (SCGF), 
\begin{align}
\lambda(\{q_\alpha\}) := \lim_{t \to \infty} \frac{1}{t}  \log \int e^{- \sum q_\alpha \tcur_\alpha } \prob(\{\tcur_\alpha\}) \prod_{\alpha'} d\tcur_{\alpha'}
\end{align}
and the scaled cumulants
\begin{align}
\ccur_{\beta_1,\ldots,\beta_{|\alpha|}} :=  (-1)^{\beta_1+\ldots + \beta_{|\alpha|}}  \left.\frac{\partial^{\beta_1}}{\partial q_1^{\beta_1} } \ldots \frac{\partial^{\beta_{|\alpha|}}}{\partial q_{|\alpha|}^{\beta_{|\alpha|}} } \lambda(\{q_\alpha\}) \right|_{\{q_\alpha\}=\{0\}}.
\end{align}
SCGFs can be calculated as the ``Perron-Froebenius'' eigenvalue (i.e. unique eigenvalue with largest real part) of the so-called {\it tilted operator}\footnote{We call it tilted {\it operator}, rather than {\it generator}, because it does not generate a Markov jump-process dynamics.} $\M(\{q_\alpha\})$. Defining the tilting matrix $\matrix{T}(\{q_\alpha\})$ with entries
\begin{align}
\matrix{T}(\{q_\alpha\})_{\ij} : = e^{\sum_\alpha q^\alpha \phys_\alpha^{\ij}}
\end{align}
the tilted operator is defined as
\begin{align}
\M(\{q_\alpha\}) : = \W \circ \matrix{T}(\{q_\alpha\}).
\end{align}
where $\circ$ is the entry-wise Hadamard product \label{hadamard}. For derivations, see \cite{andrieux} for cycle currents, \cite{companion} for one-dimensional lattice gases. For useful implicit function techniques for the calculation of the cumulants, we will make use of Refs.\,\cite{koza,altaner15a,altaner15b}. Importantly, the SCGF for the edge currents, and of any linear contraction of them, is a convex function. An accessible proof of this fact can be found in \cite{qians}).

Finally, the contraction principle in large deviation theory allows to derive the SCGF of a coarse-grained observable from that of more specialized observables. For definiteness, suppose we know the SCGF $\lambda_{\E}(\{q_{\ij}\})$ of {\it all} the edge currents (which, as a matter of fact, is usually inaccessible). Then the rate function for the observable currents reads
\begin{align}
\lambda(\{q_\alpha\}) = \lambda_{\E}\left(\{\ssum \phys^\alpha_{\ij} q_\alpha\}\right).
\end{align}

\subsubsection{Fluctuation relations for the currents}
\label{eq:gifr}

The forward and backward p.d.f. obey the FR
\begin{align}
\frac{P(\traj)}{P(\invtraj)} = \frac{p^0_{i_0}}{\overline{p}^0_{i_N}} \frac{p^{\mathrm{eq}}_{i_N}}{p^{\mathrm{eq}}_{i_0}} \exp {\ssum \A_\alpha \tcur^t_\alpha}  \label{eq:transientFR}
\end{align}
where we $(\overline{p}^{\,0}_i)_i$ is the initial sstate from which the backward trajectory is sampled, which in general can be different from that from which the forward trajectory is sampled. This relation holds for cycle currents, and for phenomenological currents provided thermodynamic consistency is respected. The appearance of the equilibrium state in the above equation is far from trivial, and it is related to a cycle/cocycle decomposition of the EPR \cite{cuetara,polettiniFT}, or equivalently to the expression Eq.\,(\ref{eq:cycleaff}) for the affinities.

The multivariate FR for the cycle currents Eq.\,(\ref{eq:fr}) is then obtained by taking the marginal for the currents, either asymptotically in time for any initial ensemble, where the boundary terms become subdominant with respect to the entropy production, or at all times on the assumption that the the trajectories are sampled with the equilibrium ensemble $\vec{p}^{\,0} = \vec{p}^{\,t} = \vec{p}^{\,\mathrm{eq}}$. The IFR Eq.\,(\ref{eq:ifr}) is easily obtained by integrating the above FR. 

\begin{center}
\includegraphics[width=.5\columnwidth]{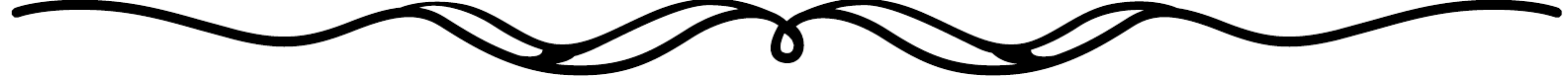}
\label{victorian}
\end{center}


In fact, a wider class of such IFRs can be generated as follows. We consider here subsets of currents, namely the first $|\mu|$ currents and of the last  $|\alpha|-|\mu|$. We have from the FR Eq.\,(\ref{eq:fr})
\begin{align}
\prob (\{ \tcur_\alpha\})   \exp {\sum_{\alpha \leq |\mu|} \tcur_\alpha \A_\alpha}  =   \prob (\{ - \tcur_\alpha\}) \exp {\sum_{\alpha > |\mu|} \tcur_\alpha \A_\alpha}.
\end{align}
Integrating over currents and employing antisymmetry, we obtain the {\it generalized IFR}
\begin{align}
\Big\langle \exp {\sum_{\mu \leq |\mu|} \tcur^t_\mu \A_\mu} \Big\rangle = \Big\langle \exp {\sum_{\alpha > |\mu|} \tcur^t_\alpha \A_\alpha} \Big\rangle.
\end{align}
In particular, for two currents, for $|\mu|=1,2$ we obtain
\begin{subequations}
\begin{align}
\Big\langle e^{\A_1 \tcur^t_1 + \A_2 \tcur^t_2}\Big\rangle & = 1 \label{eq:jarzyj}\\
\Big\langle e^{\A_1 \tcur^t_1}\Big\rangle & = \Big\langle e^{\A_2 \tcur^t_2} \Big\rangle  .
\end{align}
\end{subequations}
The latter we call the reciprocal IFR. As a minor side remark, let us now provide a simple remark highlighting that the nonequilibrium reciprocal relations are a fundamental complement to the IFR. Let us rewrite  Eq.\,(\ref{eq:jarzyj}) by conditioning over the second current:
\begin{align}
\Big\langle \Big\langle e^{-\A_1 \tcur^t_1} \Big| \tcur^t_2 \Big\rangle  e^{- \A_2 \tcur^t_2} \Big\rangle = 1.
\end{align}
Now suppose that the currents $\tcur^t_1$ and $\tcur^t_2$ are independent processes. Then
\begin{align}
\Big\langle e^{-\A_1 \tcur^t_1} \Big\rangle \Big\langle e^{- \A_2 \tcur^t_2} \Big\rangle = 1
\end{align}
Notice that this equation alone is not sufficient to grant that each individual current satisfies the IFR itself, unless the reciprocal IFR also holds true.

\subsubsection{Response relations}
\label{subsec:response}

We now consider the response of a system at equilibrium to perturbations of the thermodynamic forces. We assume a thermodynamic parametrization as described in assumption A0 on p.\,\pageref{A0}. While usually only linear response coefficients are considered, higher-order response relations were derived by Andrieux and Gaspard \cite{andrieux2,andrieux3}. The coefficients describing response at order $m = \ssum  m_\alpha$ of cumulants of order $n = \ssum n_\alpha$ are given by
\begin{align}
\lcur_{\{n_\alpha \};\{m_\alpha\}}  & :=  \frac{\partial^{m_1}}{\partial x_1^{m_1}} \dots \frac{\partial^{m_{|\alpha|}}}{\partial x_{|\alpha|}^{m_{|\alpha|}}}    \ccur^{\mathrm{eq}}_{n_1,\ldots,n_{|\alpha|}}  \\
& = (-1)^{n} \prod_\alpha \frac{\partial^{m_\alpha}}{\partial x^{m_\alpha}_\alpha} \frac{\partial^{n_\alpha}}{\partial q_\alpha^{n_\alpha}} \lambda(\{0\};\bs{x}).
\end{align} 
To produce relations between response coefficients, we take $m+n$ total derivatives of the fluctuation symmetry Eq.\,(\ref{eq:lebspo})
\begin{align}
\prod_\alpha \frac{d^{m_\alpha}}{d x^{m_\alpha}_\alpha} \frac{d^{n_\alpha}}{d q_\alpha^{n_\alpha}} \left[ \lambda(\{q_\alpha\};\bs{x}) - \lambda(\{x_\alpha - x_\alpha^{\mathrm{eq}} -q_{\alpha}\};\bs{x})  \right] = 0 
\end{align}
and evaluate at $q_\alpha = 0$, $\bs{x}=\bs{x}^{\mathrm{eq}}$. Notice that the $m$-th derivative with respect to $x$ of a function $f(x) = g(x-x',x)$ is
\begin{align}
\frac{d^{m}f}{dx^m} = \sum_{k =0}^{m} {m \choose k} \left. \frac{\partial^{k}}{\partial y_1^k}\frac{\partial^{m-k}}{\partial y_2^{m-k}} \, g(y_1,y_2) \right|_{\substack{y_1 = x-y \\ y_2 = x ~~~}}.
\end{align} 
We then obtain
\begin{align}
\lcur_{\{n_\alpha \};\{m_\alpha\}}^{\mathrm{eq.}} & = \sum_{\{k_\alpha\}}^{\{m_\alpha\}}  \left[\prod_\alpha  {m_\alpha \choose k_\alpha}  \frac{\partial^{n_\alpha+ k_\alpha}}{\partial q_\alpha^{n_\alpha+k_\alpha}}
 \frac{\partial^{m_\alpha - k_\alpha}}{\partial x^{m_\alpha- k_\alpha}_\alpha} \right] \lambda(\{0\},\bs{x}) \\ \nonumber
 & = 
 \sum_{\{k_\alpha\}}^{\{m_\alpha\}} (-1)^{\ssum (n_\alpha + k_\alpha)} \prod_\alpha {m_\alpha \choose k_\alpha}  \lcur^{\mathrm{eq}}_{\{n_\alpha + k_\alpha\};\{m_\alpha - k_\alpha\}}.
\end{align}
The same set of response coefficients obeys several different response relations for different values of $m_\alpha,n_\alpha$, at fixed order $m_\alpha + n_\alpha$, for all $\alpha$. Then the complete set of response relations at order $\{m_\alpha + n_\alpha\}, \forall \alpha$ is given by
\begin{multline}\label{eq:higher}
 \sum_{\{k_\alpha\} \neq \{0\}}^{\{m_\alpha\}} (-1)^{ k_\alpha } \prod_\alpha {m_\alpha \choose k_\alpha}   \lcur^{\mathrm{eq}}_{\{n_\alpha + k_\alpha\};\{m_\alpha - k_\alpha\}} \\ = \left\{\ba{ll} - 2 \lcur^{\mathrm{eq}}_{\{n_\alpha \};\{m_\alpha\}}, & \mathrm{if} \, \sum_\alpha n_\alpha~\mathrm{odd} \\ 
0, & \mathrm{if}~\sum_\alpha n_\alpha~\mathrm{even}
\ea \right. , \\
\mathrm{for}~ \{ n_\alpha = 1 \,\ldots, m_\alpha \}.
\end{multline}
In the above line we isolated the contribution that comes from setting all $k_\alpha = 0$. Notice that the term on the right-hand side is the lowest-order response coefficient. In general, not all of these response relations are independent.

Let us provide a specific example by compiling tables of coefficients for response relations of a system with two currents, at third order. The upper-left entry labeling the table is $(n_1+m_1,n_2+m_2)$, the other entries are self-explaining:
\begin{center}
\begin{tabular}{c|cccc}
$(3,0)$ & $\ccur_{111}$ & $\lcur_{11\,;\,1}$ & $\lcur_{1\,;\,11}$ \\
\hline 
$dq_1^3$ & 2 & 0 & 0 \\
$dq_1^2 dx_1$ & 0 & 2 & -2  \\
$dq_1 dx_1^2$ & 0 & 0 & 0 \\
$dx_1^3$ & 1 & 3 & -3
\end{tabular}
\end{center}
\begin{center}
\begin{tabular}{c|ccccccc}
$(2,1)$ & $\ccur_{112}$ & $\lcur_{11\,;\,2}$ & $\lcur_{12\,;\,1}$ & $\lcur_{1\,;\,12}$ & $\lcur_{2\,;\,11}$ \\
\hline 
$dq_1^2 dq_2$ & 2 & 0 & 0 & 0 & 0 \\
$dq_1^2 dx_2$ & 1 & 0 & 0 & 0 & 0 \\
$dq_1 dq_2 dx_1$ & -1 & 0 & 0 & 0 & 0 \\
$dq_1 dx_1 dx_2$ & 1 & 1 & 1 & 2 & 0 \\
$dq_2 dx_1^2$  & 1 & 1 & 1 & 0 & 2 \\
$dx_1^2 dx_2$ & -1 & -1 & - 2 & -2 & -1 \\
\end{tabular}
\end{center}
Tables for $(n_1+m_1,n_2+m_2) = (1,2), (0,3)$ can be obtained from the upper two upon switching $1\leftrightarrow 2$. Response relations are vectors in the image of the above matrices. To the best of our knowledge, the number of independent response relations at arbitrary order is not known. 

\section{Episode 2: Single edge current}

\label{sec:singleedge}

Focus is on one observational current supported on one edge. We first provide the central result of our construction, then characterize hidden time-reversal, describe properties of the effective affinity, analyze stalling steady states, and prove certain generalized FRs. Part of this material has been anticipated in Refs.\,\cite{polettiniobs,gili}.

\subsection{Main result}
\label{sec:central}

We concentrate on the time-extensive current $\tcur^t := \tcur^t_{12}$ along edge $\ot$, assuming it is not a bridge (an edge whose removal disconnects the network).   We construct the tilted operator
\begin{align}
\M(q) & :=
\left(\ba{cccccc} - w_{1} &  e^{-q} w_{12} & w_{13} & \ldots  \\
 e^{q} w_{21}&  - w_2 & w_{23}  \\
w_{31} & w_{32} & - w_3  \\
\vdots & & & \ddots 
 \\
 \ea \right) \label{eq:explicit1}
\end{align}
where 
$q$ is the tilting parameter. The Perron-Froebenius eigenvalue $\lambda(q)$ of $\M(q)$ is the SCGF of the current. Notice that $\M(0)= \W$ is the forward generator. The forthcoming propositions introduce and characterize the crucial objects of our study, the hidden time-reversal generator $ \widetilde{\W}$, and the effective affinity $\Q$. 

\begin{theorem}
\label{th:unitary}
There exists a unique nonvanishing value of the tilting parameter $q = \Q$, which we call the effective affinity, and a unique positive-definite diagonal matrix $\Past$, with unit trace $\tr \Past = 1$, such that the following matrix is a {\rm MJPG}
\begin{align}
\widetilde{\W} := \Past \, \M(\Q)^T \,{\Past}^{-1}.\label{eq:uniuni}
\end{align}
\end{theorem}

\begin{theorem}
\label{th:characterization}
The vector $\vec{p}^{\,\st}$ of positive diagonal entries of $\Past$ is the unique normalized steady-state ensemble of the system where edge $\ot$ is removed. The effective affinity is given by
\begin{align}
\Q = \log \frac{w_{12} p_2^\st}{w_{21} p_1^\st} \label{eq:effaff} .
\end{align}
\end{theorem}

Before attacking the proof of these propositions, we point out that this result yields the IFR as a corollary.
\begin{theorem} \label{th:ift1}
The integral fluctuation relation holds
\bea
\lambda(\Q) = 0. \label{eq:marIFT}
\eea
\end{theorem}
\begin{proof}
The tilted operator $\M(q)$ has a unique dominant Perron-Froebenius eigenvalue $\lambda(q)$ that is the SCGF of the current under scrutiny. Since by Eq.\,(\ref{eq:uniuni}) at $q = \Q$ operators $\M(\Q)$ and $\widetilde{\W}$ have the same spectrum and the Perron eigenvalue of $\widetilde{\W}$ is zero, we conclude. \end{proof}

\begin{proof}[Proof of \Th{th:unitary} and \Th{th:characterization}] We have
\begin{align}
\widetilde{\W} =
\left(\ba{cccccc} - w_{1} &  e^{\Q} w_{21} \frac{p^\st_1}{p^\st_2} & w_{31} \frac{p^\st_1}{p^\st_3}  & \hdots \\ 
e^{-\Q} w_{12} \frac{p^\st_2}{p^\st_1} &  - w_2  & w_{32}  \frac{p^\st_2}{p^\st_3}  \\
w_{13}  \frac{p^\st_3}{p^\st_1} & w_{23}\frac{p^\st_3}{p^\st_2} & - w_3  \\
\vdots & & & \ddots \\
 \ea \right).
\end{align}
Since off-diagonal entries are positive, and diagonal entries are negative, for $\widetilde{\W}$ to be a  MJPG we only need to impose that $\vec{1} := (1,1,\ldots,1)^T$ is a left null vector. We obtain the system of equations
\begin{subequations} 
\begin{align}
w_1 & = e^{-\Q}  w_{12} \frac{p^\st_2}{p^\st_1} + \sum_{j>2} w_{1j}  \frac{p^\st_j}{p^\st_1}  \\
w_2 & =  e^{\Q} w_{21} \frac{p^\st_1}{p^\st_2} + \sum_{j>2} w_{2j}\frac{p^\st_j}{p^\st_2} \\
w_i & = \sum_{j \neq i} w_{ij} \frac{p_j^\st}{p_i^\st},  \qquad \qquad \qquad i > 2.
\end{align}
\end{subequations}
Noticing that $w_1 = w_{21} + \sum_{i >2} w_{i1}$ (and similarly for $w_2$), we can rearrange the latter system of equations into
\begin{subequations} \label{eq:eee} 
\begin{align}
\sum_{j >2} \left( w_{j1} p^\st_1 - w_{1j} p^\st_j \right) & = w_{21}p^\st_1 - e^{-\Q}  w_{12} p^\st_2   \\
\sum_{j >2} \left( w_{j2} p^\st_2 - w_{2j} p^\st_j \right) & = w_{12}p^\st_2 - e^{\Q} w_{21} p^\st_1  \\
\sum_{j} \left( w_{\ji} p^\st_i - w_{\ij} p^\st_j \right) & = 0, \qquad \qquad \qquad i > 2.
\label{eq:system}
\end{align}
\end{subequations}
The crucial passage now is to recognize that the left-hand side of this system of equations can be written as $\W_{\mathrm{hid}} \vec{p}^{\,\st}$, where $\W_{\mathrm{hid}}$ is the generator for the MJPG of the network where edge $\ot$ is {\it deleted}
\begin{align}
\W_{\mathrm{hid}} := \left(\ba{ccccc} - w_1 + w_{21} & 0 & w_{13}& \hdots \\
0 & - w_2 + w_{12} & w_{23} \\
w_{31} & w_{32} & - w_3 \\
\vdots & & & \ddots
 \ea \right). \label{eq:deletion}
\end{align}
Since $\vec{1}$ must also be a left-null vector of $\W_{\mathrm{hid}}$, we necessarily have
\begin{align}
w_{21}p^\st_1 - e^{-\Q}  w_{12} p^\st_2 + w_{12}p^\st_2 - e^{\Q} w_{21} p^\st_1  = 0.
\end{align}
The only solutions are the trivial one $\Q = 0$, in which case $\vec{p}^{\,\st}$ solves the steady-state equation $\W\vec{p}^{\,\st} = 0$, and therefore  $\widetilde{\W}$ coincides with the time-reversed generator, and the nontrivial solution anticipated in Eq.\,(\ref{eq:effaff}). In this case Eqs.\,(\ref{eq:eee}) are steady-state equations in the network where the edge $\ot$ is removed, which by assumption is connected, hence there exists a unique positive normalized solution $\vec{p}^{\,\st}$, which we call the stalling steady state.
\end{proof}

\subsection{Dynamics: marginal and hidden}
\label{sec:margen}

We hereby investigate the relationship between generators $\W$ and $\widetilde{\W}$. Given the explicit expression of the effective affinity Eq.\,(\ref{eq:effaff}), the hidden TR generator reads
\begin{align}
\widetilde{\W} & =
\left(\ba{cccccc} - w_{1} & w_{12} & w_{31} \frac{p^\st_1}{p^\st_3}  & \hdots \\ 
w_{21} &  - w_2  & w_{32}  \frac{p^\st_2}{p^\st_3}  \\
w_{13}  \frac{p^\st_3}{p^\st_1} & w_{23}\frac{p^\st_3}{p^\st_2} & - w_3  \\
\vdots & & & \ddots \\
 \ea \right). \label{eq:explicit2}
\end{align}
Notice that the upper $2\times 2$ block is identical to that of the forward generator, while the rest of the matrix resembles a time-reversed generator, see Eq.\,(\ref{eq:timereversed}). Hence we are tempted to speculate that the hidden TR generator preserves the dynamics along edge $\ot$, while it reverses the dynamics in the hidden sector of the network. Let us make this intuition more precise.

We already introduced in Eq.\,(\ref{eq:deletion}) what we call the hidden MJPG $\W_{\mathrm{hid}}$ of the dynamics occurring on network where edge $\ot$ is removed. Let us further consider the {\it marginal}  MJPG that generates the two-configuration dynamics along the sole edge $\ot$:
\begin{align}
 \W_{\mathrm{mar}} := \W - \W_{\mathrm{hid}} =  \left(\ba{cccc} - w_{21} & w_{12} & 0 & \hdots \\
w_{21} & - w_{12} & 0 \\
0  & 0 & 0 \\
\vdots & & & \ddots
 \ea\right).
\end{align}
\begin{theorem}
\label{th:delsingle}
The forward and the hidden time-reversal generators can be expressed as
\begin{subequations}\label{eq:12}
\bea
\W & = \W_{\mathrm{mar}} + \W_{\mathrm{hid}} \\
\widetilde{\W} & = \W_{\mathrm{mar}} +  \overline{\W_{\mathrm{hid}}}.  
\ees
\end{theorem}
\begin{proof}
This obviously follows from the fact that $\vec{p}^{\,\st}$ is the steady state of the hidden dynamics, and by the definition of time-reversed generator Eq.\,(\ref{eq:timereversed}).
\end{proof}
This makes our intuition more precise: the hidden TR generator inverts the dynamics in the hidden sector, but not in the marginal sector.

\subsection{Thermodynamics: the effective affinity}
\label{sec:effaff}

The effective affinity can be interpreted operationally as follows. Prepare the system at the stalling steady state where $\ot$ is removed and suddenly quench the system by connecting edge $\ot$. Then $\Q$ is the thermodynamic force at the instant of connecting. In this section we further characterize the effective affinity and put it in relation to other thermodynamic forces.

Let us first report the expressions for the effective affinity Eq.\,(\ref{eq:effaff}), for the steady-state thermodynamic force $\F := \F_{12}$ along edge $\ot$ Eq.\,(\ref{eq:ssforce}), and for the cycle affinity Eq.\,(\ref{eq:cycleaff}): 
\begin{align}
\F = \log \frac{w_{12} p_2}{w_{21} p_1}, \qquad \Q = \log \frac{w_{12} p^\st_2}{w_{21} p^\st_1}, \qquad \A = \log \frac{w_{12} p^{\mathrm{eq}}_2}{w_{21} p^{\mathrm{eq}}_1}.
\end{align}
This sequence displays the insurgence of a sort of hierarchy of marginal theories: from the most complete one, on the right-end of the spectrum, through the marginal, to the agnostic one on the far left. More on the hierarchy in Sec.\,\ref{sec:hierarchy}.

Notice that in general the steady state $\vec{p}$ of the forward dynamics, and that $\vec{\tilde{p}}$ of the hidden TR dynamics, differ among themselves. Nevertheless, the following result holds.

\begin{theorem}
\label{th:sameforce}
The steady-state edge affinity $\widetilde{\F}$ in the hidden time-reversal generator is identical to that of the forward dynamics,
\bea
\widetilde{\F} = \F.
\eea
\end{theorem}

\begin{proof}
This amounts to prove that
\bea
\frac{p_1}{p_2} = \frac{\widetilde{p}_1}{\widetilde{p}_2}. \label{eq:eqss}
\eea
As we mentioned in Sec.\,\ref{setup}, the steady-state ensemble can be calculated in terms of minors. In view of the explicit form of the generators Eqs.\,(\ref{eq:explicit1},\ref{eq:explicit2}), we consider the minors obtained by removing the first row and column and the second row and column, and notice that
\begin{subequations}\label{eq:delcon}
\begin{align}
\widetilde{\W}{(1|1)} & =  \Past{(1|1)} \, \W{(1|1)} \, {\Past}{(1|1)}^{-1} \\
\widetilde{\W}{(2|2)} & =  \Past{(2|2)} \, \W{(2|2)} \, {\Past}{(2|2)}^{-1}.
\end{align}
\end{subequations}
where we remind that matrix $\matrix{A}(i_1,\ldots,i_n|j_1,\ldots,j_m)$ is obtained by removing rows $i_1,\ldots,i_n$ and columns $j_1,\ldots,j_m$.
Since these are matrix similarities their determinants coincide.
\end{proof}
While the ratio of the populations at sites $1$ and $2$ stays the same upon hidden TR, the total density might change due to normalization, which is a nonlocal factor; for this reason the steady-state currents of the forward and hidden TR current along edge $\ot$ are not the same. However, the following result is an obvious corollary of the above proposition, corroborating the intuition expressed above that the hidden TR dynamics ``tends to preserve'' the dynamics in the marginal sector.

\begin{theorem}
\label{th:currentverse}
The steady-state mean currents  $\cur_{12}$ and $\widetilde{\cur}_{12}$ respectively of the forward and of the hidden TR dynamics have the same sign.
\end{theorem}

We will now compare the ``real'' affinities of the forward and the marginal TR dynamics.

\begin{theorem}
\label{th:affaff} For any simple cycle $\C$, let $\A(\C)$ and $\widetilde{\A}(\C)$ denote the cycle affinities with respect to the forward and hidden TR dynamics, respectively. We have
\begin{subequations}
\begin{align}
\widetilde{\A}(\C) &= - \A(\C) , & & \mathrm{if}\,  \C \niton \ot  \\
\widetilde{\A}(\C) &= - \A(\C) + 2\Q , & & \mathrm{if}\, \C \owns \ot.
\end{align}
\end{subequations}
\end{theorem}

\begin{proof}
For all cycles that do not contain edge $\ot$, one has
\begin{align}
\widetilde{\A}(\C) & = \log \prod_{\ij \in \C}  \frac{\widetilde{w}_{\ij}}{\widetilde{w}_{\ji}} = \log \prod_{\ij \in \C}  \frac{w_{\ji}  p^\st_i/p^\st_j }{w_{\ij} p^\st_j/p^\st_i} .
\end{align}
All terms $p^\st_i$ cancel out one with each other along cycles, so the last term can be identified with $-\A(\C)$. As regards cycles containing edge $\ot$, one has
\begin{align}
\widetilde{\A}(\C) & =\log \prod_{\substack{\ij \in \C \\ \ij \neq 12}}    \frac{w_{\ji}  p^\st_i/p^\st_j }{w_{\ij} p^\st_j/p^\st_i}  + \log \frac{w_{12}}{w_{21}} \\
& = \log \prod_{\ij \in \C}  \frac{w_{\ji}  (p^\st_i)^2}{w_{\ij} (p^\st_j)^2} - \log \frac{w_{21}  (p^\st_1)^2}{w_{12} (p^\st_2)^2} + \log \frac{w_{12}}{w_{21}}
\end{align}
and we conclude.
\end{proof}

We will now employ the expression of the steady state in terms of spanning trees to compare edge and effective affinities. In Sec.\,\ref{sec:graphical} we provide a graphical representation by an example. Let us remind the definition of fluxes in the full network introduced in \S\,\ref{subsec:flux} and introduce fluxes in the network where $\ot$ is either deleted or contracted, that is, where it is shrunk to a unique site (here we set for simplicity the normalization $\tau(\G) \stackrel{!}{=} 1$).
\begin{subequations}
\begin{align}
\psi_{12} & = w_{12} \, \tau_2(\G) & & = w_{12} \det \W(2|1) \\
\psi_{\setminus 12} & := w_{12} \, \tau_2(\G\setminus \ot) & & = w_{12} \det\W_{\mathrm{hid}}(2|1)  \\ 
\psi_{/ 12} & := w_{12}w_{21} \, \tau_1(\G / \ot) & & = w_{12}w_{21} \det \W{(1,2|1,2)} .
\end{align}
\end{subequations}
Similar definitions are obtained by interchanging $1$ and $2$. Notice that $\psi_{/ 12} = \psi_{/ 21}$. We can then write:
\begin{align}
\Q = \log \frac{\psi_{\setminus 12}}{\psi_{\setminus 21}}.
\end{align}
We can also consider spanning trees in the graph $\G/\ot$ where edge $\ot$ is contracted (sites $1$ and $2$ are identified). The following deletion-contraction formula holds 
\begin{align}
\tau_1(\G) = \tau_1(\G\setminus \ot) + w_{12} \, \tau_2(\G/\ot). \label{eq:delconfor}
\end{align}
which in terms of determinants reads
\begin{align}
\det \W{(2|1)} = \det \W_{\mathrm{hid}}{(2|1)} + w_{12} \det \W{(1,2|1,2)}. \label{eq:delcondet}
\end{align}
Basically, this formula states that all spanning trees (on the left) either do not contain edge $\ot$ (first term on the right) or they do (second term on the rights). If they do, they need to be oriented towards the contracted vertices. Notice that because in the contracted graph $1$ is identified with $2$, clearly $\T_2(\G/\ot) = \T_1(\G/\ot)$. Then we find that
\begin{align}
0 \leq \psi_{/ 12} = \psi_{12} - \psi_{ \setminus 12} = \psi_{21} - \psi_{\setminus 21},
\end{align}
that is, the difference between the fluxes is independent of the direction of the edge. Before giving a graphical example, we are now in the position to prove the following result.

\begin{theorem}
\label{th:bigger}
The effective affinity has larger modulus than the edge affinity,
\begin{align}
|\Q| \geq |\F|.
\end{align}
\end{theorem}

\begin{proof}
Suppose $\psi_{12} > \psi_{12}$, whence both $\Q > 0$ and $\F > 0$. We have 
\begin{align}
\F & = \log \frac{\psi_{12}}{\psi_{21}} =  \log \frac{\psi_{\setminus 12} +\psi_{/ 12} }{\psi_{\setminus 21} + \psi_{/12}} \\
& = \Q +\log \frac{1+ \frac{\psi_{/12}}{\psi_{\setminus 12}}}{1+ \frac{\psi_{/12}}{\psi_{\setminus 21}}}.
\end{align}
Because $\psi_{\setminus 12} \geq \psi_{\setminus 21}$, then the latter term above is negative and thus we conclude.
\end{proof}

\subsection{Graphical representation}
\label{sec:graphical}

Let us consider the graph:
\begin{align}
\xymatrix{ 1  \ar@{-}[r] & 2 \ar@{-}[d] \\ 
4 \ar@{-}[ur] \ar@{-}[u] & 3 \ar@{-}[l]}
\end{align}
In the following, each directed edge means ``multiply by the corresponding rate''. The oriented spanning-tree polynomials with root in $i =1,2$ are given by
\bes
\tau_1(\G)
& = &\Scale[0.7]{
  \ba{c}\xymatrix{ \ar@{<-}[r] &  \ar@{<-}[d] \\  \ar[u] & } \ea
+ \ba{c}\xymatrix{ \ar@{<-}[r] & \ar@{<-}[d] \\ & \ar@{<-}[l]} \ea
+ \ba{c}\xymatrix{ \ar@{<-}[r] &  \\ \ar[u] & \ar[l]} \ea
+ \ba{c}\xymatrix{ \ar@{<-}[r]  & \\ \ar[ur]  & \ar[l]} \ea
+ \ba{c}\xymatrix{ \ar@{<-}[r]  & \ar@{<-}[d] \\ \ar[ur] & } \ea 
}
 \nonumber \\
& & \Scale[0.7]{
+ \ba{c}\xymatrix{ & \ar[d] \\ \ar[u] & \ar[l]} \ea 
+ \ba{c}\xymatrix{ & \ar@{<-}[d] \\ \ar@{<-}[ur] \ar[u] & } \ea
+ \ba{c}\xymatrix{   & \\ \ar@{<-}[ur] \ar[u] & \ar[l]} \ea
}, \\
\tau_2(\G)
& = & \Scale[0.7]{\ba{c}\xymatrix{ \ar@{->}[r] &  \ar@{<-}[d] \\  \ar[u] & } \ea
+ \ba{c}\xymatrix{ \ar@{->}[r] & \ar@{<-}[d] \\ & \ar@{<-}[l]} \ea
+ \ba{c}\xymatrix{ \ar@{->}[r] &  \\ \ar[u] & \ar[l]} \ea 
+ \ba{c}\xymatrix{ \ar@{->}[r]  & \\ \ar[ur]  & \ar[l]} \ea + \ba{c}\xymatrix{ \ar@{->}[r]  & \ar@{<-}[d] \\ \ar[ur] & } \ea
} \nonumber \\
& & \Scale[0.7]{
+ \ba{c}\xymatrix{ & \ar@{<-}[d] \\ \ar@{<-}[u] & \ar@{<-}[l]} \ea
+ \ba{c}\xymatrix{ & \ar@{<-}[d] \\ \ar@{->}[ur] \ar@{<-}[u] & } \ea
+ \ba{c}\xymatrix{   & \\ \ar@{->}[ur] \ar@{<-}[u] & \ar[l]} \ea
}.
\ees
Let us now consider the graph with deleted edge $\G\setminus \ot$:
\begin{align}
\xymatrix{ 1 & 2 \ar@{-}[d] \\ 
4 \ar@{-}[ur] \ar@{-}[u] & 3 \ar@{-}[l]} .
\end{align}
We have for the unnormalized steady state
\bes
\tau_1(\G\setminus \ot)
& = &  \Scale[0.7]{
\ba{c}\xymatrix{ & \ar[d] \\ \ar[u] & \ar[l]} \ea 
+ \ba{c}\xymatrix{ & \ar@{<-}[d] \\ \ar@{<-}[ur] \ar[u] & } \ea
+ \ba{c}\xymatrix{   & \\ \ar@{<-}[ur] \ar[u] & \ar[l]} \ea
} \\
\tau_2(\G\setminus \ot)
& = & \Scale[0.7]{
\ba{c}\xymatrix{ & \ar@{<-}[d] \\ \ar@{<-}[u] & \ar@{<-}[l]} \ea
+ \ba{c}\xymatrix{ & \ar@{<-}[d] \\ \ar@{->}[ur] \ar@{<-}[u] & } \ea
+ \ba{c}\xymatrix{   & \\ \ar@{->}[ur] \ar@{<-}[u] & \ar[l]} \ea
}
\ees
Multiplying respectively by $w_{21}$  and $w_{12}$ one obtains Hill's fluxes, yielding for the edge affinity
\begin{align}
\F 
= \log\frac
{
\Scale[0.5]{
\ba{c}\xymatrix{ \ar@{<-}@/_.3pc/[r] \ar@{->}@/^.3pc/[r] &  \ar@{<-}[d] \\  \ar[u] & }\ea
+ \ba{c}\xymatrix{ \ar@{<-}@/_.3pc/[r] \ar@{->}@/^.3pc/[r] & \ar@{<-}[d] \\ & \ar@{<-}[l]} \ea
+ \ba{c}\xymatrix{\ar@{<-}@/_.3pc/[r] \ar@{->}@/^.3pc/[r]  &  \\ \ar[u] & \ar[l]} \ea
+ \ba{c}\xymatrix{\ar@{<-}@/_.3pc/[r] \ar@{->}@/^.3pc/[r]   & \\ \ar[ur]  & \ar[l]} \ea
+ \ba{c}\xymatrix{ \ar@{<-}@/_.3pc/[r] \ar@{->}@/^.3pc/[r] & \ar@{<-}[d] \\ \ar[ur] & } \ea + \ba{c}\xymatrix{ \ar@{<-}[r] & \ar@{<-}[d] \\ \ar@{<-}[u] & \ar@{<-}[l]} \ea
+ \ba{c}\xymatrix{ \ar@{<-}[r]  & \ar@{<-}[d] \\ \ar@{->}[ur] \ar@{<-}[u] & } \ea
+ \ba{c}\xymatrix{  \ar@{<-}[r]  & \\ \ar@{->}[ur] \ar@{<-}[u] & \ar[l]} \ea} 
}{
\Scale[0.5]{
\ba{c}\xymatrix{ \ar@{<-}@/_.3pc/[r] \ar@{->}@/^.3pc/[r] &  \ar@{<-}[d] \\  \ar[u] & } \ea
+ \ba{c}\xymatrix{ \ar@{<-}@/_.3pc/[r] \ar@{->}@/^.3pc/[r] & \ar@{<-}[d] \\ & \ar@{<-}[l]} \ea
+ \ba{c}\xymatrix{\ar@{<-}@/_.3pc/[r] \ar@{->}@/^.3pc/[r]  &  \\ \ar[u] & \ar[l]} \ea
+ \ba{c}\xymatrix{\ar@{<-}@/_.3pc/[r] \ar@{->}@/^.3pc/[r]   & \\ \ar[ur]  & \ar[l]} \ea
+ \ba{c}\xymatrix{ \ar@{<-}@/_.3pc/[r] \ar@{->}@/^.3pc/[r] & \ar@{<-}[d] \\ \ar[ur] & } \ea
+ \ba{c}\xymatrix{ \ar@{->}[r] & \ar[d] \\ \ar[u] & \ar[l]} \ea
+ \ba{c}\xymatrix{ \ar@{->}[r]  & \ar@{<-}[d] \\ \ar@{<-}[ur] \ar[u] & } \ea
+ \ba{c}\xymatrix{  \ar@{->}[r]  & \\ \ar@{<-}[ur] \ar[u] & \ar[l]} \ea} }
\end{align}
and for the effective affinity
\begin{align}
\Q = \log \frac{\Scale[0.5]{\ba{c}\xymatrix{ \ar@{<-}[r] & \ar@{<-}[d] \\ \ar@{<-}[u] & \ar@{<-}[l]} \ea
+ \ba{c}\xymatrix{ \ar@{<-}[r]  & \ar@{<-}[d] \\ \ar@{->}[ur] \ar@{<-}[u] & } \ea
+ \ba{c}\xymatrix{  \ar@{<-}[r]  & \\ \ar@{->}[ur] \ar@{<-}[u] & \ar[l]} \ea} }{
\Scale[0.5]{\ba{c}\xymatrix{ \ar@{->}[r] & \ar[d] \\ \ar[u] & \ar[l]} \ea
+ \ba{c}\xymatrix{ \ar@{->}[r]  & \ar@{<-}[d] \\ \ar@{<-}[ur] \ar[u] & } \ea
+ \ba{c}\xymatrix{  \ar@{->}[r]  & \\ \ar@{<-}[ur] \ar[u] & \ar[l]} \ea }}
\end{align}
Notice that a ``real'' cycle affinity with respect for example to spanning tree $\Scale[0.3]{\ba{c}\xymatrix{ & \ar@{-}[d] \\ \ar@{-}[u] & \ar@{-}[l]} \ea }$ is
\begin{align}
\A = \log \frac{\Scale[0.5]{\ba{c}\xymatrix{ \ar@{<-}[r] & \ar@{<-}[d] \\ \ar@{<-}[u] & \ar@{<-}[l]} \ea}}{\Scale[0.5]{\ba{c}\xymatrix{ \ar@{->}[r] & \ar[d] \\ \ar[u] & \ar[l]} \ea}}
\end{align}
which gives a nice hierarchy of decreasing cycle contributions. Finally, consider the contracted graph $\G/\ot $
\begin{align}
\xymatrix{ 1  \ar@{..}[r] & 2 \ar@{-}[d] \\ 
4 \ar@{-}[ur] \ar@{-}[u] & 3 \ar@{-}[l]} .
\end{align}
where the dotted line signifies identification of sites. We have
\begin{align}
\tau_1(\G/\ot) = \Scale[0.7]{
  \ba{c}\xymatrix{ \ar@{..}[r] &  \ar@{<-}[d] \\  \ar[u] & } \ea
+ \ba{c}\xymatrix{ \ar@{..}[r] & \ar@{<-}[d] \\ & \ar@{<-}[l]} \ea
+ \ba{c}\xymatrix{ \ar@{..}[r] &  \\ \ar[u] & \ar[l]} \ea
+ \ba{c}\xymatrix{ \ar@{..}[r]  & \\ \ar[ur]  & \ar[l]} \ea
+ \ba{c}\xymatrix{ \ar@{..}[r]  & \ar@{<-}[d] \\ \ar[ur] & } \ea 
}
\end{align}
and we obtain
\begin{align}
\det \W(1,2|1,2) = \Scale[0.7]{ \ba{c}\xymatrix{ \ar@{<-}@/_.3pc/[r] \ar@{->}@/^.3pc/[r] &  \ar@{<-}[d] \\  \ar[u] & }\ea
+ \ba{c}\xymatrix{ \ar@{<-}@/_.3pc/[r] \ar@{->}@/^.3pc/[r] & \ar@{<-}[d] \\ & \ar@{<-}[l]} \ea
+ \ba{c}\xymatrix{\ar@{<-}@/_.3pc/[r] \ar@{->}@/^.3pc/[r]  &  \\ \ar[u] & \ar[l]} \ea
+ \ba{c}\xymatrix{\ar@{<-}@/_.3pc/[r] \ar@{->}@/^.3pc/[r]   & \\ \ar[ur]  & \ar[l]} \ea
+ \ba{c}\xymatrix{ \ar@{<-}@/_.3pc/[r] \ar@{->}@/^.3pc/[r] & \ar@{<-}[d] \\ \ar[ur] & } \ea }.
\end{align}
One can then verify the deletion-contraction relation and the validity of \Th{th:bigger}.

\subsection{Marginal entropy production rate}
\label{sec:marepr}

We call the quantity
\begin{align}
\sigma_{1} := \cur \, \Q \geq 0
\end{align}
the mean {\it marginal EPR}. The superscript marks the observable current $\phi = \phi_1$. We have shown in Ref.\,\cite{polettiniobs} that it can be interpreted as the mean EPR estimated by an observer who formulates a minimal Markovian model accounting for the observable steady-state mean current. Let us review the argument. We consider an observer who controls $w_{12}$, $w_{21}$, and the steady-state ratio $p_1/p_2$. To him, the rest of the system is a black box:
\begin{align}
\begin{array}{c}\xymatrix{1 \ar@{-}[rr] \ar@{-}[dr] \ar@{..}@/_/[dr] \ar@{-}@/^/[dr] & & 2  \ar@{..}@/^/[dl]  \ar@{-}@/_/[dl]  \ar@{-}[dl] \\
& \rule{1cm}{1cm}  & 
}\end{array}.
\end{align}
The observer needs to formulate a minimal model that is compatible with the observation of a current along $\ot$. The simplest possible setup is
\begin{align}
\ba{c}\xymatrix{1 \ar@{-}@/^1pc/[r]^{w}   \ar@{-}@/_1pc/[r]_{\widetilde{w}}  & 2} \ea.
\end{align}
In this minimal model the black box is responsible of returning an event at $1$ or $2$ at some effective rates $\widetilde{w}_{12},\widetilde{w}_{21}$, which we require to be independent of $w_{12}, w_{21}$. These effective rates differ from the one considered e.g. by Uhl et al. in Ref.\,\cite{uhl}. Notice that the effective model has only one effective cycle
\begin{align}
\C^{\mathrm{eff}} = \ba{c}\xymatrix{1 \ar@{<-}@/^1pc/[r]   \ar@{->}@/_1pc/[r]  & 2} \ea.
\end{align}

\begin{theorem}
The affinity of the effective cycle is equal to the effective affinity,
\begin{align}
\A(\C^{\mathrm{eff}}) = \Q.
\end{align}
\end{theorem}

\begin{proof}
The topology of the minimal model is such that, lumping together the transitions  $w$ and $\widetilde{w}$, global detailed balance must hold:
\begin{align}
\frac{w_{21}+ \widetilde{w}_{21}}{w_{12} + \widetilde{w}_{12}} = \frac{p_2}{p_1}. \label{eq:condition}
\end{align}
If we require this condition to hold independently of $w_{12}, w_{21}$, using Eq.\,(\ref{eq:delcondet}) we obtain for the effective rates
\begin{subequations}
\begin{align}
\widetilde{w}_{12} & = \frac{\det \W_{\mathrm{hid}}{(2,1)}}{\det \W{(1,2|1,2)}} \\
\widetilde{w}_{21} & = \frac{\det \W_{\mathrm{hid}}{(1,2)}}{\det \W{(1,2|1,2)}} .
\end{align}
\end{subequations}
The cycle affinity of the effective cycle $\C^{\mathrm{eff}}$ is given by
\begin{align}
\A(\C^{\mathrm{eff}}) = \log \frac{w_{12}\widetilde{w}_{21}}{w_{21}\widetilde{w}_{12}}
\end{align}
and by replacing the values of the effective rates one concludes.
\end{proof}

It is interesting to notice that, in view of Eq.\,(\ref{eq:eqss}), the  ideal observer is completely blind to hidden TR. 

Let us now show that such observer always underestimates the mean EPR and, even more strictly, the local mean EPR associated to Hill's cycle decomposition Eq.\,(\ref{eq:hillEPR}). A different derivation was proposed in Ref.\,\cite{gili}.

\begin{theorem}\label{th:fullmar}
The mean marginal EPR is positive, and it is always smaller than the Hill mean marginal EPR and of the complete mean EPR,
\begin{align}
0 \leq \sigma_{1} \leq  \sigma^{\mathrm{Hill}}_{12} \leq \sigma.
\end{align}
\end{theorem}

\begin{proof}
Positivity follows from the explicit expression
\begin{align}
\sigma_{1}  = (\psi_{\setminus 12} - \psi_{\setminus 21}) \log \frac{\psi_{\setminus 12}}{\psi_{\setminus 21}},
\end{align}
whereby the linear and the logarithmic terms are either both positive or both negative. 

To show that the marginal EPR underestimates the full EPR,  we employ the cycle decomposition of the entropy production rate Eq.\,(\ref{eq:hillepr}). On the one hand we have
\begin{align}
\sigma = \left( \sum_{\C \owns 12} +  \sum_{\C \niton 12} \right)  (\psi_{+\C} - \psi_{-\C}) \log \frac{\psi_{+\C}}{\psi_{-\C}},
\end{align}
where we remind that $-\C$ is the cycle with reversed orientation. Notice that each of the summands is positive. Hence we can focus on just the cycles that pass through $\ot$:
\begin{align}
\sigma & \geq \sum_{\C \owns 12}  \left(\psi_{+\C} - \psi_{-\C}\right) \log \frac{\psi_{+\C}}{\psi_{-\C}} \nonumber
\\ & \geq \left[ \sum_{\C \owns 12}  \left(\psi_{+\C} - \psi_{-\C} \right)\right]\log \frac{\sum_{\C \owns 12} \psi_{+\C}}{ \sum_{\C \owns 12} \psi_{-\C}} .
\end{align}
The first expression is Hill's mean EPR along edge $\ot$; we then used the log-sum inequality for the logarithm, an usual tool in matters of coarse-graining in information theory \cite{cover,espositocg}. In the latter expression we recognize the mean marginal EPR.
\end{proof}
 
\subsection{Stalling}
\label{sec:stalling}

The following two results give a connection between the thermodynamic and the kinematic interpretation of stalling.

\begin{theorem}
\label{th:stalling1}
The steady-state mean current $\cur$ vanishes if and only if the effective affinity $\Q$ vanishes.
\end{theorem}
\begin{proof}
If we prepare the system in state $\vec{p}^{\,\st}$ and then transition $\ot$ is turned on abruptly by a quench, initially we have 
\bea
\W \, \vec{p}^{\,\st} = \widetilde{\W} \, \vec{p}^{\,\st} = \left(\ba{c} w_{12} p^\st_2 - w_{21} p^\st_1 \\ 
w_{21} p^\st_1 - w_{12} p^\st_2 \\
0 \\
\vdots \\
0
\ea\right). \label{eq:www}
\eea
This  follows from the fact that $\vec{p}^{\,\st}$ is the steady state of both $\W_{\mathrm{hid}} $ and $\overline{\W_{\mathrm{hid}}}$. Then the current vanishes if and only if and the effective affinity vanishes.
\end{proof}

\begin{theorem}
At stalling, the hidden TR generator coincides with the time-reversed generator
\begin{align}
\widetilde{\W} = \overline{\W}.
\end{align}
\end{theorem}

\begin{proof}
We notice that at stalling, because $\vec{p}^{\,\st}$ is the steady state and the effective affinity vanishes, then $\W_{\mathrm{mar}} = \overline{\W_{\mathrm{mar}}}$. The conclusion follows from Eqs.\,(\ref{eq:delcon}).
\end{proof}

The latter can be interpreted as a condition of marginal detailed balance.

\subsection{Response at stalling}

In this section we look at what happens in the vicinity of a stalling steady state. Let us reintroduce the index $\phi = \phi_{1}$ denoting that this is the first observable current of a theory with only $|\mu|=1$ observable currents (crucacially, the reader should not confuse this index with the site $i=1$).  We introduce a parametrization of the rates along the observable edge $w_{12},w_{21} \to w_{12}(x),w_{21}(x)$ such that
\begin{align}
\frac{w_{12}(x)}{w_{21}(x)} = \exp x
\end{align}
and no other rate depends on $x$. We will discuss in Sec.\,\ref{parametrizations} how general such a parametrization is. Let us make explicit the dependence of the SCGF on $x$, $\lambda(q) \to \lambda(q;x)$. The effective affinity now takes the form
\begin{align}
\Q(x) = x - x^\st
\end{align}
where 
\begin{align}
x^\st := \log \frac{p_2^\st}{p_1^\st}.
\end{align}
The above expression for the effective affinity allows for simple phenomenological interpretation: in fact, as far as the parameter is thermodynamic in the sense discussed in the Parode, the effective affinity can be determined as the distance of the parameter from the value that makes the current stall, which can in principle be found by a simple tuning procedure.

The IFR Eq.\,(\ref{eq:marIFT}) now reads:
\begin{align}
\lambda(x-x^\st;x) = 0.
\end{align}
We now focus to second-order expansion in the vicinity of stalling. 
Taking the first total derivative of the marginal IFR with respect to $x$ we recover the fact that if the effective affinity vanishes, then the current vanishes:
\begin{align}
0 = \left. \frac{d}{dx} \lambda(x-x^\st;x)\right|_{x=x^\st} = \left[ \frac{\partial \lambda}{\partial q} + \frac{\partial \lambda}{\partial x} \right] (0;x^\st)= \cur^\st_1.
\end{align}
Similarly, taking the second derivative and evaluating at stalling we rederive the FDR at stalling found in Ref.\,\cite{altaner16}:
\begin{align}
\left.\frac{\partial \cur(x)}{\partial x}\right|_{x=x^\st} = \frac{1}{2} \frac{\partial^2 \lambda}{\partial q^2}(0;x^\st) = \frac{1}{2} \ccur_{11}^{\st}. \label{eq:fdr1}
\end{align}

We can also provide an an explicit expression for the variance at stalling:
\begin{theorem}
The current's variance at stalling is given by
\begin{align}
\ccur_{11}^{\st} = 2 w_{12}(x^\st)p_2^\st = 2 w_{21}(x^\st)p_1^\st . \label{eq:formofC}
\end{align}
\end{theorem}

\begin{proof}
We notice that the determinant of the tilted operator has the functional form
\begin{align}
\det \M(q) = w_{12} \det \W(1|2) e^{-q} + w_{21} \det \W(2|1) e^{q} + \mathrm{const.}  \label{eq:ciaociao}
\end{align}
On the other hand, it is the product of the eigenvalues $\det \M(q) = \prod_{i=1}^{|\I|} \lambda_i(q)$, including $\lambda_1 = \lambda$ the SCGF. Taking the first derivative
\begin{align}
\frac{\partial}{\partial q} \det \M = \frac{\partial \lambda}{\partial q} \prod_{i \neq 1} \lambda_i(q) + \lambda(q) \sum_{i \neq 1} \frac{\partial \lambda_i}{\partial q} \prod_{j \neq 1,i} \lambda_j(q).
\end{align}
Evaluating at $q=0$ where $\lambda(0) = 0$, using the fact that $\phi = - \partial_q \lambda(0)$ and comparing with Eq.\,(\ref{eq:ciaociao}) we obtain
\begin{align}
\phi_1 = \frac{w_{12} \tau_2(\G) - w_{21} \tau_1(\G)}{ \prod_{i \neq 1} \lambda_i(q)}. 
\end{align}
This shows that $\prod_{i \neq 1} \lambda_i(q) = \tau(\G)$, the product of the nonvanishing eigenvalues is the normalization factor of the steady state.

Now, taking the second derivatives and evaluating at $q= 0$ we obtain
\begin{align}
w_{12} p_2  + w_{21}  p_1 =  \ccur_{11}- 2\phi_1 \frac{\partial }{\partial q}\log  \prod_{i \neq 1}   \lambda_i(0) 
\end{align}
and from this we easily conclude by evaluating at stalling.
\end{proof}

We notice that a similar expression holds in the linear regime near equilibrium, and that this is the basis for the linear-regime analysis of Schnakenberg, see the latest paragraph in Ref.\,\cite{schnak}. It might therefore be possible to attempt a similar marginal linear-regime analysis.

\subsection{Fluctuation relations for marginal currents}
\label{sec:fr2}

In this section and the following we prove the FR announced in Eq.\,(\ref{eq:marfparode}) by a direct method involving path probabilities. In particular we consider the probability associated to the hidden TR process, which in the light of Eq.\,(\ref{eq:density}) reads
\begin{align}
\widetilde{\prob}(\traj) =  \tilde{p}^0_{i_0}  e^{-w_{i_N} t_N} \prod_{n=0}^{N-1} \widetilde{w}_{i_{n+1},i_n} \, e^{-w_{i_n} t_n}.\label{eq:densityhidden TR}
\end{align}
where $\traj$ is the path described by Eq.\,(\ref{eq:traj}), and $\vec{\tilde{p}}^{\,0}$ is the ensemble from which the initial configuration is sampled. We notice that all exit rates out of sites are the same as for the forward dynamics.

\begin{theorem}\label{th:effFR}
The marginal FR
\begin{align}
\frac{\prob(\tcur)}{\widetilde{\prob}(-\tcur)} = \exp {\Q\, \tcur}   \label{eq:marfrfr}
\end{align}
holds at all times, provided the initial configuration of all processes is sampled from the stalling state $\vec{p}^{\,\st}$.
\end{theorem}

\begin{proof}
Let us consider the ratio of the probability densities
\begin{align}
\frac{\prob(\invtraj)}{\widetilde{\prob}(\traj)} = \frac{\overline{p}_{i_N}}{\tilde{p}^0_{i_0}} \prod_{n=0}^{N-1} \frac{w_{i_n,i_{n+1}}}{\widetilde{w}_{i_{n+1},i_n}} =  \frac{\overline{p}_{i_N}}{\tilde{p}^0_{i_0}}  \prod_{i,j} \left(  \frac{w_{\ji}}{\widetilde{w}_{\ij}} \right)^{\Psi^t_{\ij}}
\end{align}
where we expressed the path probability in terms of the time-integrated fluxes. We now single out the rates corresponding to transition $\ot$ and use the explicit expression of the rates:
\begin{align}
\frac{\prob(\invtraj)}{\widetilde{\prob}(\traj)} & = \frac{\overline{p}_{i_N}}{\tilde{p}^0_{i_0}}   \left(  \frac{w_{12}}{\widetilde{w}_{21}} \right)^{\Psi^t_{12}}  \left(  \frac{w_{21}}{\widetilde{w}_{12}} \right)^{\Psi^t_{21}} \prod_{(i,j) \neq (1,2),(2,1)} \left(  \frac{w_{\ji}}{\widetilde{w}_{\ij}} \right)^{\Psi^t_{\ij}} \nonumber \\
& = \frac{\overline{p}_{i_N}}{\tilde{p}^0_{i_0}} \left(  \frac{w_{21}}{w_{12}} \right)^{\tcur^t} \prod_{\ij \neq 12} \left(  \frac{p_j^{\st}}{ p_i^{\st} }\right)^{\Phi^t_{\ij}} ,
\end{align}
where in the latter passage we used $\tcur^t_{\ij} = \Phi^t_{\ij} - \Phi^t_{\ji}$ and moved from fluxes to currents using antisymmetry. We now multiply and divide by $( p_2^{\st} / p_1^{\st} )^{\Phi^t}$ to obtain
\begin{align}
\frac{\prob(\invtraj)}{\widetilde{\prob}(\traj)} & = \frac{\overline{p}_{i_N}}{\tilde{p}^0_{i_0}} \left(  \frac{w_{21} p_1^\st}{w_{12} p_2^\st} \right)^{\tcur^t} \prod_{\ij} \left(  \frac{p_j^{\st}}{ p_i^{\st} }\right)^{\Phi^t_{\ij}} \nonumber  \\
& = \frac{\overline{p}_{i_N}}{\tilde{p}^0_{i_0}} \left(  \frac{w_{21} p_1^\st}{w_{12} p_2^\st} \right)^{\tcur^t} 
\exp \sum_{n=0}^{N-1} \log \frac{p_{i_n}^{\st}}{ p_{i_{n+1}}^{\st} }.
\end{align}
In the last passage we moved back from an expression in terms of the fluxes to one in terms of the jumping events. Now, being $ \log  p_i^{\st} / p_j^\st$ an exact discrete differential, terms cancel out and its sum along a path is just the difference between the initial and final values:
\begin{align}
\prob(\invtraj) & = \frac{\overline{p}_{i_N} p_{i_0}^{\st}}{\tilde{p}^0_{i_0} p_{i_{N}}^{\st}} \, e^{-\Q\, \tcur^t } \; \widetilde{\prob}(\traj). \label{eq:123}
\end{align}
Sampling the initial site for all processes from the stalling steady state, the finite-time prefactor cancels out. We can now integrate over all trajectories that give current $\tcur^t \equiv \tcur$, obtaining
\begin{align}
\prob(- \tcur) & = e^{-\Q\, \tcur} \widetilde{\prob}(\tcur), \label{eq:preifr}
\end{align}
which is the desired result (up to $\tcur \to -\tcur$).\end{proof}

Notice that integrating Eq.\,(\ref{eq:preifr}) with respect to $d\Phi$ immediately yields the marginal integral FR
\begin{align}
\left\langle  e^{-\Q\, \tcur^t} \right\rangle = 1.
\end{align}
Here, the probability density of the hidden TR dynamics is traced out, so that it does not play a role for the marginal IFR, a common trick that was employed e.g. in \cite{falasco} to obtain IFRs for time-symmetric observables.

According to the lines of Ref.\,\cite{gili}, we now consider the {\it hidden entropy production},
\begin{align}
\Sigma^t_{1} := \sum_\alpha \A_\alpha \,\tcur^t_\alpha - \Q \, \tcur^t,
\end{align}
\begin{theorem}
where again the subscript $1$ refers to the fact that we are consider the marginal theory w.r.t. current $\Phi^t= \Phi_1^t$. The following effective FR and the the IFR for the hidden entropy production hold asymptotically at long times:
\begin{align}
\frac{\prob(\Sigma_1)}{\widetilde{\prob}(\Sigma_1)} & = e^{-\Sigma_{1}} , & 
\left\langle  e^{\Sigma^t_{1}} \right\rangle & = 1.
\end{align}
\end{theorem}
\begin{proof}
Putting together Eqs.\,(\ref{eq:123}) and (\ref{eq:transientFR}), we obtain
\begin{align}
\frac{\prob(\traj)}{\widetilde{\prob}(\traj)} =  \frac{p_{i_0}^{\st}p^0_{i_0}}{\tilde{p}^0_{i_0} p_{i_{N}}^{\st}} \frac{p^{\mathrm{eq}}_{i_N}}{p^{\mathrm{eq}}_{i_0}} \exp \Sigma^t_{\mathrm{hid}}.
\end{align}
Notice that now we cannot make a physically relevant choice of initial ensemble to cancel the prefactor. Therefore we need to go to long times. We obtain the two relations by the usual manipulations.
\end{proof}
Notice that the IFR for the hidden entropy production implies $\sigma_1 \geq 0$ and $\sigma \geq \sigma_1$ for their mean time-scaled values, which we proved in \Th{th:fullmar} by a different route.

\subsection{Fluctuation relations for marginal cycle currents} 
\label{sec:fr3}

As discussed in Sec.\,\ref{setup}, the FR for the currents is intertwined with the notion of time-reversal of a trajectory, see Eqs.\,(\ref{eq:equipp}) and (\ref{eq:trtraj}). Time-reversal of paths is involutive, one-to-one, it inverts all of the currents and maintains the time-symmetric properties, i.e. the waiting probabilities at sites and the activities $\Psi^t_{\ij}+ \Psi^t_{ji}$.

The above relation Eq.\,(\ref{eq:marfrfr}) cannot be considered as a proper FR for current $\tcur^t$, the crucial difference with the FR Eq.\,(\ref{eq:fr}) being that it does not compare the same p.d.f. evaluated along different paths. In this section we propose a similar notion of hidden TR of a trajectory with properties that are analogous to those of time-reversal, and provide a marginal FR for a set of cycle currents associated with edge $\ot$; however, this FR cannot be further contracted to current $\tcur^t$ alone. The construction is based on a cycle analysis at the level of paths proposed by Jia et al. \cite{qians}.

To accomplish this, we consider all {\it simple} oriented cycles $\C$, with $-\C$ denoting the cycle with the inverse orientation. Suppose that a Markovian Dedalus lost in the network carries Ariadne's thread with himself. Chased by the Minothaurus, as the trajectory unravels Dedalus might cross the previously laid thread, forming loops; whenever such a crossing occurs, to avoid wasting precious thread he/she accounts for the loop and its directionality, cuts the loop away and sews the leftover strands together. This procedure produces counters of simple cycles. Let $\Psi^t(\C)$ be the winding number in one direction of the cycle. Then the cycle current is defined as
\begin{align}
\tcur^t(\C) := \Psi^t(\C) - \Psi^t({-\C}).
\end{align}
Importantly, the current along edge $\ot$ is the sum of all cycle currents whose cycle includes $\ot$, in the appropriate direction,
\begin{align}
\tcur^t = \sum_{\C \owns \ot} \tcur^t(\C).
\end{align}
Thus, cycle currents constitute a refinement of the information contained in an observable current. We can now separate the contributions to a trajectory that come from cycles that contain $\ot$, and those that come from simple cycles not containing $\ot$. For example, for 
\begin{align}
\traj = \{\overbracket{ 3 \stackrel{t_1}{\to}  \overbracket{2 \stackrel{t_2}{\to}  1 \stackrel{t_3}{\to}  4 \stackrel{t_4}{\to}  2} \stackrel{t_5}{\to} 5 \stackrel{t_6}{\to} 4 \stackrel{t_7}{\to}  3  \stackrel{t_8}{\to}} \}
\end{align}
we obtain the two cycles $2 \to 1 \to 4 \to 2$ and $3 \to 5 \to 4 \to 3$ (notice that there also exists cycle $4 \to 2 \to 5 \to 4$, which is not independent of the other two -- thus showing that the representation in terms of cycle currents is not unique: it depends on {\it when} we start to measure). We invert the first, obtaining the {\it hidden time-reversed trajectory}
\begin{align}
\widetilde{\omega}^t = \{\overbracket{ 3 \stackrel{t_8}{\to} 4   \stackrel{t_7}{\to} 5  \stackrel{t_6}{\to} \stackrel{\vspace{-.4pt}\mathrm{TR}}{\overbracket{2 \stackrel{t_2}{\to} 1 \stackrel{t_3}{\to}  4 \stackrel{t_4}{\to}  2}} \stackrel{t_5}{\to} 3 \stackrel{t_1}{\to}} \}
\end{align}
and similarly we rearrange the waiting times. This procedure of hidden time-reversal of a trajectory has the following properties, analogous to that of full time-reversal:
\begin{itemize}
\item[{\it (i)}] It is involutive, hence in particular it is one-to-one;
\item[{\it (ii)}] It inverts all of the cycle currents but those that pass through $\ot$, hence it also preserves the observable current
\begin{align}\widetilde{\tcur}^t:= \tcur[\widetilde{\omega}^t] = \tcur^t;
\end{align}
\item[{\it (iii)}] It preservers all symmetric quantities, namely waiting times and activities.
\end{itemize}

\begin{theorem}
The following FR for the cycle currents passing through edge $\ot$ holds:
\begin{align}
\frac{\prob(\{\tcur^t(\C)\}_{\C \owns 12})}{\prob(\{-\tcur^t(\C)\}_{\C \owns 12})} \asymp \exp \sum_{\C \owns 12} \tcur^t(\C) \A(\C).
\end{align}
\end{theorem}

\begin{proof}
Assuming that the trajectories are cyclic, $i_N = i_0$, with a few standard manipulations one arrives at
\begin{align}
\frac{\prob(\traj)}{\prob(\overline{\widetilde{\tcur}}^t)} = \exp \sum_{\C \owns 12} \tcur^t(\C) \A(\C).  
\end{align}
We can then marginalize for the cycle currents.
\end{proof}

This comes as close as possible to obtaining a marginal FR that is localized on edge $\ot$. Notice though that in general 
\begin{align}
\prob(\minvtraj) \neq \widetilde{\prob}(\traj).
\end{align}
Understanding whether there exists relationship between the hidden TR trajectory and the hidden TR dynamics is a major open question.

\subsection{The case of multiple edges}
\label{sec:multipledges}

For sake of notational ease we assumed so far that at most one transition per pair of sites was possible, i.e. that the network does not have multiple edges. In nonequilibrium thermodynamics it is necessary to consider the case where several different transitions are discernible, because due to the interaction with different reservoirs. In this section we show that the theory holds with few obvious modifications.

Let us consider the following network, which includes two transition mechanisms for $\ot$, labeled $\RN{1}$ and $\RN{2}$:
\begin{align}
\begin{array}{c}\xymatrix{1 \ar@{-}^{\RN{1}}@/^/[rr]  \ar@{-}_{\RN{2}}@/_/[rr] \ar@{-}[dr] \ar@{..}@/_/[dr] \ar@{-}@/^/[dr] & & 2  \ar@{..}@/^/[dl]  \ar@{-}@/_/[dl]  \ar@{-}[dl] \\
& \rule{1cm}{1cm}  & 
}\end{array}
\end{align}
The tilted operator corresponding to the measurement of the current along edge  $1 \stackrel{\RN{1}}{\gets} 2$ is given by
\begin{align}
\M^{\RN{1}}(q) & :=
\left(\ba{cccccc} - w_{1} &  e^{-q} w^{\RN{1}}_{12}  + w^{\RN{2}}_{12}& \ldots  \\
 e^{q} w^\RN{1}_{21} + w^\RN{2}_{21}&  - w_2 &   \\
\vdots & & \ddots 
 \\
 \ea \right).
\end{align}
Apart from this little difference, a direct inspection of the proofs of \Th{th:unitary} and  \Th{th:characterization} shows that all proceeds like above, with the only modification that the effective affinity now reads
\begin{align}
\Q = \log \frac{w_{12}^{\RN{1}} p_2^{\st}}{w_{21}^{\RN{1}} p_1^{\st}} 
\end{align}
where $\vec{p}^{\,\st}$ is the steady-state in the network where edge $1 \stackrel{\RN{1}}{\mbox{---}} 2$ is removed. This argument can be easily scaled up to the case of several observational currents that will be the object of study of the rest of this paper. Since the procedure is fairly straightforward, we will not produce it explicitly.

\section{Episode 3: Several edge currents}
\label{sec:multiedge} 

In this episode we take into consideration several currents, each supported on one edge, proving the analog of \Th{th:unitary} and \Th{th:effFR} and discussing some of their consequences. The most important concept is that there is a {\it hierarchy} of marginal theories and of powerful inter-hierarchical fluctuation relations.

\subsection{Main result}
\label{sec:main2}

We now consider the statistics of $|\mu|$  currents, each supported on a different edge $i_\mu j_\mu$. For the moment we assume that  the removal of all such edges does not disconnect the network.
We call $j_\mu$  the tail (or source) of edge $i_\mu j_\mu$ and $i_\mu$ its tip (or target). The set of currents is marginal in the sense explained in Sec.\,\ref{sec:cycle}: it does not include a full set of chords of the network. Let $\M(\{q_\mu\})$ be the tilted generator with entries
\begin{align}
\M(\{q_\mu\})_{i,j} = \left\{
\ba{ll}
- w_i, & \mathrm{if}~ i = j\\
w_{i_\mu j_\mu}e^{- q_\mu}, & \mathrm{if} \, \exists \mu\; \mathrm{s.t.}\;  i = i_\mu, j = j_\mu \\
w_{\!j_\mu i_\mu} e^{q_\mu}, & \mathrm{if} \, \exists \mu\; \mathrm{s.t.}\;  i = j_\mu, j = i_\mu\\
w_{i\!j}, & \mathrm{otherwise}
\ea
 \right. .
\end{align}
From now on we omit to explicitly state ``$\exists \mu$''. 

Let us consider the generator $\W_{\mathrm{hid}}$ of the Markovian dynamics defined by setting all rates $w_{i_\mu j_\mu} = w_{j_\mu i_\mu} = 0$, {\it ossia} by removing the set of edges corresponding to the observable transitions. It has entries
\begin{align}
[\W_{\mathrm{hid}}]_{i,j} = \left\{
\ba{ll}
- w_i + \sum_{k,\mu} w_{ki} (\delta_{k,j_\mu}  \delta_{i,i_\mu} + \delta_{i,j_\mu}  \delta_{k,i_\mu} ), & \mathrm{if}~ i = j,\\
0, & \mathrm{if}~ i_\mu j_\mu = i\!j, ji, \\
w_{i\!j}, & \mathrm{elsewhere}
\ea
 \right. .
\end{align}
Such dynamics admits a unique normalized steady state $\vec{p}^{\,\st}$. We define the effective affinities
\begin{align}
\Q_\mu := \log \frac{w_{i_\mu j_\mu} p^{\st}_{j_\mu}}{w_{\!j_\mu i_\mu} p^{\st}_{i_\mu}}. \label{eq:effaffmu}
\end{align}
Finally, let $\Past := \mathrm{diag}\, \{p^\st_i\}_i$.
\begin{theorem}
\label{th:unitary2}
The operator
\begin{align}
\widetilde{\W} := \Past \, \M(\{\Q_\mu\})^T \,\Past^{-1} \label{eq:unitary2}
\end{align}
is a {\rm MJPG}.
\end{theorem}

\begin{proof}
We have
\begin{align}
\widetilde{\W}(\{\Q_\mu\})_{j,i} & = \left\{
\ba{ll}
- w_i, & \mathrm{if}~ i = j\\
w_{i_\mu j_\mu} \frac{p^\st_{\!j_\mu}}{p^\st_{i_\mu}} e^{-\Q_\mu}, & \mathrm{if}~ i = i_\mu, j = j_\mu \\
w_{\!j_\mu i_\mu}\frac{p^\st_{i_\mu}}{p^\st_{\!j_\mu}} e^{+\Q_\mu}, & \mathrm{if}~ i = j_\mu, j = i_\mu\\
w_{i\!j} \frac{p^\st_{\!j}}{p^\st_i} , & \mathrm{elsewhere}
\ea
 \right. \\
 & = \left\{
\ba{ll}
- w_i, & \mathrm{if}~ i = j\\
w_{j_\mu i_\mu}, & \mathrm{if}~ i = i_\mu, j = j_\mu \\
w_{\!i_\mu j_\mu}, & \mathrm{if}~ i = j_\mu, j = i_\mu\\
w_{i\!j} \frac{p^\st_{\!j}}{p^\st_i} , & \mathrm{elsewhere}
\ea . 
 \right. \label{eq:mmm}
 \end{align}
Off-diagonal entries are positive, diagonal entries are negative, therefore all we need to prove is that columns add up to zero. Summing over $j$, for each $i$, we obtain
\begin{multline}
\sum_j \widetilde{\W}(\{\Q_\mu\})_{j,i} = - w_i + \sum_{j \neq i} [\W_{\mathrm{hid}}]_{i\!j} \frac{p_{\!j}}{p_i}  \\ + \sum_j w_{ji} \sum_\mu (\delta_{i,i_\mu}\delta_{j,j_\mu}  + \delta_{i,j_\mu}\delta_{j,i_\mu}) 
\end{multline}
where we recognized in the last term the off-diagonal entries of the generator $\W^{\setminus}$, and in the second term in this expression precisely the off-diagonal correction that is needed to have
\begin{align}
\sum_j \widetilde{\W}(\{\Q_\mu\})_{j,i} = \sum_j [\W_{\mathrm{hid}}]_{i\!j} \frac{p_{\!j}}{p_i}  = 0
\end{align}
where we used the fact that $\vec{p}$ is the steady state of $\W_{\mathrm{hid}}$. 
\end{proof}
Notice that, according to Eq.\,(\ref{eq:mmm}), forward and hidden TR generators can be written respectively as in Eqs.\,(\ref{eq:12}), where $\W_{\mathrm{hid}}$ is the generator obtained by removing all edges $\{i_\mu j_\mu\}$ and $\W_{\mathrm{mar}}$ is the complementary generator of the dynamics on those edges only. 

Let us mention that by the construction of the so-called ``Doob trasform,'' that we will briefly discuss in the conclusions, Sec.\,\ref{exode}, there is no impediment in producing a proof of \Th{th:unitary2} for any values $q_\mu = q^\ast_\mu$ in the locus of zeroes of the SCGF, as mentioned in the Parode.

\subsection{Graphical representation}

Similarly to the analysis in Sec.\,\ref{sec:graphical}, effective affinities can be given a graphical interpretation. We will only give a simple example. Consider the complete graph on four nodes
\bea
\xymatrix{ 1  \ar@{-}[r] & 2 \ar@{-}[d] \\ 
3 \ar@{-}[ur] \ar@{-}[u] & 4  \ar@{-}[ul] \ar@{-}[l]}
\eea
and let us consider the marginal theory where the observer measures currents $1\mbox{--}2$ and $3\mbox{--}4$. After some work we obtain
\begin{subequations}
\begin{align}
\Q_{12} & = \log \frac{\Scale[0.5]{
\ba{c}\xymatrix{ \ar@{<-}[r]  & \ar@{<-}[d] \\ \ar@{->}[ur] \ar@{<-}[u] & } \ea
 + \ba{c}\xymatrix{ \ar@{<-}[r]  & \ar@{<-}[d] \\  \ar@{->}[u] & \ar@{<-}[ul]} \ea
 + \ba{c}\xymatrix{ \ar@{<-}[r]  &  \\ \ar@{<-}[u] \ar@{->}[ur] & \ar@{->}[ul] } \ea
 + \ba{c}\xymatrix{ \ar@{<-}[r]  & \ar@{<-}[d] \\  \ar@{->}[ur] & \ar@{<-}[ul]} \ea
}}{\Scale[0.5]{
\ba{c}\xymatrix{ \ar@{->}[r]  & \ar@{<-}[d] \\ \ar@{<-}[ur] \ar[u] & } \ea
+ \ba{c}\xymatrix{ \ar@{->}[r]  & \ar@{<-}[d] \\  \ar@{->}[u] & \ar@{->}[ul]} \ea
+ \ba{c}\xymatrix{ \ar@{->}[r]  &\\ \ar@{->}[u] \ar@{<-}[ur] &  \ar@{->}[ul]  } \ea
+ \ba{c}\xymatrix{ \ar@{->}[r]  & \ar@{->}[d] \\  \ar@{->}[ur] & \ar@{->}[ul]} \ea
 }} & & =: \log \frac{\psi_{\setminus 12}}{\psi_{\setminus 21}} \\
\Q_{34} & = \log \frac{\Scale[0.5]{
\ba{c}\xymatrix{  & \ar@{->}[d] \\  \ar@{<-}[r]  \ar@{->}[ur] \ar@{<-}[u] & } \ea
 + \ba{c}\xymatrix{  & \ar@{->}[d] \\  \ar@{<-}[r]  \ar@{->}[u] & \ar@{<-}[ul]} \ea
 + \ba{c}\xymatrix{   &  \\ \ar@{<-}[r]  \ar@{->}[u] \ar@{<-}[ur] & \ar@{<-}[ul] } \ea
 + \ba{c}\xymatrix{  & \ar@{->}[d] \\ \ar@{<-}[r]  \ar@{->}[ur] & \ar@{<-}[ul]} \ea
}}{\Scale[0.5]{
\ba{c}\xymatrix{  & \ar@{->}[d] \\ \ar@{->}[r]  \ar@{<-}[ur] \ar[u] & } \ea
+ \ba{c}\xymatrix{   & \ar@{->}[d] \\ \ar@{->}[r]  \ar@{<-}[u] & \ar@{->}[ul]} \ea
+ \ba{c}\xymatrix{  &\\ \ar@{->}[r]  \ar@{<-}[u] \ar@{<-}[ur] &  \ar@{->}[ul]  } \ea
+ \ba{c}\xymatrix{  & \ar@{<-}[d] \\ \ar@{->}[r]  \ar@{<-}[ur] & \ar@{<-}[ul]} \ea
 }} & & =: \log \frac{\psi_{\setminus 34}}{\psi_{\setminus 43}}.
\end{align}
\end{subequations}
The right-hand side defines the hidden fluxes. We recognize in these expressions those Hill cycle fluxes that are {\it not} in common between the two currents, which explicitly read (setting $\pi(\G) \stackrel{!}{=}1$) 
\begin{subequations}
\begin{align}
\cur_{12} = \psi_{\setminus 12} - \psi_{\setminus 21} &  + 
\Scale[0.5]{\ba{c}\xymatrix{ \ar@{<-}[r]  & \ar@{<-}[d] \\ \ar@{->}[r] \ar@{<-}[u] & } \ea} +
\Scale[0.5]{\ba{c}\xymatrix{ \ar@{<-}[r]  &  \\  \ar@{->}[ur] & \ar@{<-}[ul] \ar@{->}[l]} \ea} +\Scale[0.5]{\ba{c}\xymatrix{ \ar@{<-}[r]  &  \\ \ar@{->}[ur] \ar@{<-}[u] \ar@{<-}[r] & } \ea}
 + \Scale[0.5]{\ba{c}\xymatrix{ \ar@{<-}[r]  & \ar@{<-}[d] \\  \ar@{->}[r] & \ar@{<-}[ul]} \ea}
\nonumber \\
& 
-\Scale[0.5]{\ba{c}\xymatrix{ \ar@{->}[r]  & \ar@{->}[d] \\ \ar@{<-}[r] \ar@{->}[u] & } \ea} -
\Scale[0.5]{\ba{c}\xymatrix{ \ar@{->}[r]  &  \\  \ar@{<-}[ur] & \ar@{->}[ul] \ar@{<-}[l]} \ea}
-\Scale[0.5]{
\ba{c}\xymatrix{ \ar@{->}[r]  &  \\ \ar@{<-}[ur] \ar@{<-}[r] \ar[u] & } \ea}
-\Scale[0.5]{\ba{c}\xymatrix{ \ar@{->}[r]  & \ar@{<-}[d] \\  \ar@{->}[r] & \ar@{->}[ul]} \ea}
\\
\cur_{34} =  \psi_{\setminus 34} - \psi_{\setminus 43} & +
\Scale[0.5]{\ba{c}\xymatrix{ \ar@{->}[r]  & \ar@{->}[d] \\ \ar@{<-}[r] \ar@{->}[u] & } \ea} +
\Scale[0.5]{\ba{c}\xymatrix{ \ar@{<-}[r]  &  \\  \ar@{->}[ur] & \ar@{<-}[ul] \ar@{->}[l]} \ea} + \Scale[0.5]{\ba{c}\xymatrix{  \ar@{->}[r] & \ar@{->}[d] \\  \ar@{<-}[r]  \ar@{->}[ur]  & } \ea}
+ \Scale[0.5]{\ba{c}\xymatrix{  & \ar@{->}[l] \\  \ar@{<-}[r]  \ar@{->}[u] & \ar@{<-}[ul]} \ea}
 \nonumber \\
& - \Scale[0.5]{\ba{c}\xymatrix{ \ar@{<-}[r]  & \ar@{<-}[d] \\ \ar@{->}[r] \ar@{<-}[u] & } \ea}
- \Scale[0.5]{\ba{c}\xymatrix{ \ar@{->}[r]  &  \\  \ar@{<-}[ur] & \ar@{->}[ul] \ar@{<-}[l]} \ea} 
- \Scale[0.5]{\ba{c}\xymatrix{ \ar@{->}[r]   & \ar@{->}[d] \\ \ar@{->}[r]  \ar@{<-}[ur]  & } \ea}
- \Scale[0.5]{\ba{c}\xymatrix{   & \ar@{->}[l] \\ \ar@{->}[r]  \ar@{<-}[u] & \ar@{->}[ul]} \ea}.
\end{align}
\end{subequations}
Notice that from these expressions it is not so obvious (to the best of our understanding) that the mean marginal EPR $\cur_{12} \Q_{12} + \cur_{34} \Q_{34}$, should be non-negative, as we could prove directly in Sec.\,\ref{sec:marepr} for the single-edge case. This fact will rather follow from the IFR.

\subsection{Fluctuation relations}

We derive the FR at all times using the formalism of the tilted generator and of the generating function, rather than the direct derivation in terms of path p.d.f.'s given in the previous Episode.

\begin{theorem}
The marginal FR
\begin{align}
\frac{\prob(\{\tcur_\mu\})}{\widetilde{\prob}(-\{\tcur_\mu\})} = \exp { \ssum \Q_\mu\,\tcur_\mu}   \label{eq:marfr}
\end{align}
holds at all times, provided the initial configuration is sampled with probability $P(i^0 \equiv i) = p_i^{\,\st}$.
\end{theorem}

\begin{proof}
Rather than considering the probability itself, we consider the moment generating function of the currents $\zeta^t(\{q_\mu\})$ at time $t$, and let $\vec{\zeta}^t(\{q_\mu\})= (\zeta^t_i(\{q_\mu\}))_i$ be the moment generating function conditioned to being at site $i$ at time $t$, defined as
\begin{align}
\zeta^t_i(\{q_\mu\}) := \int \prod d\tcur_\mu \,\prob(\tcur_\mu^t \equiv \tcur_\mu,i^t\equiv i) \,  \exp {-\sum \tcur_\mu q_\mu},
\end{align}
such that the moment generating function is recovered from the conditional one by $\zeta^t(\{q_\mu\}) = \vec{1}\cdot \vec{\zeta}^{\,t}(\{q_\mu\})$. We can also introduce a conditional moment generating function for the hidden TR process $\vec{\widetilde{\zeta}}$. These two conditinal functions evolve respectively by
\begin{subequations}
\begin{align} 
\frac{d}{dt} \vec{\zeta}^{\,t}(\{q_\mu\}) & = \M(\{q_\mu\})\,\vec{\zeta}^{\,t}(\{q_\mu\}), \\ 
\frac{d}{dt} \vec{\widetilde{\zeta}}^{\,t}(\{q_\mu\}) & = \widetilde{\M}(\{q_\mu\})\, \vec{\widetilde{\zeta}}^{\,t}(\{q_\mu\}).
\end{align}
\end{subequations}
Let $\zeta^0_i(\{q_\mu\}) = \widetilde{\zeta}_i^{\,0}(\{q_\mu\}) =  p_i^{0}$ be the probability of finding the system in site $i$ at time $t=0$. Defining $\U^t(\{q_\mu\}) := \exp t \M(\{q_\mu\})$ and $\widetilde{\U}^t(\{q_\mu\}) := \exp t \widetilde{\M}(\{q_\mu\})$ we have
\begin{align} 
\vec{\zeta}^{\,t}(\{q_\mu\}) & =  \U^t(\{q_\mu\})\, \vec{p}^{\,0}.
\end{align}
Similarly, for the hidden TR dynamics, evaluating at $q_\mu \to \Q_\mu -q_\mu$ we obtain
\begin{align}
\vec{\widetilde{\zeta}}^{\,t}(\{\Q_\mu - q_\mu\}) & = \widetilde{\U}^t(\{\Q_\mu - q_\mu\}) \,\vec{p}^{\,0} \\
& = \Past \left[\exp t\M(\{q_\mu\})^T \right] \, \Past^{-1}  \vec{p}^{\,0}.
\end{align}
We can now find the moment generating functions as
\begin{subequations}
\begin{align} 
\zeta^{\,t}(\{q_\mu\}) & = \vec{1}\cdot \U^t(\{q_\mu\}) \,\vec{p}^{\,0} \\
\widetilde{\zeta}^{\,t}(\{\Q_\mu - q_\mu\}) & = \vec{p}^{\,0} \cdot \Past^{-1} \U^t(\{q_\mu\}) \, \Past \vec{1} 
\end{align} 
\end{subequations}
from which it follows that the choice of $\vec{p}^{\,0} = \vec{p}^{\,\st}$ makes the latter two expressions identical at all times, thus yielding an all-time fluctuation symmetry. Eq.\,(\ref{eq:marfr}) follows by transforming back from the moment generating function to the currents' p.d.f. by the phantomatic inverse bilateral Laplace transform.
\end{proof}

As a straightforward corollary, we have the following asymptotic FRs.

\begin{theorem}\label{th:marls}\label{th:ift2}
	The following marginal symmetry of the SCGF and IFR hold
	\begin{align}
		\lambda(\{\Q_\mu - q_\mu\}) & = \widetilde{\lambda}(\{q_\mu\}), \label{eq:lss} \\ \lambda(\{\Q_\mu\}) & = 0.
	\end{align}
\end{theorem}

\begin{proof}
	The first follows from
	\begin{align}
		\lambda(\{q_\mu\}) = \lim_{t \to \infty} \frac{1}{t} \zeta^t(\{q_\mu\}).
	\end{align}
	The second follows by evaluating at $q_\mu = \Q_\mu$.
\end{proof}

Just like the single-edge-current FR actually holds at the level of trajectories as in Eq.\,(\ref{eq:123}), we can also generalize the above \th{} to a ``complete'' set of currents as follows.

\begin{theorem}\label{th:genfr}
The SCGF $\lambda(\{q_\mu\}_\mu,\{q_\alpha\}_{\alpha \geq |\mu|})$ of a ``complete'' set of currents satisfies the FR
	\begin{align}
		\lambda(\{q_\mu\}_\mu,\{q_\alpha\}_{\alpha\neq\mu}) = \widetilde{\lambda}(\{\Q_\mu - q_\mu\}_\mu,\{-q_\alpha\}_{\alpha\neq\mu}).
	\end{align}
\end{theorem}

\begin{proof}
Let us construct the tilted generator of the forward and backward dynamics in terms of the Hadamard product (see p.\,\pageref{hadamard}) as (we drop the explicit range of the indices)
	\begin{align}
	& \widetilde{\M}(\{q_\mu\},\{q_\alpha\}) = \nonumber \\
	& =  \Past \W_{\mathrm{hid}}^T \Past^{-1} \circ \matrix{T} (\{q_\alpha\}) + \W_{\mathrm{mar}} \circ \matrix{T} (\{q_\mu\}) \nonumber \\
	& =  \Past \left[ \W_{\mathrm{hid}} \circ \matrix{T} (\{-q_\alpha\}) \right]^T \Past^{-1} + \W_{\mathrm{mar}} \circ \matrix{T} (\{q_\mu\})
	\end{align}
where we used the fact that $\matrix{T} (\{q_\alpha\})^T = \matrix{T} (\{-q_\alpha\})$. We have
	\begin{align}
& \M(\{\Q_\mu - q_\mu\},\{-q_\alpha\}) = \nonumber \\
& = \W_{\mathrm{hid}} \circ \matrix{T} (\{q_\alpha\}) + \W_{\mathrm{mar}} \circ \matrix{T} (\{q_\mu\}) \nonumber \\
& = \W_{\mathrm{hid}} \circ \matrix{T} (\{-q_\alpha\}) + \W_{\mathrm{mar}} \circ \matrix{T} (\{\Q_\mu\}) \circ \matrix{T} (\{-q_\mu\}) \nonumber  \\
& = \W_{\mathrm{hid}} \circ \matrix{T} (\{-q_\alpha\}) + \Past \W_{\mathrm{mar}}^{T} \Past^{-1} \circ \matrix{T} (\{-q_\mu\})  \nonumber \\
& = \W_{\mathrm{hid}} \circ \matrix{T} (\{-q_\alpha\}) + \Past \left[\W_{\mathrm{mar}} \circ \matrix{T} (\{q_\mu\}) \right]^{T} \Past^{-1}
	\end{align}
where in the first passage we used the entry-wise associative property of the Hadamard product and in the second the main proposition \Th{th:unitary2}. We then obtain
\begin{align}
\M(\{\Q_\mu - q_\mu\},\{-q_\alpha\}) = \Past \; \widetilde{\M}(\{q_\mu\},\{q_\alpha\})^T \Past^{-1}
\end{align}
and the conclusion follows from matrix similarity.
\end{proof}

\subsection{The hierarchy of marginal theories}
\label{sec:hierarchy}

Let us now suppose we follow a systematic procedure by which we first consider the marginal theory along edge $i_1 j_1$, then we expand our knowledge to $i_2 j_2$, and so on up to $i_{|\alpha|}j_{|\alpha|}$ when we cover a ``complete'' set of chord currents. Let $ \Q^{1,\ldots,|\mu|}_{\mu} $
be the effective affinity along edge $i_\mu j_\mu$ obtained from the $|\mu|$-th marginal theory in the hierarchy, where the superscript $1,\ldots,|\mu|$ denotes the set of marginal currents. Each such theory gives an estimate of the entropy production rate
\begin{align}
\sigma_{1,\ldots,|\mu|} := \sum_{\mu = 1}^{|\mu|} \Q^{1,\ldots,|\mu|}_{\mu} \cur_\mu.
\end{align}
The following \th{} is a generalization of the result presented in Ref.\,\cite{gili}.
\begin{theorem}\label{th:hierarchy}
The following hierarchy of inequalities is satisfied:
\begin{align}
0 \leq \sigma_{1} \leq \sigma:{1,2} \leq  \ldots \leq \sigma_{1,2,\ldots,|\alpha|}.
\end{align}	
\end{theorem}
\begin{proof}
Notice that the \Th{th:genfr} can be written in terms of the currents' p.d.f.'s as
\begin{align}
\frac{\prob\left(\{\tcur_\alpha\}_{\alpha = 1}^{|\alpha|}\right)}{\widetilde{\prob}\left(\{\tcur_\alpha\}_{\alpha = 1}^{|\alpha|}\right)} = \exp\left(\sum_{\alpha = 1}^{|\alpha|} \A_\alpha \tcur_\alpha - \sum_{\mu =1}^{|\mu|} \Q^{1,\ldots,|\mu|}_\mu \tcur_\mu\right).
\end{align}
Let $\widetilde{\prob}^{1,\ldots,|\mu|}\left(\{\tcur_\mu\}_{\mu = 1}^{|\mu'|}\right)$ denote the probability of the first $|\mu'|$ currents according to the hidden TR dynamics of the $|\mu|$-th theory, for $|\mu| > |\mu|'$ (notice that this is a new object: so far we only considered the probability of the first $|\mu|$ currents in the $|\mu|$-th theory, and did not specified the level of the hierarchy). Since the latter equality holds for all $|\mu|$, by dividing two such relations for $|\mu|$ and $\mu'$, we obtain
	\begin{align}
	\frac{\widetilde{\prob}^{|\mu|}\left(\{\tcur_\mu\}_{\mu = 1}^{|\mu|}\right)}{\widetilde{\prob}^{|\mu'|}\left(\{\tcur_\mu\}_{\mu = 1}^{|\mu|}\right)} = \exp \left( \sum_{\mu=1}^{|\mu|} \Q^{1,\ldots,|\mu|}_\mu \tcur_\mu - \sum_{\mu=1}^{|\mu'|} \Q^{1,\ldots,|\mu'|}_\mu \tcur_\mu \right)
	\end{align}
We can then follow the usual route, obtaining the corresponding IFRs and using the Jensen inequality.
\end{proof}

\subsection{Stalling and response at stalling}
\label{multistalling}

The following result characterizes the stalling steady state.
\begin{theorem}
\label{iffstalling}
The effective affinities vanish if and only if the mean marginal currents vanish.
\end{theorem}

\begin{proof}
This is evident given the explicit expression for the current $\phi_\mu = w_{\!i_\mu j_\mu} p_{j_\mu} - w_{\!j_\mu i_\mu} p_{i_\mu}$, and the fact that $\vec{p}^{\,\st}$ is the unique stalling state.
\end{proof}

Notice that, with the hierarchy of theories in mind, the set of rates for which current $\phi_1$ vanishes includes the set of rates for which currents $\phi_1,\phi_2$ vanish, and so on, until one reaches equilibrium. In this respect, equilibrium is just the stalling state of a marginal theory that happens to be complete.

Let us now parametrize the rates $w_{\ij} \to w_{\ij}(\x)$ by a set of parameters $\x = (\{x_\mu\}_\mu,\rest)$ such that the first $|\mu|$ are {\it marginally thermodynamic} in the sense that
\begin{align}
\frac{w_{i_\mu j_\mu}(\x)}{w_{j_\mu i_\mu}(\x)} = \exp x_\mu \label{eq:mupar}
\end{align}
and that no other rate depends on $x_\mu$. The effective affinity now takes the form
\begin{align}
\Q_\mu(x) = x_\mu - x^\st_\mu(\rest) \label{eq:effaffeff}
\end{align}
where
\begin{align}
x^\st_\mu(\rest) := \log \frac{p_{j_\mu}^\st(\rest)}{p_{i_\mu}^\st(\rest)}.
\end{align}
Notice that the stalling value of the $\mu$-th parameter $x_\mu^\st(\rest)$ might depend on the other parameters. The following result is crucial for an operational definition of the effective affinities.
\begin{theorem}
\label{th:tuning}
At fixed dynamic parameters $\rest$, the stalling steady state obtained by removing the observable edges is the same as the stalling steady state obtained by tuning the thermodynamic parameters to the stalling values $x_\mu \stackrel{!}{=} x_\mu^\st(\rest)$.
\end{theorem}

\begin{proof}
Let us evaluate the forward generator on the stalling steady state:
\begin{align}
\W(\x) \vec{p}^{\,\st}(\x) & = [\W_{\mathrm{hid}}(\rest) + \W_{\mathrm{mar}}(\x)] \vec{p}^{\,\st} \nonumber \\
& = \W_{\mathrm{mar}}(\x)) \vec{p}^{\,\st}\nonumber \\ & = (0, \ldots,  \phi_\mu(\x) , \ldots, -\phi_{\mu}(\x), \ldots , 0)^T 
\end{align}
where we used the definition of stalling $\W_{\mathrm{hid}} \vec{p}^{\,\st} = 0$ and the fact that the hidden generator does not depend on the thermodynamic parameters. In the resulting vector we only highlighted the entries corresponding to $i_\mu$ and $j_\mu$.
 Then, in the light of \Th{iffstalling} and of Eq.\,(\ref{eq:effaffeff}), this expression vanishes whenever the effective affinities vanish.
\end{proof}

We now exploit the marginal fluctuation  symmetry Eq.\,(\ref{eq:lss}), that in its parametric form reads
\begin{align}
\widetilde{\lambda}\left(\{q_\mu\};\bs{x}\right) = \lambda\left(\{x_\mu - x_\mu^{\mathrm{st}}(\rest) - q_\mu\};\bs{x}\right).
\end{align}
Taking the second mixed derivatives with respect to $q_\mu,q_\nu,x_\mu,x_\nu$ (omitting explicit dependencies) we obtain
\begin{subequations}
\begin{align}
\frac{\partial^2 \widetilde{\lambda}}{\partial q_\mu q_\nu} & = \frac{\partial^2 \lambda}{\partial q_\mu q_\nu},  \\
\frac{\partial^2 \widetilde{\lambda}}{\partial q_\mu \partial x_\nu} & =  -\frac{\partial^2 \lambda}{\partial q_\mu \partial q_\nu} - \frac{\partial^2 \lambda}{\partial q_\mu \partial x_\nu},  \\
\frac{\partial^2 \widetilde{\lambda}}{\partial x_\nu \partial x_\mu} & = \frac{\partial^2 \lambda}{\partial x_\nu \partial x_\mu} + \frac{\partial^2 \lambda}{\partial q_\mu \partial q_\nu} +   \frac{\partial^2 \lambda}{\partial x_\mu \partial q_\nu} + \frac{\partial^2 \lambda}{\partial q_\mu \partial x_\nu}.
\end{align}
\end{subequations}
Evaluating at $\{q_\mu = 0\}$ and at $\{x_\mu = x_\mu^\st\}$ we obtain the marginal FDRs
\begin{subequations}
\begin{align}
\widetilde{\ccur}_{\mu\nu} = \ccur_{\mu\nu} & = \lcur_{\nu\,;\,\mu}^{\st} + \lcur_{\mu\,;\,\nu}^{\st} \\
& = \widetilde{\lcur}_{\nu\,;\,\mu}^{\st} + \widetilde{\lcur}_{\mu\,;\,\nu}^{\st} \\
& =  \lcur_{\mu\,;\,\nu}^\st +  \widetilde{\lcur}_{\mu\,;\,\nu}^\st \\
& = \lcur_{\nu\,;\,\mu}^\st +  \widetilde{\lcur}_{\nu\,;\,\mu}^\st
\end{align}
\end{subequations}
and the marginal RRs
\begin{subequations}
\begin{align}
\lcur_{\mu\,;\,\nu}^\st  = \widetilde{\lcur}_{\nu\,;\,\mu}^\st.
\end{align}
\end{subequations}
We can also consider the response to a perturbation of the $\kappa$-th kinetic parameter. By taking the first total derivative with respect to $x_\kappa$ and evaluating at $\{q_\mu = 0\}$ we obtain the orthogonality relation
\begin{align}
\sum_\mu \cur_\mu \frac{d x_\mu^{\st}}{dx_\kappa} = 0.
\end{align}
The second mixed derivative with respect to $q_\mu,x_\kappa$ and to  $x_\kappa,x_{\kappa'}$ yield
\begin{align}
\lcur_{\mu\,;\,\kappa}  + \widetilde{\lcur}_{\mu\,;\,\kappa}  & = \sum_\nu \ccur_{\mu\nu} \frac{d x_\nu^{\st}}{dx_\kappa} \\
\sum_\mu \lcur_{\mu\,;\,\kappa}  \frac{d x_\mu^{\st}}{dx_{\kappa'}} + \sum_\mu \lcur_{\mu\,;\,\kappa'} \frac{d x_\mu^{\st}}{dx_{\kappa}} & = \sum_{\mu,\nu} \ccur_{\mu\nu}  \frac{d x_\nu^{\st}}{dx_{\kappa}}  \frac{d x_\mu^{\st}}{dx_{\kappa'}}.
\end{align}
As regards higher-order response relations, it is clear that the treatment given in \S\,\ref{subsec:response} can all be reproduced, with the only {\it caveat} that on the right-hand side of Eq.\,(\ref{eq:higher}) there should appear the response coefficients of the hidden TR dynamics. We will not do this explicitly.

\subsection{Inequalities between effective affinities}
\label{sec:continuum}

\begin{figure}
\centering
\includegraphics[width=.45\columnwidth]{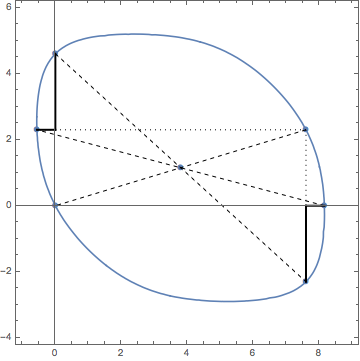}\hspace{.2cm}\includegraphics[width=.45\columnwidth]{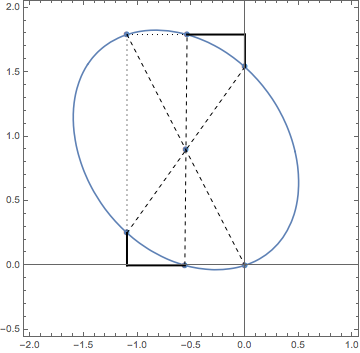}
\caption{ \label{fig:scgf}
Level curve $(q^\ast_1,q_\ast^2)$ where $\lambda(q^\ast_1,q^\ast_2) = 0$, for two models. Bullets along the curves, starting from the origin and going counterclockwise, are for the plot on the left $(0,0)$, $(\A_1,\A_2-\Q_2)$, $(\Q_1,0)$, $(\A_1,\A_2)$, $(0,\Q_2)$, $(\A_1-\Q_1,\A_2)$, and for the plot on the right $(0,0)$, $(0,\Q_2)$, $(\A_1-\Q_1,\A_2)$, $(\A_1,\A_2)$, $(\A_1,\A_2-\Q_2)$, $(\Q_1,0)$. 
At the center of the figures is the symmetry point $(\A_1/2,\A_2/2)$. The model on the left has rates $w_{21} = w_{14} = w_{32} = w_{14} = w_{43} = w_{42} = 1; w_{12} = w_{41} = w_{23} = 10; w_{34} = w_{24} = 20$. The model on the right has rates $w_{21} = w_{41} = w_{23} = w_{43} = 1; w_{14} = w_{32} = w_{24} = 10; w_{34} = 20; w_{42} = 30$.}
\end{figure}

In this section we briefly sketch some ideas on how to derive inequalities between effective affinities at different orders in the hierarchy. 

We first consider the case of a system with two ``real'' affinities $\A_1$ ($= \Q^{1,2}_1$) and $\A_2$ ($=\Q^{1,2}_2$) supported on the chords $i_1j_1$ and $i_2j_2$. Two marginal theories can be considered, one with effective affinity $\Q_1$ ($= \Q^{1}_1$) defined along edge $i_1j_1$ and one with $\Q_2$ ($=\Q_2^2$) defined along edge $i_2j_2$.

\begin{theorem}
Either both ``real'' affinities are smaller in modulus than the effective affinities, $|\A_1| \leq |\Q_1|$ and $|\A_2| \leq |\Q_2|$, or they are both larger, $|\A_1| \geq |\Q_1|$ or $|\A_2| \geq |\Q_2|$. 
\end{theorem}

\begin{proof}
We sketch an idea of the proof based on visual inspection of the two plots in Fig.\,\ref{fig:scgf}. They represent level curves where the SCGF vanishes, $\lambda(q^\ast_1,q^\ast_2) = 0$, for two specific models. Six points are highlighted along the curve: the origin $(0,0)$ the effective affinities $(\Q_1,0)$ and $(0,\Q_2)$, and their dual under the mirror symmetry with respect to the point $(\A_1/2,\A_2/2)$, which need to belong to the curve as a consequence of the fluctuation symmetry. Any level curve of a convex function is a convex set. Consider the two highlighted segments. Their coordinates are either $\Q_1-\A_1 > 0$ and $\Q_2-\A_2 >0$ for the left-hand curve, or $\A_2-\Q_2 > 0, \Q_1-\Q_1 >0$ for that on the right. Since for the right curve both $\A_1,\Q_1<0$, we are consistent with the claim of the proposition. A similar analysis of other special cases will convince the reader it is not possible to draw six such special points along a convex curve that do not satisfy the claim.
\end{proof}

\begin{figure}
\centering
\includegraphics[width=\columnwidth]{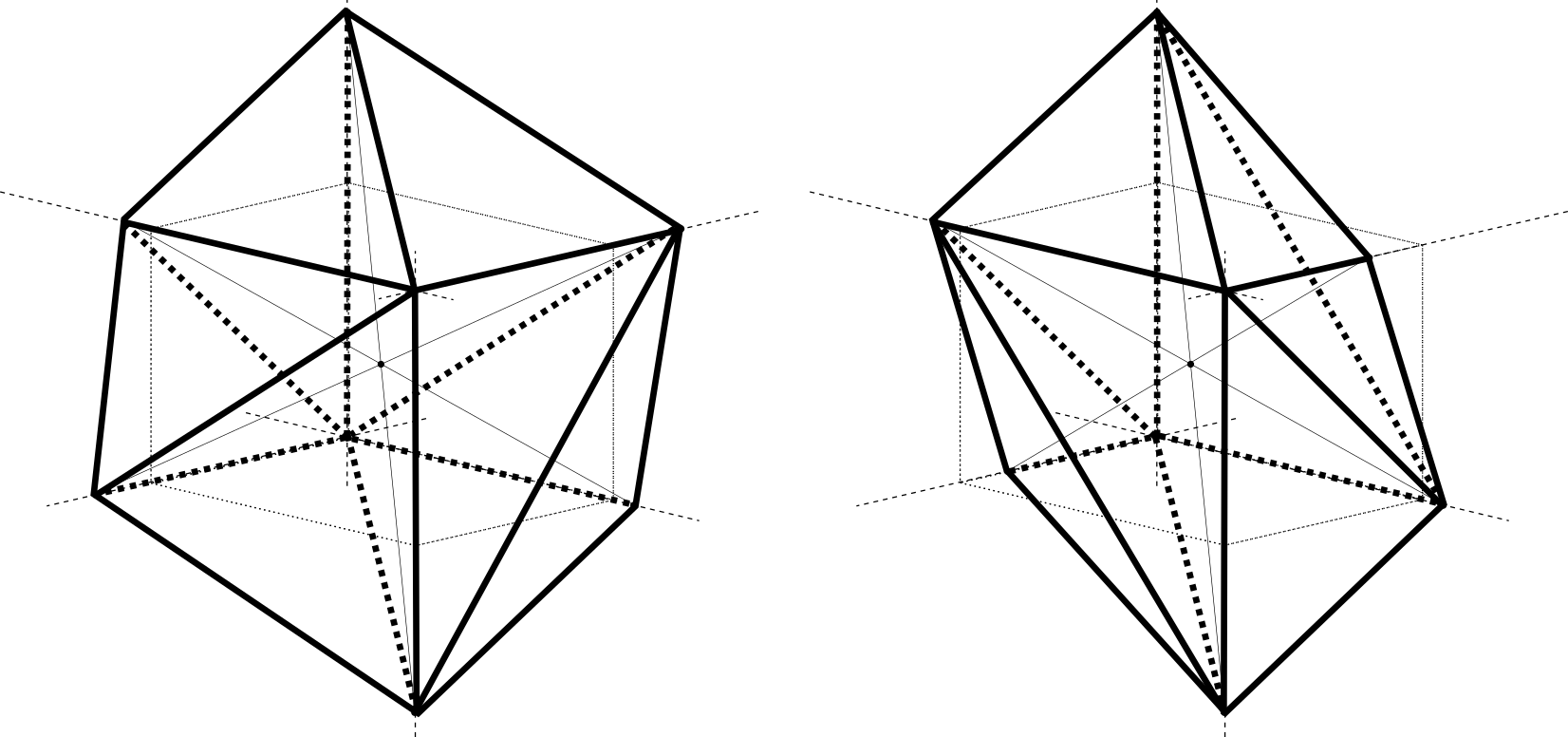}
\caption{ \label{fig:3d} In the space $(q_1,q_2,q_3)$, two polytopes having the effective affinities corresponding to level $|\alpha| = 1$ and $|\alpha| = 3$ marginal theories as their vertices. The image illustrates the conjectured fact that there always exist convex Delaunay triangulations of the vertices for any values of such effective affinities. However, the affinities corresponding to the $|\alpha| = 2$ will need to satisfy constraints.}
\end{figure}

Pushing this method beyond $|\alpha| = 2$ is far from trivial and it would make for an interesting mathematical inquiry. Let us briefly sketch some thoughts concerning the case $|\alpha| = 3$. In Fig.\,\ref{fig:3d} we plot the polytopes obtained by triangulations of the points $(0,0,0)$ --- the universal zero of the SCGF ---, the point $(\Q_1^{1,2,3},\Q_2^{1,2,3},\Q_3^{1,2,3}) =: \mathrm{A}$ --- the ``real'' affinities, the points $(\Q_1^1,0,0)$, $(0,\Q_2^2,0)$, $(0,0,\Q_3^3)$ --- the effective affinities of the single-edge marginal theories --- and their mirror images through point $1/2\mathrm{A}$ --- given by $(\Q_1^3- \Q_1^1,\Q_2^3,\Q_3^3)$,  $(\Q_1^3,\Q_2^3- \Q_2^1,\Q_3^3)$, $(\Q_1^{1,2,3},\Q_2^{1,2,3},\Q_3^{1,2,3}- \Q_3^1)$. In the plots, all (effective and ``real'') affinities are positive, without loss of generality. In the left-hand plot effective affinities are all larger than the ``real'' affinities, ($\Q^1_1 > \Q^{1,2,3}_1$, $\Q^2_2 > \Q_2^{1,2,3}$, $\Q_3^3 > \Q_3^{1,2,3}$), in the right-hand one one of the effective affinities is smaller, ($\Q^1_1 < \Q^{1,2,3}_1$, $\Q^2_2 > \Q^{1,2,3}_2$, $\Q^3_3 > \A^{1,2,3}_3$). In both cases one can find a convex envelope that includes all such points. Now consider that must belong to such envelope also the six further points corresponding to the second level in the hierarchy $|\mu|=2$, that is the three $(\Q^{1,2}_1,\Q^{1,2}_2,0)$, $(0,\Q^{2,3}_2,\Q_3^{2,3})$ and $(\Q^{1,3}_1,0,\Q_3^{1,3})$, and their FR-duals 
$(\Q_1^{1,2,3}-\Q^{1,2}_1,\Q_2^{1,2,3}-\Q^{1,2}_2,0)$, $(0,\Q_2^{1,2,3}-\Q^{2,3}_2,\Q_3^{1,2,3}-\Q_3^{2,3})$ and $(\Q_1^{1,2,3}-\Q^{1,3}_1,0,\Q_3^{1,2,3}-\Q_3^{1,3})$. It is clear that requiring that there exists a convex polytope that also includes such points as corners would pose strict inter-affinity relationships, whose determination however goes far astray from our scope.

\subsection{Gauge invariance}
\label{sec:gauge}

We have the following crucial result.

\begin{theorem}
\label{th:gauge}
Under the condition that the removal of the marginal edges does not disconnect the network, the effective affinities are invariant under a gauge transformation of the rates $w_{\ij} \to w'_{\ij} = w_{\ij} \,e^{-a_j}$.
\end{theorem}

\begin{proof}
By the matrix-tree theorem, the effective affinities write in terms of rooted oriented spanning trees on the graph where the observable edges are removed.  Notice that for every rooted oriented spanning tree $\T_{\!k}$ with root at site $k$, the gauge transformation gives  
\begin{align}
\ST'(\T_{\!k}) = \prod_{\ij \in \T_{\!k}} w'_{\ij} = \prod_{i \neq k} e^{a_i} \prod_{\ij \in \T_{\!k}} w_{\ij} =  e^{-a_k + a} \, \ST(\T_{\!k})
\end{align}
where $a = \sum_i a_i$. This is due to the fact that in a rooted oriented spanning tree, each site but the root is the tail of exactly one directed edge. This factor drops from the expression of the steady probability because of normalization. Thus
\begin{align}
w'_{\ij} p^{\st'}_j = w'_{\ji} p^{\st'}_i
\end{align}
and therefore the effective affinities are invariant.
\end{proof}
Since currents are also gauge-invariant, it follows that the effective EPR is also invariant. Gauges affect the shape of the initial state that needs to be chosen for an all-time FR to hold. 

A different, but related, acceptation of gauge invariance of the SCGF is that put forward by Wachtel et al. \cite{altaner15a}. Let us review it here for it will have implications as regards the marginal theory, that we will analyze in Sec.\,\ref{disconnect}.
\begin{theorem}
\label{th:gaugtran}
The SCGF for all the currents $\lambda(\{q_{\ij}\}_{\ij \in \E})$ is invariant under the  \emph{gauge transformation} $q_{\ij} \to q_{\ij} + \ssum \cocyc_{\ij}^{\alpha^\star} Q_{\alpha^\star}$, where the vectors $\{\cocyc_{\ij}^{\alpha^\star}\}_{\ij^\star}$ span the cocycle space of the graph.
\end{theorem}

\begin{proof}
Let us omit to specify $\ij \in \E$. Using a trick employed in Ref.\,\cite{andrieux} to prove the FR, we consider the characteristic polynomial of tilted generator
\begin{align}
\Delta\left(s;\{q_{\ij}\}\right) := \det \Big(  \W(\{q_{\ij}\};\bs{x}) - s \matrix{I} \Big).  \label{eq:cp}
\end{align}
We employ the the well-known expansion of the determinant of a matrix $\matrix{A}$ \cite{skiena}
\begin{align}
\det\,\matrix{A} = \sum_{\pi} (-1)^{|\pi|}  \prod_{i \in \I} A_{i,\pi(i)} 
\end{align}
in terms of permutations $\pi: \I \to \I$ the graph's sites, where $|\pi|$ is the permutation's parity. Any permutation admits a cycle decomposition \cite{skiena} in terms of singlets $i' \in \I^{\mathrm{sing}}(\pi) \subseteq \I$ such that $\pi(i') = i'$, and a certain number $|\iota|$ of simple cycles, $\C_\iota(\pi)$, which do not cross each other and which cover all of the remaining sites in $\I \setminus \I^{\mathrm{sing}}(\pi)$. One then obtains
\begin{align}
\prod_{i \in \I} A_{i,\pi(i)} = \prod_{i' \in \I^{\mathrm{sing}}(\pi)} A_{i'i'} \cdot \prod_\iota \prod_{\ij} A_{\ij}^{\cyc^{\ij}_\iota(\pi) },
\end{align}
yielding
\begin{multline}
\Delta\left(s;\big\{q_{\ij} + \ssum \cocyc_{\ij}^{\alpha^\star} Q_{\alpha^\star}\big\}\right)  \\
=  \sum_{\pi} (-1)^{|\pi|}  \prod_{i' \in \I^{\mathrm{sing}}(\pi)}  (- s- w_{i'})  \cdot \prod_\iota \prod_{\ij} \left( w_{\ij} e^{- q_{\ij}} \right)^{\cyc^{\ij}_\iota(\pi)}  \\
\times \exp {- \sum_{\alpha^\star,\iota} Q_{\alpha^\star} \sum_{\ij} \cocyc^{\alpha^\star}_{\ij} \cyc^{\ij}_\iota(\pi)} .
\end{multline}
The latter term is unity because cycles are orthogonal to cocycles, and therefore the characteristic polynomial is gauge invariant:
\begin{align}
\Delta\left(s;\big\{q_{\ij} + \ssum \cyc_{\ij}^{\alpha^\star} Q_{\alpha^\star}\big\}\right) = \Delta\left(s;\big\{q_{\ij}\big\}\right)
\end{align}
In fact, it follows that all of the eigenvalues are also gauge invariant, including the SCGF.
\end{proof}

The relationship between gauge invariance of cycle affinities and gauge invariance of the SCGF is as follows. A gauge transformation of the rates as described above transforms the log-ratio of the rates as $\log w_{\ij}/w_{\ji} \to \log w_{\ij}/w_{\ji} + a_i - a_j$ (a ``gauge connection''). Any such transformation can be seen as a cocyclic transformation, as it can be shown that $a_i - a_j  = \sum_{\alpha^\star}  \cyc_{\ij}^{\alpha^\star} a_{\alpha^\star}$ for some $a_{\alpha^\star}$. Indeed, using the FR for all the currents and imposing that it only depends on the cycle affinities, one obtains the latter invariance.

\subsection{Marginally thermodynamic parametrizations}
\label{parametrizations}

The local parametrization Eq.\,(\ref{eq:mupar}) automatically complies with the requirement B0 mentioned in the Parode. In this paragraph we analyze in more detail which parametrizations are {\it thermodynamic} in the sense that they satisfy the requirement A0 for the ``complete'' theory and, within such parametrization schemes, which are {\it marginally thermodynamic}. More precisely, in the first case our goal is to find the most general parametrization of the rates $\bs{x} \to w_{i\!j}(\bs{x})$ such that
\begin{align}
\frac{\partial}{\partial x_{\kappa}} \A_\alpha(\bs{x})= \delta_{\alpha,\kappa}. 
\label{eq:dissi1}
\end{align}
This question will find a conclusive answer in \Th{th:thermodynamic}, while \Th{th:mardiss1} characterizes the marginal parametrizations that satisfy
\begin{align}
\frac{\partial}{\partial x_{\mu'}} \Q_\mu(\bs{x})= \delta_{\mu,\mu'} . \label{eq:dissi2}
\end{align}
Notice that while Eq.\,(\ref{eq:dissi1}) involves all of the parameters, the second is much looser since it only involves the first $|\mu|$ parameters, and in general
\begin{align}
\frac{\partial}{\partial x_{\kappa}} \Q_\mu(\bs{x}) \neq 0, \qquad   \kappa > |\mu|. \label{eq:dissi3}
\end{align}

Let us now pave the way to \Th{th:thermodynamic} and \Th{th:mardiss1}. It is an obvious fact that transition rates can be uniquely decomposed as
\begin{align}
w_{i\!j}(\bs{x}) = v_{i\!j}(\bs{x}) \, e^{- a_{i\!j}(\bs{x})}, \label{eq:syan}
\end{align}
with an antisymmetric term $a_{i\!j} = - a_{\!ji}$ and a symmetric positive term $v_{i\!j} = v_{\!ji} > 0$. Clearly, by definition ``real'' affinities are only sensible to transformations of the antisymmetric part. Furthermore, as explained in Sec.\,\ref{sec:gauge}, they are invariant under so-called {\it gauge transformations} $a_{i\!j} \to a'_{i\!j}  = a_{i\!j} + a_i  - a_j$ \cite{polettinigauge,altaner15a}, hence the general dependence that the symmetric part of the rates can take is 
\begin{align}
a_{i\!j}(\bs{x}) = a'_{i\!j}(\{x_\alpha\}) + a_j(\bs{x}) - a_i(\bs{x}). \label{eq:transgen}
\end{align}
in terms of some {\it gauge connection} $a'_{i\!j}$ that only depends on the first $|\alpha|$ parameters, and of a {\it pure gauge} $a_j$ on which, for the moment, we impose no restriction. Notice that the affinity can be written in terms of the gauge connection only:
\begin{align}
\A_\alpha(\{x_{\alpha'}\})  = \sum_{i\!j \in \C_\alpha} a'_{i\!j}(\{x_{\alpha'}\})  \label{eq:Agf}.
\end{align}
Since the effect of a variation of the thermodynamic parameters on the affinities is via the gauge connection, we now need to impose that variations of $x_\alpha$ only modify the $\alpha$-th affinity according to Eq.\,(\ref{eq:dissi1}). The basic intuition is that any perturbation of the gauge connection along edges that are shared among different cycles will modify the affinities of both cycles. Then let us then introduce the {\it peripheral} set as the set of edges that belong to cycle $\C_\alpha$ and to no other fundamental cycle $\mathscr{D}_\alpha = \C_\alpha \cap (\mathcal{E} \setminus \bigcup_{\alpha' \neq \alpha} \C_{\alpha'})$. In practice, these consist of the generating chord $e_\alpha$ and of all edges that are only separated from the chord by vertices of degree $2$  (sometimes called bridges \cite{altanerPRL}). Then, we can at best distribute the dissipative parameter $x_\alpha$ among the edges of the corresponding peripheral set, in such a way that
\begin{align}
a'_{i\!j}(\{x_\alpha\}) = \left\{ \ba{ll} a''_{i\!j} + \lambda^{i\!j}_\alpha x_\alpha, & i\!j \in  \mathscr{D}_\alpha  \\
a''_{i\!j}, & i\!j \notin \bigcup_\alpha  \mathscr{D}_\alpha \label{eq:periperi}
\ea\right.
\end{align}
where $a''_{i\!j}$ are independent of all parameters and $\lambda^{i\!j}_\alpha = - \lambda^{\!ji}_\alpha$ are real numbers such that, for all $\alpha$,
\begin{align}
\sum_{i\!j \in \mathscr{D}_\alpha} \lambda^{i\!j}_\alpha = 1.  \label{eq:prescribed2}
\end{align}
This provides the general structure of a thermodynamic parametrization. Notice that no constraint is imposed on the dependence of the symmetric part and of the pure gauge on the parameters. 

\begin{theorem}[Thermodynamic parametrization]
\label{th:thermodynamic} A parametrization of the transition rates is thermodynamic, in the sense that it satisfies Eq.\,(\ref{eq:dissi1}), if and only if there exist (possibly non-unique) symmetric terms, a gauge connection and a pure gauge such that
\begin{align}
w_{i\!j}(\bs{x}) =  \sqrt{v_{i\!j}(\bs{x})} \, \exp \frac{1}{2} \left[ {a_j(\bs{x}) - a_i(\bs{x})} + a'_{i\!j}(\{x_\alpha\}) \right].
\end{align}
\end{theorem}

\begin{proof}
Sufficiency is trivial: from Eq.\,(\ref{eq:Agf}) we have
\begin{align}
\A_\alpha(\{x_{\alpha'}\}) = \sum_{i\!j \in \mathscr{D}_\alpha}
\lambda^{i\!j}_\alpha x_\alpha +  \sum_{i\!j \in \C_\alpha} a''_{i\!j} 
= x_\alpha +  \sum_{i\!j \in \C_\alpha} a''_{i\!j}.
\end{align}
Necessity is in inbuilt in the above construction: the symmetric/antisymmetric splitting is unique; perturbations of the antisymmetric part that do not affect the affinities are accounted for by the pure gauge; hence we can only focus on any explicit dependencies of the gauge connection on the parameters that will affect the affinities. In particular, any perturbation of the rates along an edge that is in common among several cycles will affect the affinities of both cycles, therefore violating thermodynamic parametrization.
\end{proof}

Notice that, as anticipated in Sec.\,\ref{strophe}, the symmetric term and the pure gauge, to which we referred as the {\it internal landscape}, can undergo global arbitrary transformations, while once the pure gauge is fixed, the parametrization of the gauge connection is local an quite strict.

We now want to inquire how much stricter is a marginally thermodynamic parametrization obeying Eq.\,(\ref{eq:dissi2}) with respect to a thermodynamic parametrization. We already know that effective affinities are gauge invariant. Yet they do depend on the kinetic parameters all over the network. Then, all we need is to make the dependency of the symmetric terms local as that of the gauge connection.

\begin{theorem}[Marginal thermodynamic parametrization]
\label{th:mardiss1}
A parametrization of the transition rates is marginally thermodynamic, in the sense that it satisfies Eq.\,(\ref{eq:dissi3}), if it is thermodynamic and, additionally, if symmetric terms only depend locally on the thermodynamic parameters, meaning that only those along $\ij \in \mathscr{D}_\mu$ depend on $x_\mu$.
\end{theorem}

\begin{proof}
All boils down to showing that the effective affinity can be written as
\begin{align}
\Q_\mu(\x) = \log \left( \frac{p^\st_{i'_\mu}(\x)}{p^\st_{j'_\mu}(\x)} \prod_{\ij \in \mathscr{D}_\mu} \frac{w_{\ij}}{w_{\ji}}\right)
\end{align}
where $i'_\mu$ and $j'_\mu$ are respectively the source of the first edge in $\mathscr{D}_\mu$ and the target of the last edge in $\mathscr{D}_\mu$. Notice that these depend in a contrived way on all other rates not belonging to $\mathscr{D}_\mu$. Having singled out the rates in $\mathscr{D}_\mu$, we notice that symmetric terms still cancel among themselves, hence we conclude.
\end{proof}
This latter proposition is not ``if and only if'' only because of subtle special cases of no practical interest.

\subsection{Disconnecting the configuration space}
\label{disconnect}

So far we have assumed  that the removal of the edges that support observable currents does not disconnect the configuration space. Let us relax this assumption.

The removal of edges from the network may result into several disconnected subgraphs $\G_\kappa$. If that occurs, than each connected subgraph has its own unique normalized steady state $\vec{p}^{\,\st ,\kappa}$ and any convex combination of them
\begin{align}
\vec{p}^{\,\st}(\bs{\pi}) = \sum_\kappa \pi_\kappa \vec{p}^{\,\st,\kappa}, \qquad \sum_\kappa \pi_\kappa = 1 \label{eq:multiss}
\end{align}
is a steady state for the stalling dynamics. Here 
$\pi_\kappa > 0$ is the probability of finding the system in subgraph $\G_\kappa$ at the moment of the preparation of the initial ensemble. As regards the case of several currents each supported on one edge analyzed in Sec.\,\ref{sec:multiedge}, the effective affinities introduced in Eq.\,(\ref{eq:effaffmu}) are still well-defined, and an inspection of the proof of \Th{th:unitary2} and of \Th{th:marls} reveals that all that matters to go through is that $\vec{p}^{\,\st}(\bs{\pi})$ is {\it some} steady state. The only main difference is that in all expressions there will be an explicit dependence on the parameters $\bs{\pi}$, in particular via $\Past(\bs{\pi})$ and the effective affinities $\Q_{\mu}(\bs{\pi})$. Consider the function $i \to \kappa(i)$ mapping a site $i$ to the connected component it belongs to. We then have
\begin{align}
\Q_\mu(\bs{\pi}) = \log \frac{w_{i_\mu j_\mu}  p_{j_\mu}^{\,\st,\kappa(j_\mu)}}{w_{j_\mu i_\mu} p_{i_\mu}^{\,\st,\kappa(i_\mu)} } + \log \frac{\pi_{\kappa(j_\mu)}}{\pi_{\kappa(i_\mu)} } =: \Q'_\mu + \log \frac{\pi_{\kappa(j_\mu)}}{\pi_{\kappa(i_\mu)} } 
\end{align} 
where the right-hand side defines a preferred ``ground'' value of the effective affinity $\Q'_\mu$. 
The finite-time FR holds unmodified, provided one chooses the right effective affinity and the right initial ensemble (as we will analyze in a future publication, and has already been observed in the case of the full FR by Rao \cite{rao}, in the preparation there is a difference between removing edges, and tuning parameters to values where the currents stall). 

As regards the asymptotic FR, notice that in the derivation of the fluctuation  symmetry the matrix $\Past(\bs{\pi})$ does not enter the game as it only contributes a similarity transformation. Hence the full dependence of the relation on the $\bs{\pi}$ is through the effective affinities:
\begin{align}
\widetilde{\lambda}(\{q_\mu\}) = \lambda(\{\Q_\mu(\bs{\pi}) - q_\mu\}).
\end{align} 
But since the left-hand side of this expression does not depend on $\bs{\pi}$ altogether, we obtain that the SCGF has a symmetry
\begin{align}
\lambda(\{q_\mu\}) = \lambda\left(\left\{q_\mu + \log \pi_{\kappa(j_\mu)}/ \pi_{\kappa(i_\mu)} \right\}\right).
\end{align}
This also implies that certain linear combinations of mean currents always vanish at the steady state. Consider for example the case where two subgraphs are separated by two observable edges:
\begin{align}
\ba{c}\xymatrix{\G_{1} \ar@{<-}@/^1pc/[r]^{\RN{1}}   \ar@{<-}@/_1pc/[r]_{\RN{2}}  & \G_{2}} \ea.
\end{align}
The symmetry of the rate function reads
\begin{align}
\lambda(\{q_{\RN{1}},q_{\RN{2}}\}) = \lambda(\{q_{\RN{1}}+ \alpha,q_{\RN{2}}+\alpha\}).
\end{align}
where $\alpha = \log \pi_{2}/ \pi_{1}$. Taking the derivative with respect to $\alpha$ and evaluating at $q_{\RN{1}} = q_{\RN{2}} = \alpha = 0$ we obtain $\cur_{\RN{1}}+\cur_{\RN{2}} = 0$.

Notice that \Th{iffstalling} then needs to be reformulated, as it is not sufficient that the effective affinity vanishes for some values of the $\bs{\pi}$ to have a vanishing stalling current. This is a consequence of the marginal gauge invariance of the SCGF, that we now analyze. 

The SCGF of all the marginal currents is obtained by setting $q_{\ij} = 0$ for all $\ij \in\E \setminus \E_{\mathrm{mar}}$. With reference to \Th{th:gaugtran}, clearly any gauge transformation that involves cocycles that have edges in the hidden sector would break this condition, thus it is not allowed. But, since by definition the removal of a (simple) cocycle disconnects the graph in two subgraphs (we can always take the basis cocycles to be simple in the sense described in \S\,\ref{par:cocycle}), then if the marginal configuration space contains a cocycle, then the hidden configuration space gets disconnected. In fact, in the above example the two edges form a simple cocycle. We thus see that marginal gauge invariance coincides with the freedom of choice of the stalling steady state in the definition of the effective affinities. While, according to the tenets of Schnakenberg's theory, when dealing with single-edge currents one can always take observable currents along a subset of chords, and thus arrive at uniquely defined and invariant affinities, in the case of currents defined over multiple edges it might often be the case that gauge invariance plays a role. This depends on the specific context.

\section{Episode 4: Phenomenological currents}
\label{sec:phenomenological}

Finally, we consider a set of marginal phenomenological currents. The theory exposed in the previous sections holds unchanged provided an additional strict condition, that we call {\it marginal consistency}, is verified. We characterize marginal consistency in physical terms in terms of hidden entropy production.

\subsection{Marginal phenomenological currents}

Phenomenological currents are linear combinations of edge currents
\begin{align}
\tcur^t_\mu = \sum_{\ij} \phys^{\ij}_\mu \, \tcur^t_{\ij}, \label{eq:phenocur}
\end{align}
where $\phys^{\ij}_\mu  = - \phys^\ji _\mu$. 
Let $\E_\mu = \{\ij \in \E | \phys^{\ij}_\mu \neq 0 \}$ be the set of edges that support the $\mu$-th marginal current, and $\E_{\mathrm{mar}} := \bigcup_\mu \E_\mu$ be the {\it marginal edge set} of all edges that support a marginal phenomenological current. Furthermore let $\I_{\mathrm{mar}}$ be set of sites of the graph that are boundaries of some edge in $\E_{\mathrm{mar}}$. For sake of notational simplicity, we assume that $\E \setminus \E_{\mathrm{mar}}$ is connected; the general case can be built after the considerations in Sec.\,\ref{disconnect}. 

We assume that the phenomenological currents defined above are independent, in the sense that there is no vector $\bs{\ell} = (\ell_\mu)_\mu$ such that
\begin{align}
\sum_\mu \ell_\mu \tcur^t_\mu \asymp 0
\end{align}
in the long time limit (at steady state). This amounts to ask that the matrix $(\sum_{\ij} \phys^{\ij}_\mu \cyc_{\ij}^\alpha)_{\alpha,\mu}$ has full rank. The theory of FRs for a complete set of phenomenological currents in the presence of conservation laws has been analyzed in Ref.\,\cite{bridging} at the steady state, and in Ref.\,\cite{rao} at finite times. The generalization of our marginal theory to systems with conserved quantities, either within or across the marginal/hidden sectors of the configuration space, is an interesting direction for future research.

\subsection{Marginal consistency}

Let $\Lambda(\{q_\mu\})$ be the SCGF of the marginal phenomenological currents with respect to the forward dynamics. It is the dominant eigenvalue of the tilted operator $\M(\{q_\mu\}$ obtained by replacing the off-diagonal entries $w_{\ij}$ by $w_{\ij}\exp - \ssum \phys^{\ij}_\mu q_\mu$ and by keeping the diagonal entries the same. Similarly we can introduce the SCGF $\widetilde{\Lambda}(\{q_\mu\})$ of the currents with respect to the hidden TR dynamics.

Also, let $\lambda(\{q_{\ij}\}_{\ij \in \E_{\mathrm{mar}}}),\widetilde{\lambda}(\{q_{\ij}\}_{\ij \in \E_{\mathrm{mar}}})$ be the SCGFs of all the edge currents in the marginal sector of the configuration space. We hereby exploit the contraction principle in the theory of large deviations, which states that we can obtain the SCGF of a coarser observable by replacing $q_{\ij} \to \ssum \phys^{\ij}_\mu q_{\mu}$,
\begin{subequations}
\begin{align}
\Lambda(\{q_\mu\}) & = \lambda\left(\Big\{\ssum \phys^{\ij}_\mu q_{\mu}
\Big\}_{\ij \in \E_{\mathrm{mar}} }\right), \\
\widetilde{\Lambda}(\{q_\mu\}) & = \widetilde{\lambda}\left(\Big\{\ssum \phys^{\ij}_\mu q_{\mu}
\Big\}_{\ij \in \E_{\mathrm{mar}} }\right).
\end{align}
\end{subequations}
By the theory exposed in Sec.\,\ref{sec:multiedge}, we know that the marginal fluctuation  symmetry holds 
\begin{align}
\lambda\left(\{q_{\ij}\}_{\ij \in \E_{\mathrm{mar}} }\right) = \widetilde{\lambda}\left(\{\Q_{\ij} - q_{\ij}\}_{\ij \in \E_{\mathrm{mar}} }\right),
\end{align}
where the effective affinities are given by the usual expression $\Q_{\ij} = \log w_{\ij}p_j^\st / w_\ji p_i^\st$, the stalling steady state being defined as the unique steady state in the network where all edges in $\E_{\mathrm{mar}}$ are removed. 

We now inquire under which conditions this relation extends to the SCGFs of the marginal phenomenological currents.

\begin{theorem}\label{th:mtc}
A sufficient condition for the marginal phenomenological fluctuation  symmetry to hold is the condition of \emph{marginal consistency}, i.e. that there exist phenomenological effective affinities $\Q_\mu$ such that
\begin{align}
\Q_{\ij} = \ssum \phys^{\ij}_\mu \, \Q_{\mu}, \qquad \forall \ij \in \E_{\mathrm{mar}} . \label{eq:mtc}
\end{align}
\end{theorem}

\begin{proof} Straightforward given the definitions and the contraction principle:
\begin{align}
\Lambda(\{q_\mu\}) =  \lambda\left(\Big\{\ssum \phys^{\ij}_\mu q_{\mu}
\Big\}_{\ij \in \E_{\mathrm{mar}} }\right) & = \widetilde{\lambda}\left(\{
\ssum \phys^{\ij}_\mu \, (\Q_{\mu} - q_{\mu})
\}_{\ij \in \E_{\mathrm{mar}} }\right) \nonumber \\
& = \widetilde{\Lambda}(\{\Q_{\mu} - q_\mu\}).
\end{align}
\end{proof}

The condition is not strictly necessary, for the reason noted in Sec.\,\ref{disconnect} that the effective affinities might not be unique when the removal of the observable edges disconnects the network, which is usually the case for phenomenological currents.

Marginal consistency poses a strong constraint on the shape of the stalling steady state. In fact, for all sites $i \in \I_{\mathrm{mar}}$ in the marginal edge set, the stalling steady state must satisfy
\begin{align}
\frac{p^\st_j}{p^\st_i} = \frac{w_\ji }{w_\ij} \exp {\ssum \phys^{\ij}_\mu \, \Q_{\mu}}. \label{eq:pjpi}
\end{align}
This balance will not be granted for any thermodynamically consistent model, as we will show by a simple example in Sec.\,\ref{consistency}. At the highest level of the hierarchy of marginal theories, where the effective affinities are the ``real'' affinities of the complete theory, the above condition boils down to prescribing certain values for the rate ratio $w_\ij/w_\ji $, up to a pure gauge, see Eq.\,(\ref{eq:EPR}). This is due to the peculiar property of equilibrium steady states of writing just in terms of $w_\ij/w_\ji $, see Eq.\,(\ref{eq:gaugess}). Instead, whilst gauge invariant, the effective affinities do not depend directly on the rate ratio. As a consequence, marginal consistency requires a fine tuning of the rates all over the network. This entails an unprecedented relationship between symmetric and antisymmetric contributions to the rates, i.e. an interplay between the physics that shapes the internal landscape of the system, and the thermodynamic parameters.
\begin{theorem}
If the condition of marginal consistency holds, then 
\begin{align}
\widetilde{\W} := \Past \, \M(\{\Q_\mu\})^T \,{\Past}^{-1}.\label{eq:unitary}
\end{align}
is a \MJPG.
\end{theorem}
\begin{proof}
We have, using Eq.\,(\ref{eq:pjpi})
\begin{subequations}
\begin{align}
\widetilde{w}_{\ij} & = \frac{p^\st_i}{p^\st_j} w_\ji  \exp {\ssum \phys^{\ij}_\mu \, \Q_{\mu}} = w_{\ij}, & & \forall \ij \in \E_{\mathrm{mar}} \\
\widetilde{w}_{\ij} & = \frac{p^\st_i}{p^\st_j} w_\ji , & & \forall \ij \notin \E_{\mathrm{mar}}.
\end{align}
\end{subequations}
The proof proceeds as usual by summing over the columns of $\widetilde{\W}$ and using the fact that $\vec{p}^{\,\st}$ is the stalling state.
\end{proof}

From this result it follows that, if the condition of marginal consistency holds, basically everything that has been discussed in Sec.\,\ref{sec:multiedge} holds without modifications.

\subsection{Internal stalling}

The above condition of marginal consistency so far remains formal. Notice that a state of {\it internal stalling}, where all the microscopic currents supporting the marginal phenomenological currents vanish, is also a state of {\it phenomenological stalling}:
\begin{align}
\cur_{\ij} = 0, ~\forall \ij \in \mathcal{E}_{\mu} \quad\Rightarrow \quad  \cur_\mu = 0.
\end{align}
The inverse implication is not generally true.

\begin{theorem}
\label{th:mtctc}
If marginal consistency holds, then at a phenomenological  stalling steady-state, internal currents stall.
\end{theorem}

\begin{proof}
By the fluctuation  symmetry, if $\Q_\mu = 0$ then $\cur_\mu =0$. From Eq.\,(\ref{eq:mtc}) it follows that if $\Q_\mu = 0$ then $\Q_{\ij} = 0$, and therefore $\cur_{\ij}=0$.
\end{proof}

While the reverse implication does not hold, for the simple reason that stalling states are a subclass of all possible marginally consistent states, the following statement gives a physical characterization of marginal consistency.

\begin{theorem}
\label{th:reverseconsistency}
If at a phenomenological stalling steady state the marginal stochastic entropy production vanishes, then the effective affinities satisfy marginal consistency.
\end{theorem}

\begin{proof}
We are requiring that
\begin{align}
\sum_{\ij \in \E_{\mathrm{mar}}} \Q_{\ij} \tcur^t_{\ij} = 0 \label{eq:vanvan}
\end{align}
vanishes for all values of of $\tcur^t_{\ij} \equiv t\cur_{\ij}$ satisfying the steady state condition $\partial \bs{\cur} = 0$ (or, equivalently, up to  $O(1)$ transient contributions). Let us impose phenomenological stalling. We have
\begin{align}
\sum_{\ij \in \E_{\mu}} \phys^{\ij}_\mu \phi_{\ij} = 0.
\end{align}
Therefore $\bs{\phi}$ is in the kernel of matrix $\phys = (\phys^{\ij}_\mu)_{\ij  \in \E_{\mu}, \mu}$. Therefore it can be an arbitrary vector in the orthogonal complement of the image of $\phys$. Since Eq.\,(\ref{eq:vanvan}) must vanish for any such vector, then $\Q_{\ij}$ must live in the image of $\phys$, that is, satisfy Eq.\,(\ref{eq:mtc}). \end{proof}

\subsection{Response}

Let us endow ourselves with a marginally thermodynamic parametrization $x_{\ij}$ of the edges $\ij \in\E_{\mathrm{mar}}$. We can consider the phenomenological currents' covariance at phenomenological stalling
\begin{align}
\ccur_{\mu\mu'}^{\mathrm{phen.\,st.}} 
& = \sum_{\ij,\ij'} \phys^{\ij}_\mu  \phys^{\ij'}_{\mu'} \ccur_{\ij,\ij'}^{\mathrm{phen.\,st.}}.
\end{align}
Notice that we would not be able to proceed further if phenomenological stalling did not imply internal stalling, because we do not generally have a FDR for nonvanishing currents. Under the assumption of marginal consistency we obtain
\begin{align}
\cur_{\mu\mu'}^\st  & = \sum_{\ij,\ij'} \phys^{\ij}_\mu  \phys^{\ij'}_{\mu'} \left(  \frac{\partial}{\partial x_{\ij}} \cur_{\ij'}^\st  + \frac{\partial}{\partial x_{\ij'}} \cur_{\ij}^\st  \right) \nonumber \\
& = \nabla_\mu \cur_{\mu'}^\st  +  \nabla_{\mu'}\cur_\mu^\st 
 \end{align}
where
\begin{align}
\nabla_\mu := \sum_{\ij} \phys^{\ij}_\mu  \frac{\partial}{\partial x_{\ij}}
\end{align}
is the directional (Lie) derivative along vector $\{\phys^{\ij}_\mu\}_{\ij}$. This means that there is a preferred set of directions where the variation of the parameters should be performed, hence that we need to sub-parametrize rates in terms of parameters $x_\mu$ in a way that $ \nabla_\mu = \partial/\partial x_\mu$ becomes a proper derivative.

\subsection{Complete vs. marginal consistency}
\label{consistency}

In this section we show by a simple example that marginal consistency is stricter ``complete'' consistency. While this example is not representative of a physical apparatus, the gap between marginal and complete thermodynamic consistency appears to us the be the most interesting open question left aside in this paper.

We take into consideration a simple model whose configuration space is depicted by the following graph with rates labeled by a reservoir index taking values in $\rho = \RN{1},\RN{2},\RN{3},\RN{4}$:
\begin{align}
\xymatrix{1  \ar@{-}@/_/[rr]_{\RN{2}}  \ar@{-}@/^/[rr]^{\RN{1}} 
 & & 2  \\ \\ 
3 \ar@{-}@/_/[uu]_{\RN{4}}   \ar@{-}@/^/[uu]^{\RN{3}}  \ar@{-}@/^/[rr]^{\RN{2}}   \ar@{-}@/_/[rr]_{\RN{1}} 
& & 4 \ar@{-}@/_/[uu]_{\RN{3}} \ar@{-}@/^/[uu]^{\RN{4}}
 }
\end{align}
where we assume that the rates satisfy
\begin{subequations}\label{eq:ldbex}
\begin{align}
w^{\RN{1}}_{12}/w^{\RN{1}}_{21} & = w^{\RN{1}}_{34}/w^{\RN{1}}_{43} \\
w^{\RN{2}}_{12}/w^{\RN{2}}_{21} & = w^{\RN{2}}_{34}/w^{\RN{2}}_{43} \label{eq:ldbex2} \\ 
w^{\RN{3}}_{13}/w^{\RN{3}}_{31} & = w^{\RN{3}}_{24}/w^{\RN{3}}_{42} \\ 
w^{\RN{4}}_{13}/w^{\RN{3}}_{31} & = w^{\RN{3}}_{24}/w^{\RN{4}}_{42}
\end{align}
\end{subequations}
The physical rationale is that rates of type $\rho$ are due to the interaction with a bath at inverse temperature $\beta^{\,\rho}$ and satisfy the condition of local detailed balance \cite{ldbmassi,ldbmaes}
\begin{align}
\frac{w_{\ij}}{w_\ji } = \exp {- \beta^{\,\rho} (\epsilon_i - \epsilon_j)}.
\end{align}
In the above model it is implied that the energy gaps $\epsilon_1 - \epsilon_2 = \epsilon_3 - \epsilon_4$ are the same. The physical currents are given by (dropping the dependency on $t$)
\begin{subequations}
\begin{align}
\tcur^{\,\rho} & = \tcur^{\,\rho}_{12} + \tcur^{\,\rho}_{34}, & & \rho = \RN{1},\RN{2}\\
\tcur^{\,\rho} & = \tcur^{\,\rho}_{13} + \tcur^{\,\rho}_{24}, & & \rho = \RN{3},\RN{4}.
\end{align}
\end{subequations}
The above parametrization Eq.\,(\ref{eq:ldbex}) grants that the entropy production along a trajectory is a functional of the phenomenological currents only. The situation is different for the effective description of an observer who, for example, only measures $\tcur^{\,\RN{1}}$. The effective affinities along the two edges supporting this current are given by
\begin{align}
\Q^{\RN{1}}_{12} & = \log \frac{w^{\RN{1}}_{12} p^{\st}_2}{w^{\RN{1}}_{21} p^{\st}_1}, & 
\Q^{\RN{1}}_{34} & = \log \frac{w^{\RN{1}}_{34} p^{\st}_4}{w^{\RN{1}}_{43} p^{\st}_3} 
\end{align}
where $\vec{p}^{\,\st}$ is the stalling steady state on the network
\begin{align}
\ba{c}\xymatrix{1  \ar@{-}@/_/[rr]_{\RN{2}} 
 & & 2  \\ \\ 
3 \ar@{-}@/_/[uu]_{\RN{4}}   \ar@{-}@/^/[uu]^{\RN{3}}  \ar@{-}@/^/[rr]^{\RN{2}} & & 4 \ar@{-}@/_/[uu]_{\RN{3}} \ar@{-}@/^/[uu]^{\RN{4}}
 }\ea.
\end{align}
Marginal consistency requires that the two effective affinities coincide, $\Q^{\RN{1}}_{12} = \Q^{\RN{1}}_{34}$, which is only the case if
\bea
\frac{p^{\st}_1}{p^{\st}_2} = \frac{p^{\st}_3}{p^{\st}_4}.
\eea
This condition is {\it not} implied by local detailed balance and it requires a fine tuning of rates all over the network. This has broad consequences on the overall flows measured in the system. In the above example, notice that in view of Eq.\,(\ref{eq:ldbex2}), marginal thermodynamic consistency implies that the two steady forces along edges $\RN{2}$ should also coincide:
\bea
\log \frac{w_{21}^{\RN{2}} p^{\st}_1}{w_{12}^{\RN{2}} p^{\st}_2} = \log \frac{w_{43}^{\RN{2}} p^{\st}_3}{w_{34}^{\RN{2}} p^{\st}_4}.
\eea
As a consequence, at the stalling steady state both transitions of type $\RN{2}$ are in same direction. This is obviously impossible, therefore the steady currents $\phi^{\RN{2}}_{12}, \phi^{\RN{2}}_{34}$ must vanish as well at stalling. The only current at a phenomenological stalling steady state is along cycles formed by edges of type $\RN{3}$ and $\RN{4}$.

\section{Stasimon: A negative result} 
\label{sec:stasimon}

Marginal currents play a role in the formulation of so-called {\it uncertainty relations} between a current's variance and the full EPR, recently formulated \cite{baratounc} and proven under quite general conditions \cite{gingrichunc,pietzonkaunc}. Remarkably, the results hold for marginal currents of any kind, though it has been argued that the bound is only strict when the current is the entropy production itself \cite{lazarescu}. Therefore it is interesting to inquire whether stricter bounds in terms of marginal measures of the entropy production, rather than the full entropy production, might hold \cite{lazarescu}.

The uncertainty relation states that for any current
\begin{align}
\frac{\langle \cur^2 \rangle}{\langle \cur \rangle^2} \geq \frac{2}{\sum \A_\alpha \langle \cur_\alpha \rangle}.
\end{align}
The bound is significant when the current is the entropy production, $\phi = \sum \A_\alpha \cur_\alpha$. Otherwise the bound performs poorly, and one might want to lower the the measure of EPR so to make the right-hand side of this inequality as high as possible.
Since $\sum \F_\alpha \langle \cur_\alpha \rangle \geq  \Q  \langle \phi \rangle$, one tempting hypothesis is that the effective EPR might also satisfy the bound:
\begin{align}
\frac{\langle \cur^2 \rangle}{\langle \cur \rangle} \stackrel{?}{\geq} \frac{2}{\Q}. \label{eq:bound}
\end{align}

The uncertainty relation derives from the following quadratic bound on the rate function \cite{pietzonkaunc}
\begin{align}
\lambda(\{q_\alpha\}) \geq \ssum q_\alpha \langle \cur_\alpha \rangle \left( \frac{\ssum q_{\alpha'} \langle \cur_{\alpha'} \rangle}{\ssum \A_{\alpha''} \langle \cur_{\alpha''} \rangle} -1 \right).
\end{align}
Choosing, for definiteness, the first current as our observable $\phi = \phi_1$, and setting $q_\alpha = \delta_{\alpha,1} q$, one obtains
\begin{align}
\lambda(q) \geq  q \langle \cur \rangle \left( \frac{q \langle\cur\rangle}{\sigma} -1 \right).
\end{align}
The bound is found by taking second derivatives and evaluating at $q=0$. In our marginal theory, the quadratic function $q \langle \cur \rangle (q/\Q - 1)$ that would yield the relation Eq.\,(\ref{eq:bound}) in a similar way does not bound $\lambda(q)$, rather it approximates it as a quadratic taking the same vanishing values at $q = 0,\Q$. In fact, not only the bound on the SCGF does not hold, but by randomly inspecting the behavior of several example systems we were able to find several cases where there is a (mild) violation of the uncertainty relation. One such case is given by the generator 
\begin{align}
\W = \left(
\begin{array}{cccc}
 -15 & 2 & 6 & 1 \\
 3 & -12 & 10 & 0 \\
 7 & 10 & -21 & 5 \\
 5 & 0 & 5 & -6 \\
\end{array}
\right).
\end{align}
Thus we rule out a reasonable hypothesis.

\section{Exode: Conclusions}
\label{exode}

Generic considerations on future perspectives were presented in Sec.\,\ref{epode}. In this section we provide a more prosaic list of partial results and of technical issues related to the derivations of the results that are either unresolved and/or that might be of some interest from a mathematical point of view.

\subsubsection{Relationship to Doob's transform}
\label{sec:doob}

The tilted operator $\M(\{q_\mu\})$ is generally not similar to a MJPG. This has implications, for example, for the efficient computation of the SCGF, which cannot be performed by a direct application of the Gillespie algorithm but requires more advanced techniques \cite{giardina,nemoto}. Nevertheless, for all values of the tilting fields $\{q_\mu\}$ one can build a MJPG $\W(\{q_\mu\})$ out of the tilted operator $\M(\{q_\mu\})$, by performing the so-called Doob transform \cite{verley}
\begin{align}
\W(\{q_\mu\}) = \R(\{q_\mu\}) \M(\{q_\mu\})^T \R^{-1}(\{q_\mu\}) - \lambda(\{q_\mu\}) \matrix{I},
\end{align}
where $\R(\{q_\mu\})$ is the diagonal matrix whose entries are the components of the right Perron-Froebenius eigenvector of $\M(\{q_\mu\})$. In a way, the typical currents according to this new dynamics reproduce the statistics of rare currents of the original dynamics, somehow realizing Onsager's regression hypothesis. Furthermore, it can be shown that the diagonal matrix $\R(\{q_\mu\})$ [resp. $\matrix{L}(\{q_\mu\})$] whose entries are the components of the right [resp. left] eigenvector of $ \M(\{q_\mu\})$ can be interpreted as the probability of a fluctuation conditioned on being at some configuration at the final [resp. initial] time \cite{companion}.

Our theory implies that the hidden TR generator $\widetilde{\W}$ is the Doob transform evaluated at $\lambda(\{\Q_\mu\}) = 0$, where the Doob's transform and the forward generator are related by a similarity transformation. Furthermore, this implies that
\begin{align}
\Past = \R(\{\Q_\mu\}),
\end{align}
which yields a different interpretation of the right eigenvector for that particular value of the tilting parameters.

A complementary perspective on tilting techniques regards Markov processes that are {\it conditioned} upon observing a certain event, as thoroughly discussed in Ref.\,\cite{chetrite}. While such conditional processes are not Markovian on their own, there exists the possibility of constructing Markovian generators that best reproduce their behavior. Then the question is open what particular rare events are typical of hidden time-reversal generators, and what kind of conditioning on Markov processes they represent.

Knowledge of the effective affinity might be implemented in algorithms for the reconstruction of the SCGF based on cloning and/or on iterative procedures, e.g. to pinpoint certain specific values of the SCGF.

\subsubsection{Direct derivations of the FDR and activity}

Let us consider the derivation of nonequilibrium FDRs carried over e.g. in Ref.\,\cite{baiesi} and compare it to our own approach at stalling steady states. We focus on the single-edge scenario detailed in Sec.\,\ref{sec:singleedge}. The trajectory p.d.f. Eq.\,(\ref{eq:density}) can be written in terms of two terms, one that is time-symmetric and one that is time-antisymmetric:
\begin{align}
\prob[\traj](x) = \exp \frac{1}{2} \left( \Theta[\traj](x) + \Sigma[\traj](x) \right) p^0_{i_0}.
\end{align}
where $\Theta$ is a suitably defined time-symmetric term that contains both the Poisson-type waiting-time distribution, and the activities (symmetrized fluxes). We made explicit the dependency on parameter $x$. We have
\begin{align}
\frac{\partial}{\partial x} \Sigma[\traj](x) & = \tcur_{12}[\traj].
\end{align}
Let's look at the response of the average steady-state current:
\begin{align}
\frac{\partial}{\partial x}\cur_{12}(x) & = \lim_{t \to \infty} \frac{1}{t} \int \mathcal{D}\traj \, \tcur_{12}[\traj] \; \frac{\partial}{\partial x} \prob[\traj](x) \\
& =  \frac{1}{2} \langle \cur_{12}^2 \rangle + \lim_{t \to \infty} \frac{1}{2t} \left\langle \tcur_{12} \frac{\partial}{\partial x} \Theta(x) \right\rangle.
\end{align}
This shows that in general to the response of a current contributes the self-correlation and the correlation of the current with a time-symmetric observable, that contributes in important ways out of equilibrium.

	Let us see how from this we can derive the FDR at equilibrium, that is, assuming there exists a value $x = x^{\mathrm{eq}}$ such that $P[\traj](x^{\mathrm{eq}}) = P[\invtraj](x^{\mathrm{eq}})$ (up to boundary terms). Key to the result is that the time symmetry is not affected by the derivative, therefore
\begin{align}
\frac{\partial  \Theta}{\partial x}[\traj] (x^{\mathrm{eq}})= \frac{\partial  \Theta}{\partial x}[\invtraj](x^{\mathrm{eq}}).
\end{align}
Then:
\begin{align}
\left\langle \tcur_{12} \frac{\partial \Theta}{\partial x} \right\rangle^{\mathrm{eq}} & = \int \mathcal{D}\traj \, \prob[\traj](x^{\mathrm{eq}})  \, \tcur_{12}[\traj] \frac{\partial  \Theta}{\partial x}[\traj] (x^{\mathrm{eq}}) \nonumber \\
& = \int \mathcal{D}\traj \, \prob[\invtraj](x^{\mathrm{eq}})  \, \tcur_{12}[\invtraj]  \frac{\partial  \Theta}{\partial x}[\invtraj] (x^{\mathrm{eq}})\nonumber \\
& = - \left\langle \tcur_{12} \frac{\partial \Theta}{\partial  x} \right\rangle^{\mathrm{eq}}
\end{align}
and therefore it must vanish. 

It would be desirable to have a similar direct proof that the active response vanishes at stalling steady-states. We could use, instead of the time-reversed trajectory, the hidden time-reverse trajectory introduced in Sec.\,\ref{sec:fr3}. However, if we go through the same passages as above the proof halts as we do not have a trajectory such that $\prob[\minvtraj]= \widetilde{\prob}[\traj]$ (not even at stalling). We believe that finding a proper way to deal with such problem would disclose a whole new set of techniques that could be used to analyze marginal systems.

Furthermore, it would be interesting to consider the large deviation functions of the joint activities and currents, which in the case of a ``complete'' theory is known analytically \cite{bertini}.



\subsubsection{Gauge invariance}

In Quantum Field Theory, Wilson loops of gauge connections are the fundamental gauge-invariant quantities. Importantly, they satisfy the reconstruction property \cite{giles,pullin}: given the Wilson loops, one can reconstruct the gauge connection up to gauge transformations. The analogue of Wilson loops in nonequilibrium thermodynamics are the ``real'' affinities, as argued in Refs.\,\cite{polettinigauge,dice}. When considering a system exchanging energy with several of heat reservoirs at different temperatures (possibly a continuum of  them), ``real'' affinities take the well-known form $\oint \delta Q / k_B T$.

The marginal theory though deals with new objects, the effective affinities, that are not defined along a single loop, are not defined only in terms of the gauge connection alone, and somehow have a ``renormalized'' character, they are ``dressed''. Yet they are gauge invariant and they constitute important observables. So, it would be interesting to ponder how this construction might go back to gauge theories, in particular as comes to non-Abelian gauge theories, to the Mandelstam identities, and to the reconstruction property.

\subsubsection{The deletion-contraction paradigm}

Some passages in our theory offer a connection to the paradigm of deletion-contraction in algebraic graph theory \cite{sokal}, that has applications to as remote areas as knot polynomials \cite{kauffman} and Feynman diagrams \cite{nakanishi,aluffi}. Deletion-contraction formulas apply in particular to spanning-tree polynomials \cite{bollobas}. We notice in passing that most results in this area regard symmetric Laplacians, while MJPGs are not necessarily symmetric. Then, the deletion-contraction formula we employed in Eq.\,(\ref{eq:delconfor}) rooted oriented spanning tree only holds if the root is one of the two contracted vertices. It is very easy to generate counterexamples for other roots. This opens up the question whether it is possible to prove more general deletion-contraction formulas for weighted oriented graphs.

\section*{Acknowledgements}

We are thankful to Bernhard Altaner for introducing the first version of the problem that inspired this work. We enjoyed discussions with Artur Wachtel and all other group members. MP also thanks the bartenders and attendees of Tramways for moral support and company, and Alexandra Elbakyan for furnishing a fundamental tool for scientific documentation. The research was supported by the National Research Fund of Luxembourg (project FNR/A11/02) and by the European Research Council, project NanoThermo (ERC-2015-CoG Agreement No. 681456).

\end{document}